\documentclass[twocolumn,prx,notitlepage,longbibliography]{revtex4-1}
\usepackage{epsfig,color,amssymb}
\usepackage{subfigure}
\usepackage{amsfonts}
\usepackage{amscd}
\usepackage{amsmath}    
\usepackage{multirow}
\usepackage{dcolumn} 
\usepackage{bm} 
\usepackage{graphicx} 
\usepackage{enumerate}
\usepackage{epsfig}
\usepackage{subfigure}
\usepackage{xcolor}
\usepackage{epstopdf}
\usepackage{multirow}  
\usepackage{ulem}
\usepackage{braket}
\usepackage{tcolorbox}
\usepackage{comment}
\usepackage{enumitem} 
\usepackage{amsthm} 
\usepackage{hyperref}
\usepackage{listings}
\renewcommand{\emph}[1]{\textit{#1}}

\newtheorem{lemma}{Lemma}
\newtheorem{corollary}{Corollary}

\newtheorem{definition}{Definition}

\newtheorem{observation}{Observation}

\begin{document}

\lstset{
	language=Matlab,                		
	numbers=left,                  			
	numberstyle=\footnotesize,      		
	stepnumber=1,                   			
	numbersep=5pt,                  		
	showspaces=false,               		
	showstringspaces=false,         		
	showtabs=false,                 			
	breaklines=true,                			
	breakatwhitespace=false,        		
	escapeinside={\%*}{*)}          		
}

\preprint{APS/123-QED}
\title{Phase-Matching Quantum Key Distribution}

\author{Xiongfeng Ma}
\email{xma@tsinghua.edu.cn}

\author{Pei Zeng}
\author{Hongyi Zhou}
\affiliation{Center for Quantum Information, Institute for Interdisciplinary Information Sciences, Tsinghua University, Beijing 100084, China}
\date{}


\begin{abstract}
Quantum key distribution allows remote parties to generate information-theoretic secure keys. The bottleneck throttling its real-life applications lies in the limited communication distance and key generation speed, due to the fact that the information carrier can be easily lost in the channel. For all the current implementations, the key rate is bounded by the channel transmission probability $\eta$. Rather surprisingly, by matching the phases of two coherent states and encoding the key information into the common phase, this linear key-rate constraint can be overcome---the secure key rate scales with the square root of the transmission probability, $O(\sqrt{\eta})$, as proposed in twin-field quantum key distribution [Nature (London) 557, 400 (2018)]. To achieve this, we develop an optical-mode-based security proof that is different from the conventional qubit-based security proofs. Furthermore, the proposed scheme is measurement device independent, i.e., it is immune to all possible detection attacks. The simulation result shows that the key rate can even exceed the transmission probability $\eta$ between two communication parties. In addition, we apply phase postcompensation to devise a practical version of the scheme without phase locking, which makes the proposed scheme feasible with the current technology. This means that quantum key distribution can enjoy both sides of the world---practicality and security.
\end{abstract}

\maketitle

\section{Introduction}
Quantum key distribution (QKD)  \cite{Bennett1984Quantum,Ekert1991Quantum} is the most successful application in quantum information science, whose security was proved at the end of the last century \cite{Mayers2001Unconditional,Lo1999Unconditional,Shor2000Simple}. Since then, there has been a tremendous interest in developing this quantum technology for real-life applications, starting from the first 32-cm demonstration in the early 1990s \cite{bennett1992experimental} to the recent satellite QKD over 1200 km \cite{liao2017satellite}. In these implementations, photons are used as information carriers, owing to their fast transmission speed and robustness against decoherence from the environment. Also, optical quantum communication can be easily integrated with the current telecommunication network infrastructure.

Now, the transmission loss of photons has become a major obstacle in practical implementations. The quantum channel transmission efficiency is characterized by the transmittance $\eta$, defined as the probability of a photon being successfully transmitted through the channel and being detected. For most of the current implemented schemes, such as the well-known Bennet-Brassard 1984 (BB84) protocol \cite{Bennett1984Quantum}, single-photon sources\footnote{In practice, single-photon sources are often replaced with weak coherent state sources or heralded single-photon sources. Nevertheless, only the single-photon components are used for secure key distribution.} are employed for key information encoding. Since the photon carries the quantum information, when it is lost in the channel, no secure key can be distributed. Thus, the transmittance $\eta$ becomes a natural upper bound of the key generation rate. A more strict derivation shows a linear key-rate bound with respect to the transmittance \cite{Curty2004Entanglement,takeoka2014fundamental,Pirandola2017Fundamental}, $R\le O(\eta)$. Since the transmittance $\eta$ decays exponentially with the communication distance in the fiber-based network, this linear key-rate bound severely limits the key generation rate.



The following two approaches to overcome this rate limit have been considered: quantum repeaters \cite{EntSwap1993,Briegel1998Repeater,azuma2015all} and trusted relays. Unfortunately, using quantum repeater schemes with current technology is infeasible because they require high-quality quantum memory and complicated local entanglement distillation operations. The trusted-relay approach, however, relies on the assumption that the quantum relays between two users are trustworthy, which is difficult to ensure or verify practically; this severely undermines the primary goal of QKD, i.e., security. In 2012, the measurement-device-independent quantum key distribution (MDI-QKD) scheme was proposed to close all the detection loopholes \cite{Lo2012Measurement}, which enhances the security of a practical QKD system. Nevertheless, the key rate of the MDI-QKD scheme is still bounded by $O(\eta)$. Therefore, the linear key-rate bound  \cite{takeoka2014fundamental,Pirandola2017Fundamental} was widely believed to hold for practical QKD systems without relays.

Significant efforts have been devoted to improve the key rate by proposing different schemes. Recently, Lucamarini\textit{~et~al.} proposed a novel phase-encoding QKD protocol, called twin-field quantum key distribution (TF-QKD) \cite{Lucamarini2018TF}, which shows the possibility to overcome the key-rate limit and make a quadratic improvement over phase-encoding MDI-QKD \cite{Tamaki2012PhaseMDI}. In both schemes, single-photon detection is used, whereas coincident detection is required in other MDI-QKD schemes \cite{Lo2012Measurement,Ma2012Alternative}. From technical point of view, the single-photon detection is the key reason for the quadratic improvement. Unfortunately, a rigorous security proof is still missing at the moment. In fact, as shown later, the widely used photon number channel model \cite{Ma2008PhD} used in the security proof of MDI-QKD is proven to be invalid for this kind of setting.

Following the TF-QKD scheme \cite{Lucamarini2018TF}, we investigate phase-encoding MDI-QKD schemes with single detection \cite{Tamaki2012PhaseMDI} and propose a phase-matching quantum key distribution (PM-QKD) scheme that can surpass the linear key-rate bound, inspired by the relative-phase-encoding Bennett-1992 \cite{Bennett1992Quantum}, phase-encoding MDI-QKD \cite{Tamaki2012PhaseMDI,Ma2012Alternative}, and passive differential-phase-shift  QKD \cite{Guan2015RRPDS}. The two communication parties prepare two coherent states independently, encode the key information onto the common phase, and match phases via interference detection at an untrusted measurement site. Details are given in Sec. \ref{Sc:protocol}. By developing an optical-mode-based security proof, we show that the key rate of the proposed scheme scales with the square root of the transmittance, $R=O(\sqrt{\eta})$, in Sec. \ref{Sc:security} with technical details given in Appendix \ref{Sc:SecureProof}. Also, the proposed phase-matching scheme falls into the MDI framework, which is immune to all possible detection attacks. Our security proof can be directly applied to TF-QKD. In Sec. \ref{Sc:practical}, we deal with related practical issues and develop a phase postcompensation technique to ease the experimental requirements. In Sec. \ref{Sc:simulation}, we simulate the key-rate performance of PM-QKD and compare it to former QKD protocols, with all the practical factors taken into account. Finally, in Sec. \ref{Sc:outlook}, we discuss possible future work directions.





\section{PM-QKD protocol}\label{Sc:protocol}
In PM-QKD, the two communication parties, Alice and Bob, generate coherent state pulses independently. For a $d$-phase PM-QKD protocol, Alice and Bob encode their key information $\kappa_a, \kappa_b \in \{0, 1, \dots, d-1\}$, into the phases of the coherent states, respectively, and send them to an untrusted measurement site that could be controlled by Eve, as shown in Fig.~\ref{fig:PM}(a). Eve is expected to perform interference detection. Define a successful detection as the case where one and only one of the two detectors clicks, denoted by $L$ click and $R$ click. This interference measurement would match the phases of Alice and Bob's signals. Conditioned on Eve's announcement, Alice and Bob's key information is correlated.

In this work, we mainly focus on PM-QKD with $d=2$ and phase randomization. That is, Alice and Bob add extra random phases on their coherent state pulses before sending these pulses to Eve. After Eve's announcement, Alice and Bob announce the extra random phases and postselect the signals based on the random phases. This PM-QKD scheme is detailed below and shown in Fig.~\ref{fig:PM}(b). For simplicity, by using the name ``PM-QKD'' in the text below, we refer to the case of $d=2$ plus phase randomization.



\begin{itemize}[label={}]
\item
\textbf{State preparation}: Alice randomly generates a key bit $\kappa_a$ and a random phase $\phi_a\in[0,2\pi)$ and then prepares the coherent state $\ket{\sqrt{\mu_a}e^{i(\phi_a+\pi\kappa_a)}}_{A}$. Similarly, Bob generates $\kappa_b$ and $\phi_b\in[0,2\pi)$ and then prepares $\ket{\sqrt{\mu_b}e^{i(\phi_b+\pi\kappa_b)}}_{B}$.

\item
\textbf{Measurement}: Alice and Bob send their optical pulses, $A$ and $B$, to an untrusted party, Eve, who is expected to perform an interference measurement and record the detector ($L$ or $R$) that clicks.

\item
\textbf{Announcement}: Eve announces her detection results. Then, Alice and Bob announce the random phases $\phi_a$ and $\phi_b$, respectively.

\item
\textbf{Sifting}: Alice and Bob repeat the above steps many times. When Eve announces a successful detection, (a click from exactly one of the detectors $L$ or $R$), Alice and Bob keep $\kappa_a$ and $\kappa_b$ as their raw key bits. Bob flips his key bit $\kappa_b$ if Eve's announcement was an $R$ click. Then, Alice and Bob keep their raw key bit only if $|\phi_a-\phi_b| = 0$ or $\pi$. Bob flips his key bit $\kappa_b$ if $|\phi_a-\phi_b| = \pi$.

\item
\textbf{Parameter estimation}: For all the raw data that they have retained, Alice and Bob analyze the gains $Q_{\mu}$ and quantum bit error rates $E^Z_{\mu}$. They then estimate $E^{X}_{\mu}$ using Eq.~\eqref{eq:Emu}.

\item
\textbf{Key distillation}: Alice and Bob perform error correction and privacy amplification on the sifted key bits to generate a private key.

\item
Notations: Denote a coherent state in mode $A$ to be $\ket{\sqrt{\mu}e^{i\phi}}_A$, where $\mu$ is the intensity and $\phi$ is the phase; $\mu_a = \mu_b = \mu/2$; Alice's (Bob's) key bit $\kappa_{a(b)}\in\{0,1\}$; total gain $Q_\mu$; phase error rate $E_\mu^X$; and bit error rate $E_\mu^Z$.
\end{itemize}

\begin{figure}[htbp]
\centering
\resizebox{8cm}{!}{\includegraphics{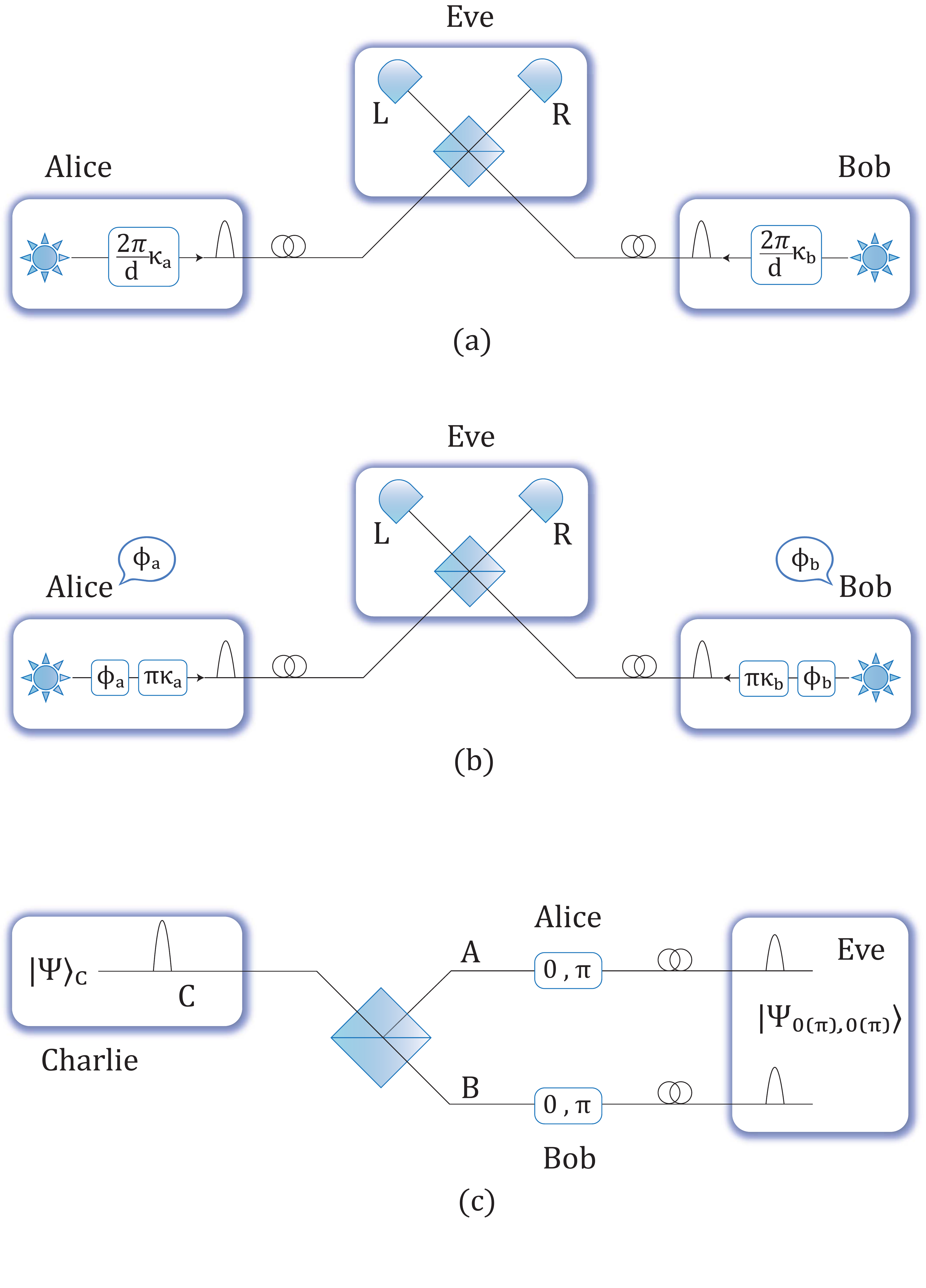}}
\caption{
(a) Schematic diagram of PM-QKD. Alice generates a coherent state, $\ket{\sqrt{\mu_a}e^{2\pi i\kappa_a/d}}$, where $\kappa_a\in \{0, 1, \dots, d-1\}$. Similarly, Bob generates $\ket{\sqrt{\mu_b}e^{2\pi i\kappa_b/d}}$. Alice and Bob send the two coherent states which interfere at an untrusted measurement site. (b) Schematic diagram of PM-QKD with $d=2$ plus phase randomization. Alice prepares $\ket{\sqrt{\mu_a}e^{i(\phi_a+\pi\kappa_a)}}$ and Bob prepares $\ket{\sqrt{\mu_b}e^{i(\phi_b+\pi\kappa_b)}}$. The two coherent states interfere at an untrusted measurement site. If the phase difference $|(\phi_a+\pi\kappa_a)- (\phi_b+\pi\kappa_b)|$ is $0$, detector $L$ clicks; if the phase difference is $\pi$, detector $R$ clicks. After Eve announces her measurement result, Alice and Bob publicly announce $\phi_a$ and $\phi_b$. (c) Equivalent scenario for the postselected signals with $\phi_a = \phi_b$. A trusted party (Charlie) prepares $\ket{\Psi}_{C}$, splits it and sends it to both Alice and Bob. Without loss of generality, we consider the case where Alice and Bob both modulate this by the same phase $0$ or $\pi$ to create the systems $A$ and $B$. If $\ket{\Psi}_{C}$ only contains odd- or even-photon number components, we can see that $\ket{\Psi_0} = \ket{\Psi_\pi}$.
}
\label{fig:PM}
\end{figure}


The above implementation of the PM-QKD protocol clearly resembles the  phase-encoding MDI-QKD protocol \cite{Ma2012Alternative,Tamaki2012PhaseMDI}, where the key bits are encoded in the relative phase of two coherent pulses (the reference and signal pulses).
However, in the PM-QKD protocol, the reference pulse can be regarded as being shared by Alice and Bob. Therefore, they no longer need to send the reference pulse, and the key becomes the global phase of the coherent signal pulses. Another significant difference between PM-QKD and the former phase-encoding MDI-QKD scheme is that no basis switching is required. In this respect, it resembles the Bennett-1992 \cite{Bennett1992Quantum} and passive DPS \cite{Guan2015RRPDS} QKD protocols. Note that, a similar proposal named ``MDI-B92'' protocol has been proposed by Ferenczi in 2013 \cite{Ferenczi2013}. With the decoy-state method, the quantum part of PM-QKD would be similar to that of the TF-QKD protocol without basis sifting.


\section{Security of PM-QKD} \label{Sc:security}
To provide an intuitive understanding of the manner in which PM-QKD works, we demonstrate its security by considering an equivalent scenario shown in Fig.~\ref{fig:PM}(c). Here, a trusted party (Charlie) prepares a pure state $\ket{\Psi}_{C}$, splits it using a $50$--$50$ beam splitter, and sends it to Alice and Bob separately. Alice and Bob encode their key information $\kappa_a$ and $\kappa_b$ into systems $A$ and $B$ by modulating the phases, and then they send these to Eve who is supposed to tell whether $|\kappa_a-\kappa_b|=0$ or 1. Thus, the four possible output states that could be sent to Eve can be expressed as $\ket{\Psi_{0, 0}}$, $\ket{\Psi_{0,\pi}}$, $\ket{\Psi_{\pi, 0}}$, and $\ket{\Psi_{\pi, \pi}}$.

Without the loss of generality, we consider the following case in which both encoded key bits are the same, $\kappa_a = \kappa_b$. Eve attempts to learn the key bit $\kappa_{a(b)}$ from the state sent to her, which is either $\ket{\Psi_{0, 0}}$ or $\ket{\Psi_{\pi, \pi}}$. Here, the phase, controlled by $\kappa_a$ and $\kappa_b$ and modulated into $A$ and $B$, has become the ``global phase'' of the combined system $A$ and $B$. If $\ket{\Psi}_C$ is a Fock state $\ket{k}_C$ with $k$ photons, then $\ket{\Psi_{0,0}} = \ket{\Psi_{\pi,\pi}}$, since the global phases of Fock states are meaningless. In this case, Eve cannot tell in principle whether the modulated phases are $0$ or $\pi$ and can only learn that $\kappa_a = \kappa_b$.

In our PM-QKD protocol, both Alice and Bob transmit weak coherent pulses, $\ket{\sqrt{\mu_{a}}e^{i(\phi_a+\pi\kappa_a)}}_{A}$ and $\ket{\sqrt{\mu_{b}}e^{i(\phi_b+\pi\kappa_b)}}_{B}$ to Eve. The phase sifting condition $\phi_a = \phi_b =\phi$ is equivalent to imagining that Charlie employs a source state of $\ket{\Psi}_{C} = \ket{\sqrt{\mu} e^{i\phi}}_{C}$ in Fig.~\ref{fig:PM}(b). For a phase-randomized state $\ket{\sqrt{\mu} e^{i\phi}}_{C}$, it is equivalent for Charlie to prepare a Fock state $\ket{\Psi}_{C} = \ket{k}_{C}$ with a probability of $P(k) = e^{-\mu} \mu^k/k!$. Thus, the PM-QKD protocol is secure if Eve cannot learn the phase $\phi$. However, in the real PM-QKD protocol, the phase $\phi$ will eventually be announced during the sifting process. When this happens, Charlie's source $\ket{\Psi}_C$ can no longer be regarded as combinations of different photon-number states $\ket{k}_C$. The key challenge of the security proof of PM-QKD lies in the fact that the quantum source cannot be regarded as a mixture of photon number states, after Alice and Bob announce the phases, $\phi_a$ and $\phi_b$. That is, the photon number channel model \cite{Ma2008PhD} and the ``tagging'' method used in the security proof by Gottesman et al.~\cite{gottesman04} (we will refer it as GLLP security proof) can no longer be applied. In Appendix~\ref{sc:attack}, a beam-splitting attack is proposed to show that the GLLP formula is incorrect after the phase announcement. Therefore, one cannot simply reduce a randomized-phase coherent state protocol to a single-photon-based protocol.

Our security proof of PM-QKD is based on analyzing the distillable entanglement of its equivalent entanglement-based protocol. Following the Shor-Preskill security argument \cite{Shor2000Simple}, the key rate of PM-QKD protocol (for the sifted signals) is given by
\begin{equation} \label{eq:keyrate}
\begin{aligned}
r_{PM} &\ge 1 - H(E_{\mu}^Z) - H(E_{\mu}^{X}), \\
\end{aligned}
\end{equation}
where $E^Z_\mu$ is the quantum bit error rate (QBER) that can be directly estimated in the experiment; $E^X_\mu$ is the phase error rate, which reflects the information leakage; and $H(x)=-x\log_2x-(1-x)\log_2(1-x)$ is the binary Shannon entropy function. We demonstrate in Appendix \ref{Sc:SecureProof} that $E^X_\mu$ can be bounded by
\begin{equation} \label{eq:Emu}
\begin{aligned}
E_{\mu}^X &\le e^Z_0q_0 +\sum_{k=0}^\infty e_{2k+1}^Z q_{2k+1} + (1 - q_0 - \sum_{k=0}^\infty q_{2k+1}). \\
\end{aligned}
\end{equation}
Here, $q_k$ is the estimated ratio of the ``$k$-photon signal'' to the full detected signal:
\begin{equation} \label{eq:qk}
\begin{aligned}
q_k &=  \frac{(e^{-\mu} \mu^k/k!)Y_k }{Q_\mu}, \\
\end{aligned}
\end{equation}
where $Q_\mu$ is the total gain of the pulses, and $Y_k$ and $e^Z_k$ are the yield and bit error rate, respectively, if Charlie's light source is a $k$-photon number state. Alice and Bob can estimate the yield and bit error rate via the decoy-state method \cite{Hwang2003Decoy,Lo2005Decoy,Wang2005Decoy}. Note that the parameters $Y_k$ and $ e^Z_k$ can still be used to characterize Eve's behavior even though the source is not actually a combination of photon-number states.

Unlike most of the existing security analysis of discrete-variable QKD, our analysis is not single-qubit based. For a long time, the sources in QKD implementations, such as weak coherent sources and spontaneous parametric down-conversion sources, have been fabricated as an approximation of single qubit, following the BB84 protocol \cite{Bennett1984Quantum}. Here, we show the security of PM-QKD with a coherent light source by directly applying the Lo-Chau entanglement distillation argument \cite{Lo1999Unconditional} on analyzing the optical modes. This technique could be helpful for both a new QKD scheme design and security analysis.

In the equivalent scenario considered above, shown in Fig.~\ref{fig:PM}(b), a trusted party Charlie is introduced. We need to emphasize that the virtual Charlie will be removed in the real implementation in Sec. \ref{Sc:practical}. If Charlie does exist, Eve may inject some probes after Charlie's outputs, and then she measures them at the output of Alice and Bob to learn their operations. This is the main problem of detection-device-independent QKD \cite{lim2014detector,gonzalez2015quantum}. In the PM-QKD protocol shown in Fig.~\ref{fig:PM}(a), Alice and Bob can simply isolate their light source and modulators in an optical circulator to prevent such Trojan-horse-like attacks. Hence, the PM-QKD scheme, like other MDI-QKD schemes, is secure against Trojan-horse-like attacks.

\section{Practical implementation}\label{Sc:practical}

Now we address a few practical issues. In the protocol shown in Fig.~\ref{fig:PM}, Alice and Bob only retain their signals when their announced phases, $\phi_a$ and $\phi_b$ are either exactly the same or with a $\pi$ difference. However, since the announced phases are continuous, the successful sifting probability tends to zero. Moreover, we assume that Alice's and Bob's laser sources are perfectly locked, such that their phase references meet, but it is very challenging in practice to achieve such phase locking.

To address these practical issues, we employ a phase postcompensation method \cite{Ma2012Alternative}, where Alice and Bob first divide the phase interval $[0,2\pi)$ into $M$ slices $\{\Delta_j\}$ for $0\le j \le M-1$, where $\Delta_j = [2\pi j/M, 2\pi (j+1)/M)$. Instead of comparing the exact phases, Alice and Bob only compare the slice indexes. This makes the phase-sifting step practical, but introduces an intrinsic misalignment error. Also, Alice and Bob do not perform the phase-sifting immediately in each round, and instead, they do it in data postprocessing. In the parameter estimation step, they perform the following procedures, as shown in Fig.~\ref{fig:phasesel}.

\begin{enumerate}
\item
For each bit, Alice announces the phase slice index $j_a$ and randomly samples a certain amount of key bits and announces them for QBER testing.

\item
In the phase postcompensation method, given an offset compensation $j_d\in \{0,1,...,M/2 -1\}$, Bob sifts the sampled bits with the phase sifting condition $|j_b + j_d - j_a| \mod M = 0$ or $M/2$. For the case of $M/2$, Bob flips the key bit $\kappa_b$. After sifting, Bob calculates the QBER $E^Z$ with Alice's sampling key bits. Bob tries all possible $j_d\in\{0,1,\cdots,M-1\}$, and figures out the proper $j_d$ to minimize the sampling QBER. Using the phase sifting condition with the proper $j_d$, Bob sifts (and flips if needed) the unsampled bits and announces the locations to Alice. Alice sifts her key bits accordingly.

\item
Alice and Bob analyze the overall gain $Q_{\mu_i}$ and QBER $E^Z_{\mu_i}$ for different values of intensities $\mu_a =\mu_b =\mu_i/2$. They estimate the phase error rate $E^{(X)}_{\mu}$ by Eq.~\eqref{eq:Emu}.
\end{enumerate}

Here, in the phase postcompensation step, Bob does not need to fix a $j_d$ for the whole experiment. Bob can group the raw data into data blocks, and then he is able to adjust the offset $j_d$ for different data blocks. Alternatively, Bob can also adjust $j_d$ in real time based on a prediction via data-fitting samples nearby. A detailed description of the operations and security arguments is presented in Appendix \ref{Sc:SecurePost}. We emphasize that the fluctuation of phases will only introduce additional bit errors, but not affect the security.


\begin{figure}[htbp]
\centering
\includegraphics[width=8cm]{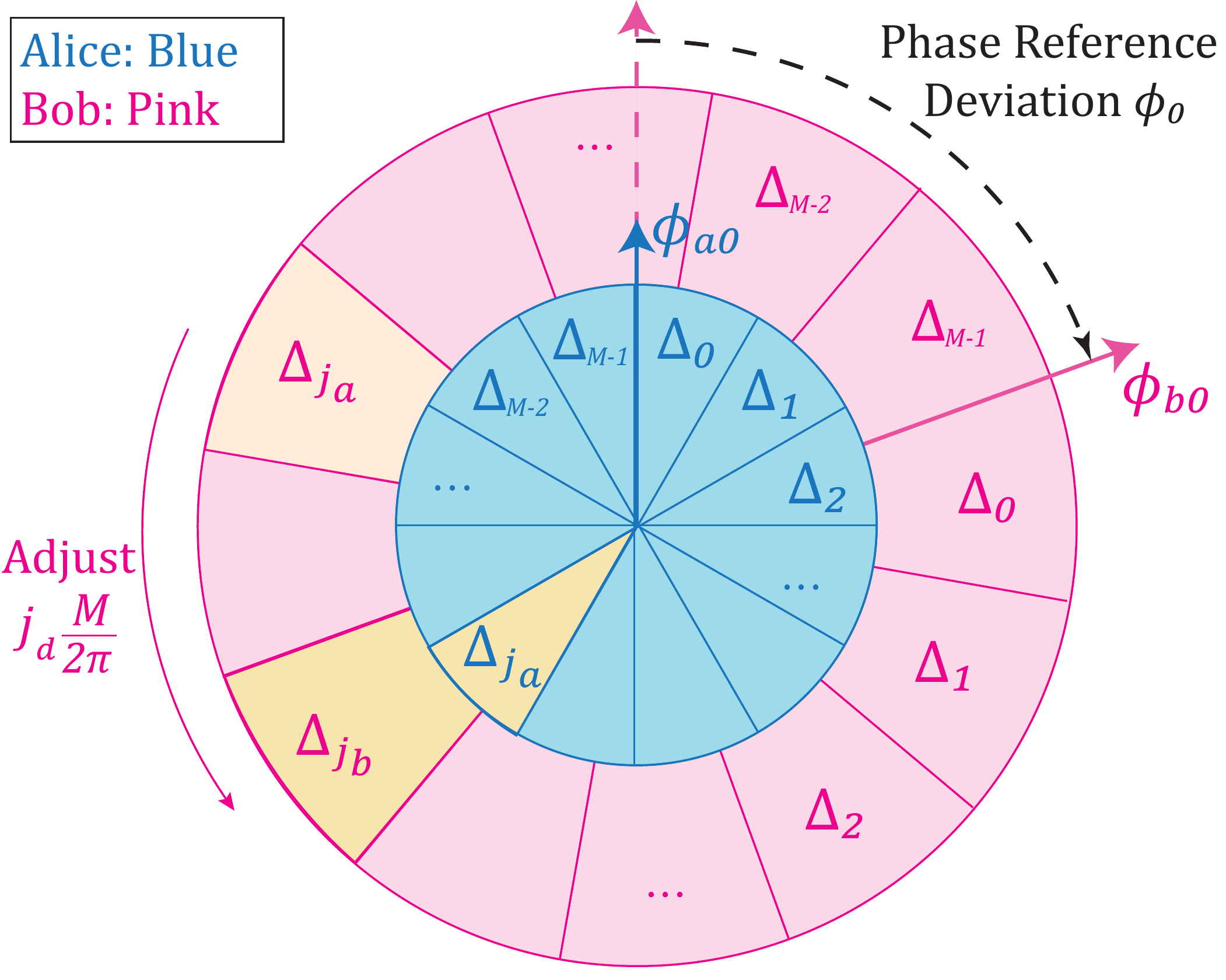}
\caption{Phase postcompensation. Without loss of generality, here we consider the case $\kappa_a = \kappa_b$. Denote the phase references of Alice and Bob as $\phi_{a0}$ and $\phi_{b0}$, and hence the reference deviation is $\phi_0=\phi_{b0}-\phi_{a0} \mod 2\pi$. Bob can figure out the proper phase compensation offset $j_d$ by minimizing the QBER from random sampling as follows. Bob sets up a $j_d$, sifts the bits by $|j_b + j_d - j_a| = 0$, and evaluates the sample QBER. He tries all possible $j_d\in\{0,1,\cdots,M-1\}$ and figures out the proper $j_d$ to minimize the sample QBER. Then, he announces the sifted locations of unsampled bits to Alice. As shown in the figure, we set $M=12$, and the reference deviation $\phi_0= 70^\circ$; hence Bob can set $j_d =2$ to compensate the effect of $\phi_0$.
} \label{fig:phasesel}
\end{figure}

An alternative method to phase postcompensation is that, Alice and Bob can generate strong pulses independent of quantum signals for phase calibration. The two pulses interfere at the measurement site so that Alice and Bob can identify the phase fluctuations between two channels. Note that the strong pulses can be set in an optical mode slightly deviated from the quantum signals to reduce cross-talks. According to the phase difference measured from the strong calibration pulses, Alice and Bob can estimate the offset $j_d$ accurately. This is a feasible way to replace the phase postcompensation method. The main difference of this phase calibration method from the usual phase-locking method \cite{santarelli1994heterodyne} is that, no active feedback is required. Alice and Bob learn the phase difference only for phase sifting in data postprocessing.

Note that similar ideas of phase postcompensation and phase calibration have already been adopted in some continuous-variable QKD protocols, such as the Gaussian-modulated coherent-state protocol \cite{Qi2007experimental,Qi2015generating} and the self-referenced protocol \cite{Soh2015self}.

Another important issue is how Alice and Bob set random phases $\phi_a, \phi_b$. In practice, it would be experimentally challenging to continuously set an accurate phase to a coherent state. As discussed above, they only need to set the slice indexes $j_a, j_b$ instead of exactly modulating phases. There are two methods to achieve this.
\begin{itemize}
\item
Alice and Bob first generate strong laser pulses with (unknown) randomized phases, either by turning on/off the lasers or active phase randomization. They then split the pulses and apply homodyne detection on one beam to measure the phase $\phi_a, \phi_b$ accurately enough to determine the slice indexes $j_a, j_b$. They use the rest beam for further quantum encoding.

\item
Alice and Bob first generate pulses with stable phases. Then, they actively randomize the phase using phase modulators. Alice and Bob can record the corresponding random numbers as for slice indexes $j_a, j_b$. Note that the phase randomization can be discrete, which has been proven to be secure and efficient with a few discrete phases \cite{zhu2015discrete}.
\end{itemize}

In order to optimize the estimated phase slice shift $j_d$ for minimizing the bit error rate, the phase compensation or calibration should be resettled with respect to phase drift. In practice, there are two major factors which may cause phase drift. One is the laser linewidth $\Delta \nu$, which causes a dispersion effect on the output pulse. The phase varies randomly with respect to the coherent time $T_{coh} \approx (\Delta \nu)^{-1}$. To alleviate the dispersion, a CW laser source with a long coherence time should be employed. The other major factor for phase drift is the variation of optical path length $\Delta L$. In a recent work of TF-QKD \cite{Lucamarini2018TF}, Lucamarini\text{~et~al.} experimentally tested the phase drift in a MDI setting. The results show that the phase drift follows a Gaussian distribution with zero mean and a standard deviation of about $6.0$ rad ms$^{-1}$ for a total distance of $550$ km. To enhance the performance of PM-QKD protocol, former works on phase stabilization of optical fibers \cite{Droste2013optical,Carvacho2015postselection,Lipka2017optical} can be employed.

With all the practical factors taken into account, the final key-rate formula can be expressed as
\begin{equation} \label{eq:keyRate}
\begin{aligned}
R_{PM} &\ge \dfrac{2}{M}Q_\mu [ - f H(E_{\mu}) + 1 - H(E_{\mu}^X)], \\
\end{aligned}
\end{equation}
where the phase error rate $E_{\mu}^X$ is given by Eq.~\eqref{eq:Emu}, $2/M$ is the sifting factor and $f$ is the error correction efficiency.

\section{Simulation results} \label{Sc:simulation}
We simulate the performance of PM-QKD with the parameters given in Fig.~\ref{fig:simulation}(b), assuming a lossy channel that is symmetrical for Alice and Bob. The dark count rate $p_d$ is from Ref. \cite{Tang2014MDI200}, and the other parameters are set to be typical values. The simulation formulas for $Q_\mu, E^Z_\mu$ and $E^X_\mu$ of PM-QKD are given in Eqs.~\eqref{eq:Qmu}, ~\eqref{eq:Emusimu}, and ~\eqref{eq:EmuXzoom}, respectively, in Appendix \ref{Sc:GainQBER}. The simulation formulas for BB84 and MDI-QKD are listed in Appendix \ref{Sc:Others}.

The simulation results are shown in Fig.~\ref{fig:simulation}(a). From the figure, one can see that PM-QKD is able to exceed the linear key-rate bound when $l> 250$ km with practical settings such as dark counts, misalignment errors, and sifting factors. Compared with MDI-QKD, PM-QKD can achieve a longer transmission distance of $l=450$ km and the key rate is increased by approximatedly $4\sim6$ orders of magnitude when $l>300$ km. Moreover, if we set up a cutoff line of key rate $R = 10^{-8}$ as real-life consideration, then the longest practical transmission distance of PM-QKD is over $400$ km, whereas the ones of BB84 and MDI-QKD are all lower than $250$ km.

Several QKD schemes are compared in Fig.~\ref{fig:simulation}(c). The comparison shows that the PM-QKD scheme outperforms the existing protocols in the following aspects. First, the PM-QKD scheme has a quadratic improvement of key rate, $O(\sqrt{\eta})$. In the security aspect, PM-QKD enjoys the measurement-device-independent nature that is immune to all detection attacks. In the practical aspect, it removes the requirement of the basis switch, which can simplify the experiment apparatus and reduce the randomness consumption. 

Here, we would like to clarify that the linear key-rate bound by Pirandola et al.~\cite{Pirandola2017Fundamental} is derived for point-to-point QKD protocols. In the PM-QKD or other MDI-QKD schemes, there is an untrusted relay held by Eve. The quadratic improvement in PM-QKD seems unsurprising if we regard the untrusted middle node as a quantum repeater.

\begin{figure*}
\centering
\includegraphics[width=16cm]{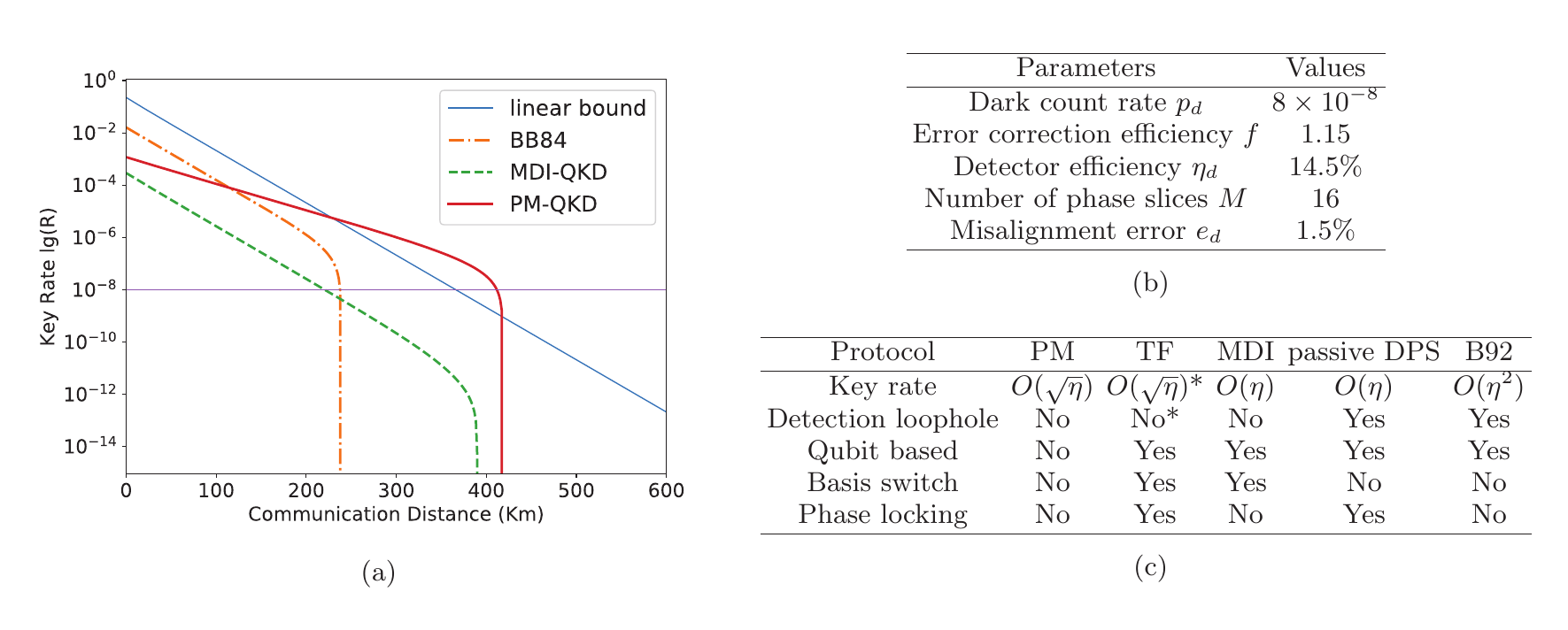}
\caption{(a) Simulation of our PM-QKD protocol. For the considered simulation parameters, the key rate of PM-QKD surpasses that of the conventional BB84 protocol when $l>120$ km and it exceeds the linear key-rate bound by Pirandola et al.~\cite{Pirandola2017Fundamental} when $l>250$ km. In addition, our protocol is also able to achieve a long transmission distance of $l=418$ km. (b) Parameters used for simulation. (c) Comparison of different QKD protocols: PM-QKD; TF-QKD \cite{Lucamarini2018TF}; MDI-QKD \cite{Lo2012Measurement}; passive differential phase-shift (DPS) QKD \cite{Guan2015RRPDS}; Bennett-1992 (B92) QKD \cite{Bennett1992Quantum}. Key rate: dependence of the key rate on the channel transmittance $\eta$; Detection loophole: whether the protocol is immune to all detection loopholes; Qubit based: whether the security analysis is based on the single-qubit case; Basis switch: whether the source (measurement device) must prepare (measure) states in complementary bases;  Phase locking: whether the protocol requires a fixed phase reference frame for the two users. *Our security proof and performance analysis apply to TF-QKD if the basis information $X$, $Y$ is ignored.}
\label{fig:simulation}
\end{figure*}

Note that the security of recently proposed TF-QKD protocol \cite{Lucamarini2018TF} can be reduced to the security of PM-QKD protocol if the information for the two bases $X$, $Y$ is ignored in TF-QKD It remains an open problem whether we can reach a higher key rate by taking advantage of the basis information together.

\section{Outlook} \label{Sc:outlook}

There are a few interesting directions on PM-QKD. First, it is interesting to work out the security of the general $d$-phase PM-QKD protocol shown in Fig.~\ref{fig:PM}(a) with and without phase randomization. Meanwhile, in the above discussions, PM-QKD in Fig.~\ref{fig:PM}(b) is treated as a single-basis scheme with phase randomization. In a dual viewpoint, we can regard the different global phases $\phi_a, \phi_b$ as different bases and naturally treat the phase-sifting step as basis-sifting. This is interesting, since it shows the advantage of QKD using multi-nonorthogonal bases. Moreover, this multibases view may help us to generalize the phase-matching scheme to, for example, a polarization-based one.

Second, the phase-sifting factor $2/M$ is very small, which undermines the advantage of PM-QKD for near-distance (i.e., $l<120$ km) communication. One possible solution is to apply a biased phase randomization; i.e., the $\phi_{a(b)}$ is not uniformly randomized in $[0, 2\pi)$.

Third, we bound the total phase error $E^X_\mu$ by pessimistically considering the phase errors $e^X_k$ for even photon number components to be $1$ in Eq.~\eqref{eqn:EX Fock}, as shown in Appendix \ref{Sc:decoy}. There may be a scope to improve the bound of total phase error $E^X_\mu$. For example, the two-photon error $e_2^Z$ and $e_2^X$ can be better estimated with more decoy states, which leads to a tighter bounds of $E^X_\mu$.

Finally, since we can overcome the key-rate linear bound, it is interesting to investigate a repeaterless secret key capacity bound for QKD, for example, whether it is possible to push the key rate to $R = O(\eta^{1/3})$ or $O(\eta^{1/4})$ without repeaters. Note that the key-rate bound has been derived for the single-repeater case \cite{Pirandola2016Capacities}, $-\log(1-\sqrt{\eta})$, which is close to $\sqrt{\eta}$ when $\eta$ is small. So far, our key rate, Eq.~\eqref{eq:keyRate}, is still far away from this bound. It is an interesting direction to improve the PM-QKD protocol to approach this bound.

\acknowledgments
This work was supported by the National Natural Science Foundation of China Grant No.~11674193 and the National Key R\&D Program of China Grants No. 2017YFA0303900 and No. 2017YFA0304004. We thank T.~Chen, S.~Pirandola, and F.~Xu for enlightening discussion, especially T.~Chen for providing the alternative methods for compensating the phase reference difference.

\begin{appendix}
\section{Security Proof of PM-QKD}\label{Sc:SecureProof}
In this section, we provide a security proof for the PM-QKD protocol with $d=2$ via entanglement distillation \cite{Bennett1996BDSW}. Note that the existing security proofs of discrete-variable QKD assume qubit (or qudit) states transmitted through the channel. Here, we develop a new security proof by exploring continuous optical modes directly.

The organization of the proof is presented as follows. First, we briefly review the main results of the Lo-Chau \cite{Lo1999Unconditional} and Shor-Preskill \cite{Shor2000Simple} security proofs in Appendix~\ref{Sc:SecureEDP}. Then, we provide a virtual entanglement-based protocol (Protocol I) in Appendix~\ref{Sc:SecureScenario} and present a key result as Lemma \ref{Lem:parity}. In Appendix~\ref{Sc:SecureCoherent}, we employ a few equivalency arguments, and eventually, we prove the security of PM-QKD in Appendix~\ref{Sc:SecureAnnounce}. In Appendix~\ref{Sc:decoy}, we employ the decoy-state method to give tight bounds on phase error rates. In Appendix~\ref{Sc:SecurePost}, we solve the phase-reference issue with the phase postcompensation technique.

Here, we introduce some definitions and notations for later discussions. For an optical mode $A$, whose creation operator is $a^\dag$, its Hilbert space is denoted as $\mathcal{H}^A$. Let $\mathcal{D}(\mathcal{H}^A)$ denote the space of density operators acting on $\mathcal{H}^A$ and $\mathcal{L}(\mathcal{H}^A)$ denote the space of linear operators acting on $\mathcal{H}^A$. A Fock state $\ket{k}_A$ with $k$ photons in mode $A$ is defined as
\begin{equation} \label{eq:defFock}
\ket{k}_A \equiv \dfrac{(a^\dag)^k}{\sqrt{k!}} \ket{0}_A,
\end{equation}
where $\ket{0}_A$ is the vacuum state. A coherent state $\ket{\alpha}_A$ is defined as
\begin{equation}
\begin{aligned}
\ket{\alpha}_A &\equiv e^{-\frac{1}{2}|\alpha|^2} \sum_{k=0}^{\infty} \dfrac{\alpha^k}{\sqrt{k!}} \ket{k}_A \\
&= e^{-\frac{1}{2}|\alpha|^2} \sum_{k=0}^{\infty} \dfrac{(\alpha a^\dag)^k}{k!} \ket{0}_A  \\
&= e^{-\frac{1}{2}|\alpha|^2} e^{\alpha a^\dag} \ket{0}_A.
\end{aligned}
\end{equation}
The photon number of $\ket{\alpha}_A$ follows a Poisson distribution,
\begin{equation}
P(k) = e^{-\mu} \dfrac{\mu^k}{k!},
\end{equation}
where $\mu = |\alpha|^2$ is the mean photon number or the light intensity.

Define the odd subspace $\mathcal{H}^A_{odd}\subseteq\mathcal{H}^A$ which is spanned by the odd Fock states $\{\ket{k}_A\}$, where all the photon numbers $k$ are odd. Similarly, define the even subspace $\mathcal{H}^A_{even}\subseteq\mathcal{H}^A$, where the photon numbers are even. Name a state $\rho\in \mathcal{D}(\mathcal{H}^A_{odd})$ to be the odd state and a state $\rho\in \mathcal{D}(\mathcal{H}^A_{even})$ to be even state. Name a state $\rho\in \mathcal{D}(\mathcal{H}^A_{odd})$ or $\rho\in \mathcal{D}(\mathcal{H}^A_{even})$ to be parity state.
Denote

\begin{equation} \label{eqn:coherent state parity}
\begin{aligned}
\ket{\alpha_{odd}}_A &\equiv \dfrac{1}{2\sqrt{c_{odd}}} (\ket{\alpha}_A - \ket{-\alpha}_A) \\
&= \dfrac{1}{\sqrt{c_{odd}}} e^{-\frac{1}{2} |\alpha|^2 } \sum_{k=0}^{\infty} \dfrac{(\alpha)^{2k+1}}{\sqrt{(2k+1)!}}\ket{2k+1}_A, \\
\ket{\alpha_{even}}_A &\equiv \dfrac{1}{2\sqrt{c_{even}}} (\ket{\alpha}_A + \ket{-\alpha}_A) \\
&= \dfrac{1}{\sqrt{c_{even}}} e^{-\frac{1}{2} |\alpha|^2 } \sum_{k=0}^{\infty} \dfrac{(\alpha)^{2k}}{\sqrt{(2k)!}}\ket{2k}_A, \\
\end{aligned}
\end{equation}

where
\begin{equation}
\begin{aligned}
\label{eqn: c normal}
c_{odd} &= e^{-\mu} \sum_{k=0}^{\infty} \dfrac{\mu^{2k+1}}{(2k+1)!} = e^{-\mu} \sinh \mu, \\
c_{even} &= e^{-\mu} \sum_{k=0}^{\infty} \dfrac{\mu^{2k}}{(2k)!} = e^{-\mu} \cosh \mu, \\
\end{aligned}
\end{equation}
are the normalization factors with $c_{odd} + c_{even} =1$, and $\mu = |\alpha|^2$ is the light intensity. It is not hard to see that $\ket{\alpha_{odd}}_A \in\mathcal{H}^A_{odd}$ and $\ket{\alpha_{even}}_A\in\mathcal{H}^A_{even}$.

Denote the photon number measurement $\{M_k\}_k$ as
\begin{equation}
M_k \equiv \ket{k}_A\bra{k}.
\end{equation}
Denote the parity measurement $\{M_{odd}, M_{even}\}$ as
\begin{equation} \label{eqn:Mpar}
\begin{aligned}
M_{odd} &\equiv \sum_{k=0}^{\infty} \ket{2k+1}_A\bra{2k+1}, \\
\quad M_{even} &\equiv \sum_{k=0}^{\infty} \ket{2k}_A\bra{2k}. \\
\end{aligned}
\end{equation}

For a beam splitter (BS), we express the input optical modes as $A, B$, with creation operators $a^\dag, b^\dag$, respectively, and the output optical modes as $C, D$, with creation operators $c^\dag, d^\dag$, respectively. The BS transforms modes $A$ and $B$ to $C$ and $D$ according to
\begin{equation}
\begin{pmatrix}
c^\dag \\
d^\dag
\end{pmatrix}
=
\dfrac{1}{\sqrt{2}}
\begin{pmatrix}
1 & 1 \\
1 & -1
\end{pmatrix}
\begin{pmatrix}
a^\dag \\
b^\dag
\end{pmatrix}.
\end{equation}

For a qubit system $A^\prime$, the Hilbert space is denoted by $\mathcal{H^{A^\prime}}$. The Pauli operators on $\mathcal{H^{A^\prime}}$ are denoted as $X_{A^\prime}, Y_{A^\prime}$ and $Z_{A^\prime}$. The eigenstates of $X_{A^\prime}, Y_{A^\prime}$ and $Z_{A^\prime}$ are denoted by $\{\ket{\pm}_{A^\prime}\}; \{\ket{\pm i}_{A^\prime}\}$ and $\{\ket{0}_{A^\prime},\ket{1}_{A^\prime}\}$, respectively. The $X$-basis measurement is denoted by $M_X: \{\ket{+}_{A^\prime}\bra{+}, \ket{-}_{A^\prime}\bra{-} \}$. The $Z$-basis measurement is denoted by $M_Z: \{\ket{0}_{A^\prime}\bra{0}, \ket{1}_{A^\prime}\bra{1} \}$.

A control-phase gate, $C_\pi$, from a qubit $A^\prime$ to an optical mode $A$, is defined as
\begin{equation} \label{eq:Cpi}
\begin{aligned}
C_{\pi} \equiv \ket{0}_{A^\prime}\bra{0}\otimes U_A(0) + \ket{1}_{A^\prime}\bra{1}\otimes U_A(\pi), \\
\end{aligned}
\end{equation}
where $U_A(\phi)\equiv e^{i\phi a^\dag a}$ is a $\phi$-phase shifter operation on the mode $A$.

\begin{definition} \label{def:equiv}
Two QKD protocols are \emph{equivalent} if the following criteria are satisfied:
\begin{enumerate}
\item
The quantum states transmitted in the channel are the same.

\item
All announced classical information is the same.

\item
Alice and Bob perform the same measurement on the same quantum states to obtain the raw key bits.

\item
Alice and Bob use the same postprocessing to extract secure key bits.
\end{enumerate}
\end{definition}

Obviously, equivalent QKD protocols will lead to identical key rates.

\subsection{Security proof via entanglement distillation}\label{Sc:SecureEDP}
Here, we briefly review the security proof based on entanglement distillation \cite{Lo1999Unconditional,Shor2000Simple}. Suppose that in QKD, Alice generates an $l$-bit key string $S$, and Bob generates an estimate of the key string $S'$. Denote the space of $S$ as $\mathcal{S}$, whose dimension is $2^l$. An adversary Eve attempts to learn about $S$ from the information leakage.

After QKD, Alice and Bob should share the \emph{same} key \emph{privately}. A key $S$ is called ``correct'', if $S'= S$ for any strategy of Eve, and is called ``$\epsilon_{cor}$-correct'', if
\begin{equation} \label{eqn:Defcorrect}
\begin{aligned}
Pr[S'\neq S]\leq \epsilon_{cor}.
\end{aligned}
\end{equation}
To define a private key, consider the quantum state $\rho_{AE}$ that describes the correlation between Alice's classical key $S$ and Eve's system $E$ (for any attacks). A key $S$ is called ``$\epsilon_{sec}$-private'' from $E$ if \cite{Ben2005composable,Renner2005universally}
\begin{equation}
\min_{\sigma_E} \dfrac{1}{2} || \rho_{AE} - \omega_A \otimes \sigma_E ||_1 \leq \epsilon_{sec},
\end{equation}
where $\omega_A = (|\mathcal{S}|)^{-1} \sum_{S=0}^{|\mathcal{S}|-1} \ket{S}_A\bra{S}$ is the equally mixed key state over all possible keys in space $\mathcal{S}$ and $||\cdot||_1$ is the trace norm. A QKD protocol is called ``secure'' if the generated key is both correct and private. It is called ``$\epsilon$-secure'' if the generated key is $\epsilon_{cor}$-correct and $\epsilon_{sec}$-private with $\epsilon=\epsilon_{cor} + \epsilon_{sec}$.

Here, we consider the case where Alice and Bob share an $m$-pair-of-qubits gigantic state $\rho_{A^\prime B^\prime}^{(m)}$. If we can show that
\begin{equation}
F(\rho_{A^\prime B^\prime}^{(m)}, \ket{\Phi^+}_{A^\prime B^\prime}^{(m)}) \geq \sqrt{1 - \epsilon_l^2}
\end{equation}
where $\ket{\Phi^+}_{A^\prime B^\prime}^{(m)}$ is the state of $m$ perfect EPR pairs $\ket{\Phi^+} = (\ket{00}+\ket{11})/\sqrt{2}$, and $0\le\epsilon_l\le1$, then it can be shown that it is $\epsilon_l$-private and $\epsilon_l$-correct, and then it is $2\epsilon_l$-secure. That is, to show the security of a QKD protocol, we only need to show that the state for key extraction $\rho_{A^\prime B^\prime}^{(m)}$ is close to the perfect EPR pairs $\ket{\Phi^+}_{A^\prime B^\prime}^{(m)}$.

Here is the intuition of the Lo-Chau security proof \cite{Lo1999Unconditional}. Suppose Alice and Bob share an $n$-pair-of-qubit gigantic state $\rho_{A^\prime B^\prime}^{(n)}$ at the beginning. If they perform an efficient entanglement distillation protocol (EDP) to distill $m$ EPR pairs, then they will be able to share nearly $m$ bit correct and private keys. Now the task becomes how to find the right EDP.

Bennett, DiVincenzo, Smolin, and Wootters (BDSW) show that \cite{Bennett1996BDSW}, if an $n$-pair-of-qubit state $\rho_{A^\prime B^\prime}^{(n)}$ can be written as a classical mixture of Bell state products,
\begin{equation} \label{eqn:classical Bell}
\begin{aligned}
&\rho_{A^\prime B^\prime}^{(n)} = \sum_{b_1, b_2, ...,b_n} p_{b_1, b_2,...,b_n} \ket{b_1, b_2,...,b_n}\bra{b_1, b_2,...,b_n}, \\
&\ket{b_1, b_2,...b_n}  = \bigotimes_{i=1}^{n} \ket{\Phi^{(i)}_{b_i}}, \\
\end{aligned}
\end{equation}
where $\ket{\Phi^{(i)}_{b_i}}$ is one of the four Bell states labeled by $b_i \in \{0,1,2,3\}$ on the $i$ th qubit, then one can distill entanglement by employing EDPs. In the one-way hashing method \cite{Bennett1996BDSW}, a specific type of one-way EDP, Alice measures a series of commuting operators based on some random-hashing matrix on her $n$ qubits, and she sends the results to Bob. Bob measures the same operators on his $n$ qubits and infers the locations and types of the errors from the difference in the measurement results. After that, Alice and Bob correct the errors and obtain $m$ ($m\leq n$, almost perfect) EPR pairs.

In general, of course, the initial state $\rho_{A^\prime B^\prime}^{(n)}$ can deviate from Bell-diagonal states and become highly entangled between different pairs. In the Lo-Chau  security proof \cite{Lo1999Unconditional}, it has been shown that, for the one-way hashing EDP introduced above, the error syndrome and EDP performance of such $\rho_{A^\prime B^\prime}^{(n)}$ is the same as that of the state after dephasing between pairs,
\begin{equation}
\begin{aligned}
\rho_{A^\prime B^\prime, dep}^{(n)} &\equiv W \rho_{A^\prime B^\prime}^{(n)} W, \\
W & = \sum_{b_1, b_2, ...,b_n} \ket{b_1, b_2,...,b_n}\bra{b_1, b_2,...,b_n}.
\end{aligned}
\end{equation}
where $\ket{b_1, b_2,...,b_n}$ is defined in Eq.~\eqref{eqn:classical Bell}. Therefore, one can reduce the EDP for general $\rho_{A^\prime B^\prime}^{(n)}$ (i.e., the case with coherent attacks) to the case in Eq.~\eqref{eqn:classical Bell}.

In the Shor-Preskill security proof \cite{Shor2000Simple}, the one-way EDP protocol is reduced to a ``prepare-and-measure'' QKD protocol, by employing the CSS code\cite{CSS1996qec}. Later, other techniques of decoupling $X$-error correction and $Z$-error correction are introduced for this reduction. For a one-way EDP protocol based on the CSS code, the distillation rate of EPR pairs, $r \equiv \lim_{n\rightarrow\infty}m/n$, is given by
\begin{equation} \label{eqn:Shor-Preskill key rate}
r = 1 - H(E^Z) - H(E^X),
\end{equation}
where $E^Z$ and $E^X$ are the $Z$-error rate and $X$-error rate, respectively, and $H(x)=-x\log_2x-(1-x)\log_2(1-x)$ is the binary Shannon entropy function. The error rates can defined as measurement results on $\rho^{(n)}_{A^\prime B^\prime}$,
\begin{equation}
\begin{aligned}
\label{eqn: EX EZ}
E^Z &\equiv Tr[ \bigoplus_{j=1}^{n} \dfrac{1}{2}(1- Z^{(j)}_{A^\prime} \otimes Z^{(j)}_{B^\prime}) \rho^{(n)}_{A^\prime B^\prime}], \\
E^X &\equiv Tr[ \bigoplus_{j=1}^{n} \dfrac{1}{2}(1- X^{(j)}_{A^\prime} \otimes X^{(j)}_{B^\prime}) \rho^{(n)}_{A^\prime B^\prime}], \\
\end{aligned}
\end{equation}
where $X^{(j)}_{A^\prime(B^\prime)}$ is the Pauli $X$ operator and $Z^{(j)}_{A^\prime(B^\prime)}$ is the Pauli $Z$ operator, on the $j$ th pair-of-qubits system $A^\prime (B^\prime)$. Since the bit value is measured in the $Z$-basis and the phase value is measured in the $X$-basis, we also call the $Z$-basis error the ``bit error'' and the $X$-basis error the ``phase error''.

Following the Shor-Preskill security proof, distillable entanglement of the one-way EDP protocol is the key rate for some prepare-and-measure QKD protocols such as BB84 \cite{Bennett1984Quantum}. In the BB84 protocol, we can estimate $E^Z, E^X$ by randomly measuring the qubits in the $X$- and $Z$-basis. In the PM-QKD protocol, on the other hand, Alice and Bob can only measure in the $Z$-basis. In the following sections, we will introduce entanglement-based PM-QKD protocols and discuss how to infer the value of $E^X$.

\subsection{Entanglement-based PM-QKD protocol} \label{Sc:SecureScenario}
We first introduce an entanglement-based PM-QKD protocol, called Protocol I, as shown in Fig.~\ref{fig:Pro1}.

\begin{figure*}[htbp]
\centering
\includegraphics[width=16cm]{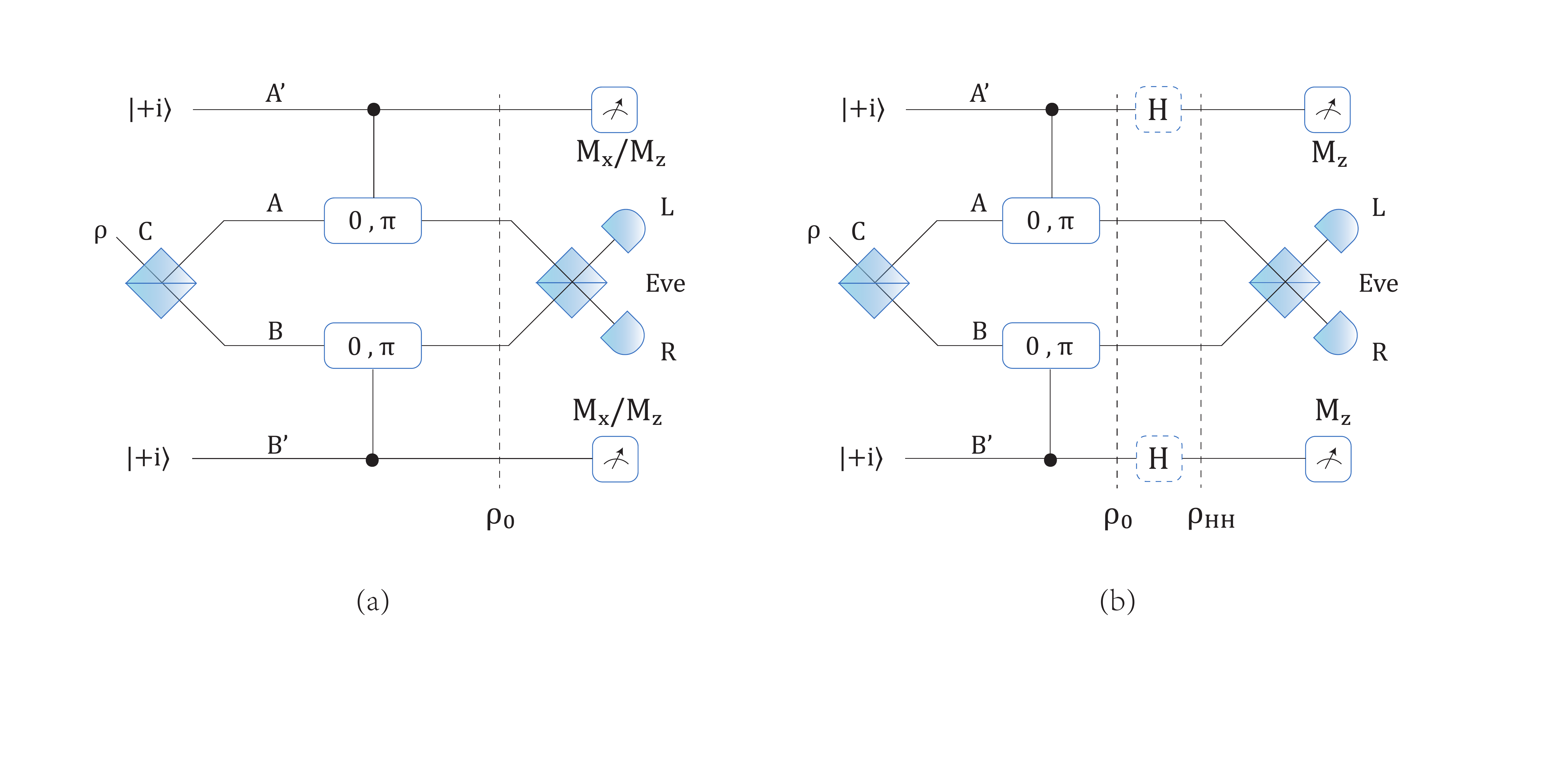}
\caption{(a) Schematic diagram of Protocol I. (b) Equivalent view of Protocol I, where the $X$-basis measurement on $A^\prime$ or $B^\prime$ is realized by a Hadamard gate followed by the $Z$-basis measurement.} \label{fig:Pro1}
\end{figure*}

\textbf{\uline{Protocol I}}
\begin{enumerate}
\item \label{StepPrep}
State preparation: A trusted party, Charlie, picks a state $\rho$ on optical mode $C$, splits $\rho$ into two pulses, $A$ and $B$; and sends them to Alice and Bob, respectively. Alice and Bob initialize their qubits in $\ket{+i}$. Alice applies the control gate $C_{\pi}$, defined in Eq.~\eqref{eq:Cpi}, to qubit $A^\prime$ and optical pulse $A$. Similarly, Bob applies $C_{\pi}$ to $B^\prime$ and $B$.

\item \label{StepMeasure}
Measurement: The two optical pulses $A$ and $B$ are sent to an untrusted party, Eve, who is supposed to perform interference measurement and record which detector ($L$ or $R$) clicks.

\item \label{StepAnnounce}
Announcement: Eve announces the detection result, $L/R$ click or failure, for each round.

\item \label{StepSifting}
Sifting: When Eve announces an $L/R$ click, Alice and Bob keep the qubits of systems $A^\prime$ and $B^\prime$. In addition, Bob applies a Pauli $Y$-gate to his qubit if Eve's announcement is $R$ click.

\item \label{StepParaEst}
Parameter estimation: After many rounds of the above steps, Alice and Bob end up with a joint $2n$-qubit state, denoted by $\rho_{A^\prime B^\prime}^{(n)}$. They then perform random sampling on the remaining $\rho_{A^\prime B^\prime}^{(n)}$ to estimate $E^Z$ and infer $E^X$ by Eq.~\eqref{eqn:EX par}.

\item\label{StepExtraction}
Key distillation: Alice and Bob apply a standard EDP when the error rates are below a certain threshold. The distillation ratio $r$ is given by Eq.~\eqref{eqn:Shor-Preskill key rate}. Once Alice and Bob obtain $nr$ (almost) pure EPR pairs, they both perform local $Z$ measurements on the qubits to generate private keys.
\end{enumerate}

Our first observation is that, if the state $\rho$ prepared by Charlie is a parity state, namely, $\rho\in\mathcal{D}(\mathcal{H}^C_{odd})$ or $\rho\in\mathcal{D}(\mathcal{H}^C_{even})$, then the $X$-error rate and $Z$-error rate are correlated. Here, we denote the $Z$-error rate and $X$-error rate for an odd state $\rho_{odd}$ as $e_{odd}^Z$ and $e_{odd}^X$, respectively, and similarly, denote for an even state $\rho_{even}$ as $e_{even}^Z$ and $e_{even}^X$.

\begin{lemma}\label{Lem:parity}
In Protocol I, for $\rho\in\mathcal{D}(\mathcal{H}^C_{odd})$, $e_{odd}^X = e_{odd}^Z$; and for $\rho\in\mathcal{D}(\mathcal{H}^C_{even})$, $e_{even}^X = 1- e_{even}^Z$.
\end{lemma}

\begin{proof}
First consider the case when $\rho$ is a Fock state $\ket{k}_C$, defined in Eq.~\eqref{eq:defFock}. After passing through the BS, as shown in Fig.~\ref{fig:Pro1}, the state on modes $A$ and $B$ becomes
\begin{equation}
\dfrac{1}{\sqrt{2^k k!}} (a^\dag + b^\dag)^k \ket{00}_{AB}.
\end{equation}
The joint state on system $A^\prime, B^\prime, A, B$ before the $C_\pi$ operations is
\begin{equation}
\begin{aligned}
& \dfrac{1}{\sqrt{2^k k!}} \ket{+i+i}_{A^\prime B^\prime} (a^\dag + b^\dag)^k \ket{00}_{AB} \\
& = \dfrac{1}{\sqrt{2^k k!}}\dfrac{1}{2} [(\ket{00} - \ket{11}) + i (\ket{01} + \ket{10}) ]_{A^\prime B^\prime} (a^\dag + b^\dag)^k \ket{00}_{AB}.
\end{aligned}
\end{equation}
After the $C_\pi$ operations, this state becomes

\begin{widetext}
\begin{equation} \label{eqn:rho0 fock}
\begin{aligned}
\ket{\Psi^{(k)}_0} = &
\begin{cases}
&\dfrac{1}{2\sqrt{2^k k!}} [ (\ket{00} + \ket{11} )_{A^\prime B^\prime} (a^\dagger + b^\dagger)^k \ket{00}_{A B} + i(\ket{01} - \ket{10} )_{A^\prime B^\prime} (a^\dagger - b^\dagger)^k \ket{00}_{AB} ], \\
& \text{if } k \text{ is odd}, \\
&\dfrac{1}{2\sqrt{2^k k!}} [ (\ket{00} - \ket{11} )_{A^\prime B^\prime} (a^\dagger + b^\dagger)^k \ket{00}_{A B} + i(\ket{01} + \ket{10} )_{A^\prime B^\prime} (a^\dagger - b^\dagger)^k \ket{00}_{AB} ], \\
& \text{if } k \text{ is even}.
\end{cases} \\
\end{aligned}
\end{equation}
\end{widetext}

As shown in Fig.~\ref{fig:Pro1}(b), the $X$-basis measurement on $A^\prime$ or $B^\prime$ is realized by a Hadamard gate followed by the $Z$-basis measurement. Denote the state after local Hadamard gates as $ \ket{\Psi^{(k)}_{HH}} \equiv (H^{A^\prime}\otimes H^{B^\prime}) \ket{\Psi^{(k)}_0}$; then,

\begin{widetext}
\begin{equation} \label{eqn:rhoHH fock}
\begin{aligned}
\ket{\Psi^{(k)}_{HH}} = &
\begin{cases}
&\dfrac{1}{2\sqrt{2^k k!}} [ (\ket{00} + \ket{11} )_{A^\prime B^\prime} (a^\dagger + b^\dagger)^k \ket{00}_{A B} + i(\ket{01} - \ket{10} )_{A^\prime B^\prime} (a^\dagger - b^\dagger)^k \ket{00}_{AB} ], \\
& \text{if } k \text{ is odd}, \\
&\dfrac{1}{2\sqrt{2^k k!}} [ (\ket{01} + \ket{10} )_{A^\prime B^\prime} (a^\dagger + b^\dagger)^k \ket{00}_{A B} + i(\ket{00} - \ket{11} )_{A^\prime B^\prime} (a^\dagger - b^\dagger)^k \ket{00}_{AB} ], \\
& \text{if } k \text{ is even}.
\end{cases} \\
\end{aligned}
\end{equation}
\end{widetext}

In other words, the $X$-error rate $e^X_k$ can be understood as the error rate by performing the $Z$-basis measurement on the state of $\ket{\Psi^{(k)}_{HH}}$. The relation between the $X$- and $Z$-error rates can be obtained by comparing Eqs.~\eqref{eqn:rho0 fock} and \eqref{eqn:rhoHH fock}. For the odd-photon-number case, since $\ket{\Psi^{(2k+1)}_{HH}} = \ket{\Psi^{(2k+1)}_0}$, we have $e^X_{2k+1} = e^Z_{2k+1}$. For the even-photon-number case, since $\ket{\Psi^{(2k)}_{HH}} = I^{A^\prime}\otimes Y^{B^\prime}\ket{\Psi^{(2k)}_0}$, we have $e^X_{2k} = 1 - e^Z_{2k}$.


Now, let us consider the case of pure parity states. For an odd state $\ket{\psi_{odd}}_C = \sum c_{2k+1} \ket{{2k+1}}_C $, where $\sum_k |c_{2k+1}|^2 = 1$, since the BS and $C_\pi$ are unitary operations, the state after these operations can be written as
\begin{equation}
\label{eq:pure odd psi}
\ket{\Psi_0^{(odd)}} = \sum_{k} c_{2k+1} \ket{\Psi_0^{(2k+1)}}.
\end{equation}
After local Hadamard gates, the state becomes $ \ket{\Psi^{(odd)}_{HH}} \equiv (H^{A^\prime}\otimes H^{B^\prime}) \ket{\Psi^{(odd)}_0} $. From Eqs.~\eqref{eqn:rho0 fock}, \eqref{eqn:rhoHH fock}, \eqref{eq:pure odd psi}, we can see that $\ket{\Psi^{(odd)}_{HH}} = \ket{\Psi^{(odd)}_0}$, and hence $e^Z_{odd} = e^X_{odd}$. With the same argument, we have $e^Z_{even} = 1 - e^X_{even}$.

For general parity states, we can regard them as mixtures of pure parity states,
\begin{equation}
\begin{aligned}
\rho_{odd} &= \sum_{i} p_i \ket{\psi_{odd}^{(i)}}\bra{\psi_{odd}^{(i)}}, \\
\rho_{even} &= \sum_{i} p_i \ket{\psi_{even}^{(i)}}\bra{\psi_{even}^{(i)}}.
\end{aligned}
\end{equation}
This is equivalent to Charlie sending out $\ket{\psi_{odd(even)}^{(i)}}$ with probability $p_i$. For each pure state component, we have $e^X_{odd(i)} = e^Z_{odd(i)}$ and $e^X_{even(i)} = 1- e^Z_{even(i)}$. Thus, the relations hold for all (mixed) parity states.
\end{proof}

In general, Charlie might not use a parity state. Consider the case that Charlie performs the parity measurement $\{M_{odd},M_{even}\}$, defined in Eq.~\eqref{eqn:Mpar}, before sending to Alice and Bob. Denote the measurement outcome probabilities for the odd and even parity to be $p_{odd}$ and $p_{even}$, respectively. This is equivalent to Charlie preparing an odd state $\rho_{odd}$ and an even state $\rho_{even}$ with probabilities $p_{odd}$ and $p_{even}$, respectively. Then, the state can be written as
\begin{equation}
\rho = p_{odd} \rho_{odd} + p_{even} \rho_{even},
\end{equation}
that is, $\rho\in \mathcal{D}(\mathcal{H}_{odd}^C \oplus \mathcal{H}_{even}^C)$.


Suppose Charlie announces the parity information publicly to Alice and Bob.
Then, they can label the sifted qubits with ``odd" and ``even". Denote $q_{odd}$ and $q_{even}$, with $q_{odd}+q_{even}=1$, to be the fractions of odd- and even-labeled states in the sifted $n$-pairs of qubit states $A^\prime, B^\prime$, respectively. Then, according to Lemma \ref{Lem:parity}, the total $X$-error rate can be calculated by
\begin{equation} \label{eqn:EX par}
\begin{aligned}
E^X &= q_{odd} e^X_{odd} + q_{even} e^X_{even} \\
&= q_{odd} e^Z_{odd} + q_{even} (1 - e^Z_{even}). \\
\end{aligned}
\end{equation}
Here, the parameters $q_{odd}$, $q_{even}$, $e^Z_{odd}$, and $e^Z_{even}$ can be evaluated according to Charlie's parity announcement. When the parity information is missing, Alice and Bob need to estimate these parameters, which will be described later in Appendix \ref{Sc:decoy}.

\subsection{Coherent state protocol and equivalent process} \label{Sc:SecureCoherent}

From Protocol I to the PM-QKD protocol, there are a few practical issues that need to be addressed. First, the trusted party Charlie and the BS should be removed. In order to do this, we consider a special case where Charlie prepares a coherent state $\ket{\sqrt{\mu}}_C$ as the photon source. In addition, Charlie adds a random phase $\phi\in\{0,\pi\}$ with equal probabilities to the coherent state, so the density matrix can be written as
\begin{equation} \label{eqn:coherent state parity addition}
\begin{aligned}
\rho(\alpha, -\alpha) &= \dfrac{1}{2}(\ket{\alpha}\bra{\alpha} + \ket{-\alpha}\bra{-\alpha}) \\
&= c_{odd}\ket{\alpha_{odd}}\bra{\alpha_{odd}} + c_{even}\ket{\alpha_{even}}\bra{\alpha_{even}}, \\
\end{aligned}
\end{equation}
where $\alpha=\sqrt{\mu}$, and $\ket{\alpha_{odd}}$, $\ket{\alpha_{even}}$, $c_{odd}$ and $c_{even}$ are defined in Eqs.~\eqref{eqn:coherent state parity} and Eq.~\eqref{eqn: c normal}. Clearly, one can see that $\rho(\alpha, -\alpha) \in \mathcal{D}(\mathcal{H}_{odd}^C \oplus \mathcal{H}_{even}^C)$.


\begin{figure*}[htbp]
\centering
\includegraphics[width=16cm]{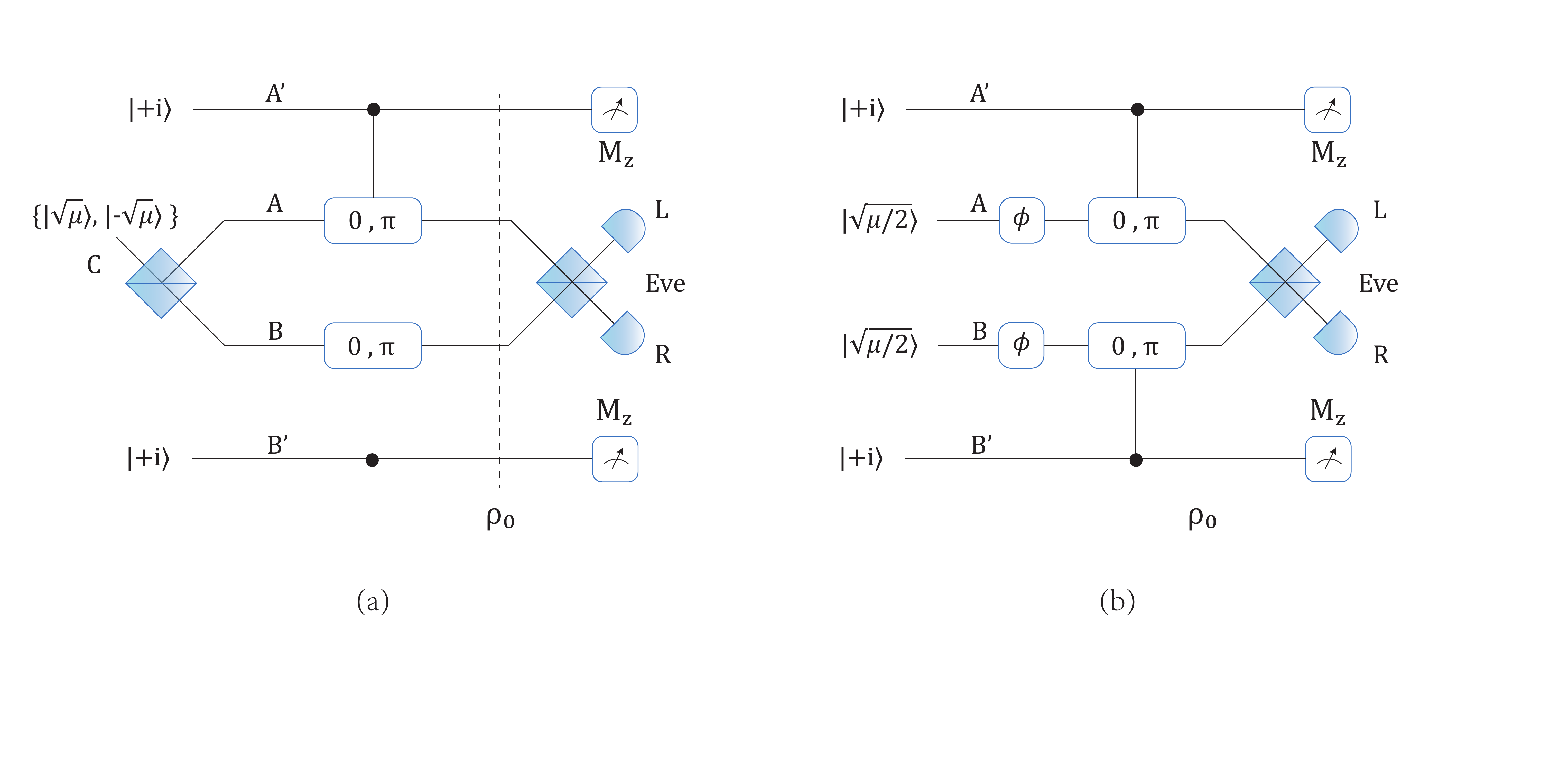}
\caption{(a) A specific realization of Protocol I, where Charlie prepares $\ket{\sqrt{\mu}}$ and $\ket{\sqrt{-\mu}}$ with equal probabilities. (b) Schematic diagram of Protocol II, where Alice and Bob prepare $\ket{\sqrt{\mu/2}}$ and add the same random phase $\phi\in\{0,\pi\}$.} \label{fig:Pro2}
\end{figure*}

Now, Alice and Bob want to prepare the state given in Eq.~\eqref{eqn:coherent state parity addition} without Charlie's assistance. They first locally prepare two coherent states $\ket{\sqrt{\mu/2}}_{A(B)}$ and add the same random phase $\phi\in\{0,\pi\}$ to their states, as shown in Fig.~\ref{fig:Pro2}(b). Then, the state becomes
\begin{widetext}
\begin{equation}
\label{eqn:two coherent state parity addition}
\begin{aligned}
\rho_{AB}(\sqrt{\mu},-\sqrt{\mu}) =  \dfrac{1}{2}(\ket{\sqrt{\mu/2}}_A\bra{\sqrt{\mu/2}} \otimes \ket{\sqrt{\mu/2}}_B\bra{\sqrt{\mu/2}}
 + \ket{\sqrt{-\mu/2}}_A\bra{\sqrt{-\mu/2}} \otimes \ket{\sqrt{-\mu/2}}_B\bra{\sqrt{-\mu/2}}).
\end{aligned}
\end{equation}
\end{widetext}
Apparently, $\rho_{AB}(\sqrt{\mu},-\sqrt{\mu})$ is the state after $\rho(\alpha, -\alpha)$ defined in Eq.~\eqref{eqn:coherent state parity addition} going through the BS. Thus, this new protocol (Protocol II) is equivalent to Protocol I with the input state $\rho(\sqrt{\mu},-\sqrt{\mu})$.

Protocol II runs as follows, as shown in Fig.~\ref{fig:Pro2}(b). Here, $\mu_a = \mu_b = \mu/2$.

\textbf{\uline{Protocol II}}
\begin{enumerate}
\item 
State preparation:
Alice and Bob prepare coherent states $\ket{\sqrt{\mu_a}}$ and $\ket{\sqrt{\mu_b}}$ on optical modes $A$ and $B$, respectively. They initialize the qubits $A^\prime$ and $B^\prime$ in $\ket{+i}$. They add the same random phase $\phi\in\{0,\pi\}$ on the optical modes $A$ and $B$. Alice applies the control gate $C_{\pi}$, defined in Eq.~\eqref{eq:Cpi}, to qubit $A^\prime$ and optical pulse $A$. Similarly, Bob applies $C_{\pi}$ to $B^\prime$ and $B$.

\item 
Measurement: The two optical pulses $A$ and $B$ are sent to an untrusted party, Eve, who is supposed to perform interference measurement and record which detector ($L$ or $R$) clicks.

\item 
Announcement: Eve announces the detection result, $L/R$ click or failure, for each round.

\item 
Sifting: When Eve announces an $L/R$ click, Alice and Bob keep the qubits of systems $A^\prime$ and $B^\prime$. In addition, Bob applies a Pauli $Y$-gate to his qubit if Eve's announcement is an $R$ click.

\item 
Parameter estimation: After many rounds of the above steps, Alice and Bob end up with a joint $2n$-qubit state, denoted by $\rho_{A^\prime B^\prime}^{(n)}$. They then perform random sampling on the remaining $\rho_{A^\prime B^\prime}^{(n)}$ to estimate $E^Z$ and infer $E^X$ by Eq.~\eqref{eqn:EX par}.

\item 
Key distillation: Alice and Bob apply a standard EDP when the error rates are below a certain threshold. The distillation ratio $r$ is given by Eq.~\eqref{eqn:Shor-Preskill key rate}. Once Alice and Bob obtain $nr$ (almost) pure EPR pairs, they both perform local $Z$ measurements on the qubits to generate private keys.
\end{enumerate}

From Protocol II to the PM-QKD protocol, there are still some missing links.
\begin{enumerate}
\item
Alice and Bob need to add the same random phase $\phi\in\{0,\pi\}$ to their states, which would cost them a private bit.

\item
They use Eq.~\eqref{eqn:EX par} to infer $E^X$, which needs information on the parameters of $ q_{odd}, q_{even}, e^Z_{odd}$ and $e^Z_{even}$.

\item
The phase references of Alice and Bob's coherent states are locked. That is, the phases of the initial coherent states, $\ket{\sqrt{\mu_a}}$ and $\ket{\sqrt{\mu_b}}$, are the same, and hence remote phase locking is required.
\end{enumerate}
We shall remove the simultaneous phase randomization requirement in Appendix \ref{Sc:SecureAnnounce}, bound $q_{odd}, q_{even}, e^Z_{odd}$ and $e^Z_{even}$ in Appendix \ref{Sc:decoy}, and remove the phase-locking requirement in Appendix \ref{Sc:SecurePost}.

\subsection{Security with phase announcement}\label{Sc:SecureAnnounce}
To remove the simultaneous phase randomization requirement in Protocol II, a straightforward idea is that Alice and Bob add random phases $\phi_{a(b)}\in\{0,\pi\}$ independently (with equal probability $1/2$); during the sifting step, they announce the phases $\phi_a, \phi_b$ and postselect the bits only for $\phi_a=\phi_b$, as shown in Fig.~\ref{fig:Pro3}. All the other steps remain unchanged. Here, there are two modifications that need to be analyzed, the phase announcement of $\phi_a, \phi_b$ and postselection of $\phi_a=\phi_b$.

We deal with the phase announcement of $\phi_a, \phi_b$ first. Consider Protocol IIa, as shown in Fig.~\ref{fig:Pro3}, where Alice and Bob announce the random phase $\phi$ in Protocol II during postprocessing.


\begin{figure*}[htbp]
\centering
\includegraphics[width=18cm]{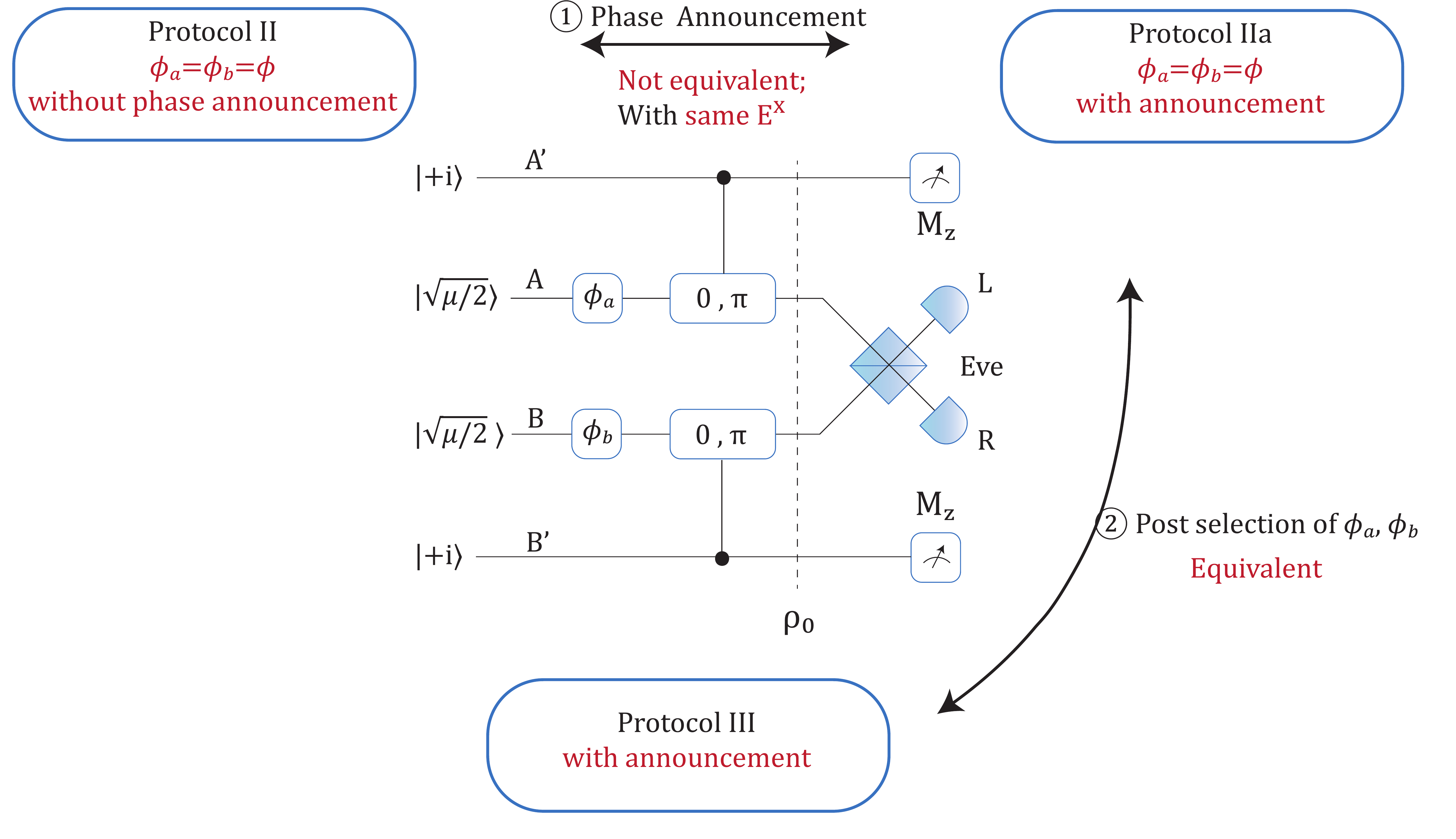}
\caption{The schematic diagram of protocols II, IIa, and III. In Protocol II, Alice and Bob add the same phase $\phi\in\{0,\pi\}$ on their coherent state $\sqrt{\mu/2}$. In Protocol IIa, Alice and Bob announce the phase $\phi$ after Eve's announcement. In Protocol III, Alice and Bob add phases $\phi_a$ and $\phi_b$, independently. They announce $\phi_a$ and $\phi_b$ after Eve's announcement.}

\label{fig:Pro3}
\end{figure*}

\textbf{\uline{Protocol IIa}}
\begin{enumerate}
\item 
State preparation:
Same as Protocol II.

\item 
Measurement: Same as Protocol II.

\item 
Announcement: Same as Protocol II.

\item 
Sifting: The first part is the same as Protocol II. After that, Alice and Bob announce the random phase $\phi$.

\item 
Parameter estimation: Same as Protocol II.

\item 
Key distillation: Same as Protocol II.
\end{enumerate}

Note that the only difference between Protocol II and IIa is that Alice and Bob announce their phase $\phi$ after Eve's announcement. Since the classical information announced during postprocessing is different, according to Definition~\ref{def:equiv}, Protocol II and IIa are not equivalent. In fact, from some specific attacks (see Appendix~\ref{sc:attack}), one can see that the security of the two protocols with and without phase announcement can be very different.

In Protocol II, as shown in Fig.~\ref{fig:Pro3}, because of phase randomization, one can always assume that Alice and Bob (or Eve) perform the total parity measurement $\{M_{odd}^t, M_{even}^t\}$, which is defined as, similar to Eq.~\eqref{eqn:Mpar},
\begin{equation}
\begin{aligned}
\label{eqn:MeasureTotalParity}
M_{odd}^t &\equiv \sum_{k_1 + k_2 \text{ is odd}} \ket{k_1 k_2}_{AB}\bra{k_1 k_2}, \\
M_{even}^t &\equiv \sum_{k_1 + k_2 \text{ is even}} \ket{k_1 k_2}_{AB}\bra{k_1 k_2}. \\
\end{aligned}
\end{equation}
Here, $\{\ket{k_1 k_2}_{AB}\}$ are the Fock states on mode $A,B$, with $k_1$ photons on $A$ and $k_2$ photons on $B$. That is, in Protocol II, the photon source can be regarded as a mixture of odd states and even states. The parity-state channel model is similar to the photon number channel used in the security proof of the decoy-state method \cite{Lo2005Decoy,Ma2008PhD}.

In Protocol IIa, as shown in Fig.~\ref{fig:Pro3}, on the other hand, no such parity measurement is allowed because it does not commute with the phase announcement. That is, Eve can distinguish whether or not Alice and Bob perform $\{M_{odd}^t, M_{even}^t\}$ after the announcement of phases $\phi_a$ and $\phi_b$. In other words, the quantum signals sent by Alice and Bob can no longer be regarded as a mixture of $\rho_{odd}$ and $\rho_{even}$ in Protocol IIa. To analyze the security of Protocol IIa, we notice the following observation.

\begin{observation}\label{obs:same rho}
The joint states $\rho_{A^\prime B^\prime}^{(n)}$ obtained in protocols I, II, and IIa are the same after sifting in Step \ref{StepSifting}.
\end{observation}

Note that qubit systems $A^\prime$ and $B^\prime$ are local and never sent to Eve. Before the sifting step, these qubits are these qubits are identical and sampled independently (i.i.d.) and the same for protocols I, II, and IIa. The state $\rho_{A^\prime B^\prime}^{(n)}$ is postselected by Eve's announcement. Of course, Eve can manipulate $\rho_{A^\prime B^\prime}^{(n)}$ by different measurement or announcement strategies. We emphasize that \emph{Eve announces before the phase announcement in Protocol IIa}. Then, her strategy cannot depend on the phase announcement. Therefore, in all three protocols, the state $\rho_{A^\prime B^\prime}^{(n)}$ remains the same. This is crucial to our security analysis. From Observation \ref{obs:same rho}, we can have the following Corollary.

\begin{corollary}\label{Cor:same EX}
Both the $X$- and $Z$-error patterns of the joint states $\rho_{A^\prime B^\prime}^{(n)}$ in protocol II and IIa are the same. Then, the $X$-error rates $E^X$ of protocols II and IIa are the same, as given by Eq.~\eqref{eqn:EX par}.
\end{corollary}

Now, we deal with the postselection of $\phi_a=\phi_b$. Regardless of the values of $\phi_a, \phi_b$, because of control-phase gates $C_\pi$, the state $\rho_{AB}$ sent to Eve is
\begin{equation}
\rho_{AB} = \rho_A(\sqrt{\mu/2}, -\sqrt{\mu/2}) \otimes \rho_B(\sqrt{\mu/2}, -\sqrt{\mu/2}),
\end{equation}
where $\rho(\sqrt{\mu/2}, -\sqrt{\mu/2})$ is defined in Eq.~\eqref{eqn:coherent state parity addition}.
The state $\rho_{AB}$ is independent of $\phi_a, \phi_b$; hence Eve's attack cannot depend on $\phi_a, \phi_b$, and the sifted qubits are also independent of the value of $\phi_a (=\phi_b)$.

Furthermore, we notice that discarded qubits with $|\phi_a-\phi_b|=\pi$ can also be used for entanglement distillation. Here, adding a phase $\pi$ to system $B$ is equivalent to performing a Pauli $Y$-gate to system $B^\prime$ for a parity-state source in protocols I, II and IIa. Thus, for the qubits with $|\phi_a-\phi_b|=\pi$, Bob performs $Y$-gate on qubit $B^\prime$.

Therefore, if Alice and Bob randomize their phases $\phi_a, \phi_b\in \{0,\pi\}$ independently and perform phase-sifting operations after Eve's announcement, shown in Fig.~\ref{fig:Pro3}, the modified protocol is equivalent to Protocol IIa. We call this Protocol III, which runs as follows, as shown in Fig.~\ref{fig:Pro3}. Here, $\mu_a = \mu_b = \mu/2$.

\textbf{\uline{Protocol III}}
\begin{enumerate}
\item 
State preparation:
Alice and Bob prepare coherent states $\ket{\sqrt{\mu_a}}$ and $\ket{\sqrt{\mu_b}}$ on optical modes $A, B$, separately. They initial their qubits $A^\prime, B^\prime$ in $\ket{+i}$. They independently add random phases $\phi_a, \phi_b \in \{0,\pi\}$ on optical modes $A, B$. Alice applies the control gate $C_{\pi}$, defined in Eq.~\eqref{eq:Cpi}, to qubit $A^\prime$ and optical pulse $A$. Similarly, Bob applies $C_{\pi}$ to $B^\prime$ and $B$.

\item 
Measurement: The two optical pulses, $A$ and $B$, are sent to an untrusted party, Eve, who is supposed to perform an interference measurement and record which detector ($L$ or $R$) clicks.

\item 
Announcement: Eve announces the detection result, $L$/$R$ click or failure, for each round.

\item 
Sifting: When Eve announces an $L/R$ click, Alice and Bob keep the qubits of systems $A^\prime$ and $B^\prime$. Bob applies a Pauli $Y$-gate to his qubit if Eve's announcement is $R$ click. After Eve's announcement, Alice and Bob announce their encoded phase $\phi_a, \phi_b$. Bob applies a Pauli $Y$-gate to his qubit if $|\phi_a-\phi_b|=\pi$.

\item 
Parameter estimation: After many rounds of the above steps, Alice and Bob end up with a joint $2n$-qubit state, denoted by $\rho_{A^\prime B^\prime}^{(n)}$. They then perform random sampling on the remaining $\rho_{A^\prime B^\prime}^{(n)}$ to estimate $E^Z$ and infer $E^X$ by Eq.~\eqref{eqn:EX par}.

\item 
Key distillation: Alice and Bob apply a standard EDP when the error rates are below a certain threshold. The distillation ratio $r$ is given by Eq.~\eqref{eqn:Shor-Preskill key rate}. Once Alice and Bob obtain $nr$ (almost) pure EPR pairs, they both perform local $Z$ measurements on the qubits to generate private keys.
\end{enumerate}

\subsection{Decoy-state method and phase randomization} \label{Sc:decoy}
Here, we introduce a method to estimate $q_{odd}$, $q_{even}$, $e_{odd}^Z$ and $e_{even}^Z$. Without loss of generality, we mainly discuss the decoy-state method in Protocol I. Similar arguments can be applied to protocols II, IIa, and III.

Recall that in Protocol I, if Charlie prepares coherent state $\ket{\sqrt{\mu}}$ and adds a random phase $\{0,\pi\}$ on it, it is equivalent to preparing odd- and even-parity states with probabilities $p_{odd}^\mu = c_{odd}$ and $p_{even}^\mu = c_{even}$, respectively, as defined in Eq.~\eqref{eqn: c normal}. Define the yield $Y_{odd}^\mu$ ($Y_{even}^\mu$) as the probability of successful detection conditional on the odd-parity (even-parity) state. The fraction of odd- and even-parity states in the final detected signal is given by
\begin{equation}
\begin{aligned}
\label{eqn:q parity}
q_{odd}^\mu &= p_{odd}^\mu \dfrac{Y_{odd}^\mu}{Q_\mu}, \\
q_{even}^\mu &= p_{even}^\mu \dfrac{Y_{even}^\mu}{Q_\mu},
\end{aligned}
\end{equation}
where $Q_\mu$ is the total gain of the signals. For signals with intensity $\mu$, we have
\begin{equation}
\begin{aligned}
\label{eqn:decoy parity}
Q^{\mu} &= p_{odd}^{\mu} Y^{\mu}_{odd} + p_{even}^{\mu} Y^{\mu}_{even}, \\
E^{Z,\mu} Q^{\mu} &= e_{odd}^{Z,\mu}p_{odd}^{\mu} Y^{\mu}_{odd} + e_{even}^{Z,\mu}p_{even}^{\mu} Y^{\mu}_{even}, \\
\end{aligned}
\end{equation}
where $E^Z_\mu$, $e^{Z,\mu}_{add}$ ($e^{Z,\mu}_{even}$) are the quantum bit error rate (QBER) and the $Z$-error rate of odd-state (even-state) signals with total a intensity of $\mu$, respectively.

If we directly estimate $q_{odd}^\mu$, $q_{even}^\mu$, $e_{odd}^{Z,\mu}$ and $e_{even}^{Z,\mu}$ from Eq.~\eqref{eqn:q parity}, Eq.~\eqref{eqn:decoy parity}, and the constriant that
\begin{equation}
q_{odd}^\mu, q_{even}^\mu, e_{odd}^{Z,\mu}, e_{even}^{Z,\mu}, Y_{odd}^\mu, Y_{even}^\mu \in [0,1],
\end{equation}
then the estimation is too loose to bound the $X$-error $E^X$ in Eq.~\eqref{eqn:EX par}.

Now, we introduce a more efficient method to estimate $q_{odd}$, $q_{even}$, $e_{odd}^Z$ and $e_{even}^Z$. Essentially, we employ the idea of the decoy-state method \cite{Lo2005Decoy}. That is, in Protocol I, Charlie adjusts the intensity $\mu$ of his prepared coherent lights. After Eve's announcement, Charlie announces the value of $\mu$. 
Furthermore, Charlie randomizes the phase $\phi$ on state $\ket{\sqrt{\mu}e^{i\phi}}$ continuously from $[0,2\pi)$. In this case, the state prepared by Charlie can be written as
\begin{equation} \label{eqn:phase random}
\dfrac{1}{2\pi}\int_0^{2\pi} d\phi \ket{\sqrt{\mu} e^{i\phi}} \bra{\sqrt{\mu} e^{i\phi}} = \sum_{k=0}^{\infty} P(k) \ket{k}\bra{k}.
\end{equation}
That is, in Protocol I, if Charlie prepares $\ket{\sqrt{\mu} e^{i\phi}}$ with random phase $\phi$, this is equivalent to preparing the Fock states $\{\ket{k}\}$ with probability $P(k)$. Obviously, Fock states $\{\ket{k}\}$ are parity states. Then, by directly applying Lemma \ref{Lem:parity}, we can estimate the $X$-error by
\begin{equation} \label{eqn:EX Fock}
\begin{aligned}
E^X &= \sum_{k=0}^{\infty} q_{2k+1} e^Z_{2k+1} + \sum_{k=0}^{\infty} q_{2k} (1 - e^Z_{2k}), \\
\end{aligned}
\end{equation}
where $q_0$ is the detection caused by the vacuum signal (i.e.~dark counts) and $e_0^Z = e_0=1/2$ is the vacuum $Z$-error rate.

The source components are Fock states $\{\ket{k}\}$, whose yields $\{Y_k\}$ and $Z$-error rates $\{e^Z_k\}$ are independent of $\mu$. The fractions $q_k^\mu$ of the ``$k$-photon component'' in the final detected signals are given by
\begin{equation}
\label{eqn:q Fock}
q^{\mu}_k = P^\mu(k) \dfrac{Y_k}{Q_\mu}.
\end{equation}
The overall gain and QBER are given by
\begin{equation}
\begin{aligned}
\label{eqn:decoy Fock}
Q_\mu &= \sum_{k=0}^{\infty} P^\mu(k) Y_k, \\
E^Z_\mu Q_\mu &= \sum_{k=0}^{\infty} e_{k}^{Z} P^\mu(k) Y_k. \\
\end{aligned}
\end{equation}
The main idea of the decoy-state method is that Alice and Bob can obtain a set of linear equations in the form of Eqs.~\eqref{eqn:q Fock} and \eqref{eqn:decoy Fock} by using a few values of $\mu$. When an infinite amount of decoy states is used, Alice and Bob can estimate all the parameters $Y_{k}$ and $e_{k}^Z$ accurately, with which they can estimate $q_k^\mu$, $e^Z_k$ and the upper bound of the $X$-error by Eq.~\eqref{eqn:EX Fock}.

To apply the decoy-state method to protocols II and III, we modify the phase randomization requirements accordingly. In the decoy-state version of Protocol III, the phases $\phi_a, \phi_b$ should be randomized independently in $[0, 2\pi)$ rather than $\{0, \pi\}$. Also, there will be an additional phase-sifting condition, $|\phi_a-\phi_b|=0$ or $\pi$. Such random-phase announcement and postselection would not affect the security, with the following reasons, similar to the argument of the equivalence between protocols II and III. In particular, Observation \ref{obs:same rho} in Appendix~\ref{Sc:SecureAnnounce} still holds.

First, we want to argue that the random-phase postselection of $|\phi_a-\phi_b|=0$ or $\pi$ would not affect the security. In fact, any phase postselection would not be affected by Eve's announcement. Note that the random phases $\phi_a, \phi_b$ are determined by Alice and Bob locally. Thus, the sifted qubit pairs, $\rho_{A^\prime B^\prime}^{(n)}$, would be the same in the two cases: 1) Alice and Bob independently randomize the phases, $\phi_a, \phi_b$, and employ this phase postselection $|\phi_a-\phi_b|=0$ or $\pi$; 2) Charlie randomizes the same phases, $\phi_a=\phi_b$, to $A$ and $B$.

Second, we want to argue that the random-phase announcement would not affect the security. Eve announces the detection events before Alice and Bob's random-phase announcement. Thus, her hacking strategy cannot depend on the random phases. Note that with the postselection of $|\phi_a-\phi_b|=0$ or $\pi$, the announced phase $\phi_a$ (or $\phi_b$) can be regarded as the phase reference in Protocol III.

Third, we apply infinite decoy states to estimate $q_k^\mu$ and $e^Z_k$ accurately. In practical implementations, it is interesting to explore if finite decoy states, such as the widely applied vacuum and weak decoy states, are enough to make a valid estimation.

Now, we can modify Protocol III with continuous phase randomization $\phi_a, \phi_b$ and the decoy-state method, namely Protocol IV, as shown in Fig.~\ref{fig:Pro4}(a).

\textbf{\uline{Protocol IV}}
\begin{enumerate}
\item 
State preparation:
Alice and Bob randomly select $\mu_a,\mu_b$ from a given set $\{\mu_0, \mu_1,...\}/2$. They prepare coherent state $\ket{\sqrt{\mu_a}}$ and $\ket{\sqrt{\mu_b}}$ on optical modes $A, B$ separately. They initial their qubits $A^\prime, B^\prime$ in $\ket{+i}$ and independently add random phase $\phi_a, \phi_b\in [0,2\pi)$ on the optical modes $A, B$. Alice applies the control gate $C_{\pi}$, defined in Eq.~\eqref{eq:Cpi}, to qubit $A^\prime$ and optical pulse $A$. Similarly, Bob applies $C_{\pi}$ to $B^\prime$ and $B$.

\item 
Measurement: The two optical pulses $A$ and $B$, are sent to an untrusted party, Eve, who is supposed to perform interference measurement and record which detector ($L$ or $R$) clicks.

\item 
Announcement: Eve announces the detection result, $L$/$R$ click or failure, for each round.

\item 
Sifting: When Eve announces an $L/R$ click, Alice and Bob keep the qubits of systems $A^\prime$ and $B^\prime$. Bob applies a Pauli $Y$-gate to his qubit if Eve's announcement is $R$ click. After Eve's announcement, Alice and Bob announce their encoded intensities $\mu_a, \mu_b$ and phases $\phi_a, \phi_b$. They keep the signal if $\mu_a = \mu_b$ and $|\phi_a-\phi_b|=0$ (or $\pi$). Bob applies a Pauli $Y$-gate to his qubit if $|\phi_a-\phi_b|=\pi$.

\item 
Parameter estimation: After many rounds of the above steps, Alice and Bob end up with a joint $2n$-qubit state, denoted by $\rho_{A^\prime B^\prime}^{(n)}$.  They then perform random sampling on the remaining $\rho_{A^\prime B^\prime}^{(n)}$ to estimate $E^Z$ and infer $E^X$ by Eq.~\eqref{eqn:EX Fock}.

\item 
Key distillation: Alice and Bob apply a standard EDP when the error rates are below a certain threshold. The distillation ratio $r$ is given by Eq.~\eqref{eqn:Shor-Preskill key rate}. Once Alice and Bob obtain $nr$ (almost) pure EPR pairs, they both perform local $Z$ measurements on the qubits to generate private keys.
\end{enumerate}

Finally, we need to reduce the entanglement-based protocol to a prepare-and-measure protocol. Following the Shor-Preskill argument, we move the key measurement in Step \ref{StepExtraction} before the EDP, the parameter estimation and the $C_{\pi}$ gates. That is, they measure the systems $A^\prime$ and $B^\prime$ at the beginning. Therefore, the entanglement-based protocol (Protocol IV) becomes the PM-QKD protocol, as shown in Fig.~\ref{fig:Pro4}(b).

\begin{figure*}[htbp]
\centering
\includegraphics[width=16cm]{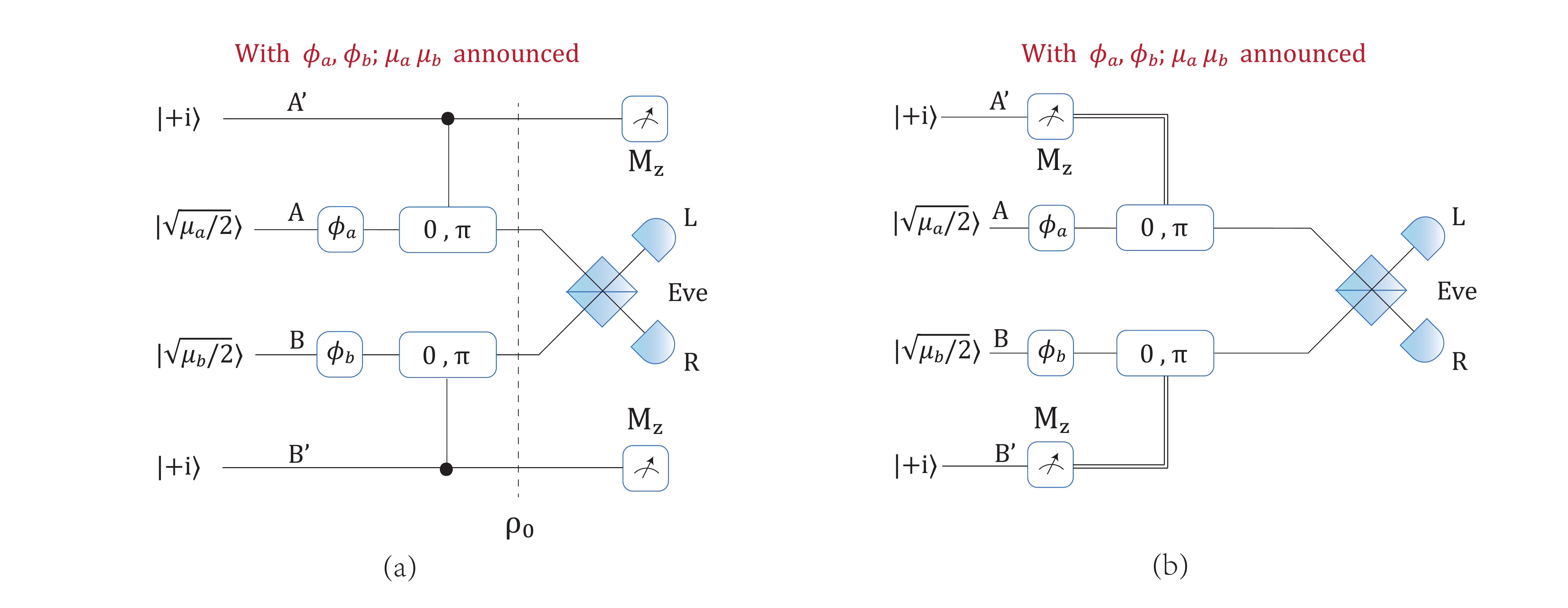}
\caption{(a) Schematic diagram of Protocol IV, where the decoy-state method is applied. The random phase $\phi_a, \phi_b\in [0, 2\pi)$. (b) Schematic diagram of PM-QKD protocol.}
\label{fig:Pro4}
\end{figure*}

\textbf{\uline{PM-QKD protocol}}
\begin{enumerate}
\item 
State preparation:
Alice randomly generates a key bit $\kappa_a\in\{0,1\}$, a random phase $\phi_a\in[0,2\pi)$, and a random intensity $\mu_a\in\{\mu/2, \mu_1/2, \mu_2/2, \cdots\}$. She prepares a coherent-state optical pulse $\ket{\sqrt{\mu_a}e^{i(\phi_a + \pi\kappa_a)}}_A$. Similarly, Bob generates a key bit $\kappa_b$ and prepares $\ket{\sqrt{\mu_b}e^{i(\phi_b+\pi\kappa_b)}}_B$.


\item 
Measurement: The two optical pulses $A$ and $B$ are sent to an untrusted party, Eve, who is supposed to perform interference measurement and record which detector ($L$ or $R$) clicks.

\item 
Announcement: Eve announces the detection result, $L$/$R$ click or failure, for each round.

\item 
Sifting: When Eve announces an $L$/$R$ click, Alice and Bob keep the key bits, $\kappa_a$ and $\kappa_b$. Bob flips his key bit $\kappa_b$ if Eve's announcement is $R$ click. After Eve's announcement, Alice and Bob announce their encoded intensities and phases $\mu_a, \phi_a; \mu_b, \phi_b$. They maintain the signal if $\mu_a = \mu_b$ and $|\phi_a - \phi_b| = 0$ or $\pi$. Bob flips his key bit $\kappa_b$ if $|\phi_a-\phi_b|=\pi$.

\item 
Parameter estimation: After many rounds of the above steps, Alice and Bob end up with joint $n$ pairs of the raw key. They then perform random sampling on this raw key. With Eqs.~\eqref{eqn:q Fock} and \eqref{eqn:decoy Fock}, they estimate $q_{k}$ and $e^Z_{k}$.

\item 
Key distillation: Alice and Bob apply classical communication for error correction and privacy amplification. The key distillation ratio $r$ is given by Eq.~\eqref{eqn:Shor-Preskill key rate}.
\end{enumerate}

Note that the PM-QKD protocol mentioned above is the exact PM-QKD protocol mentioned in the main text with the decoy-state method contained.

Here, in Protocol IV and the PM-QKD protocol, there are some unrealistic assumptions remaining.
\begin{enumerate}
\item
The phase of Alice and Bob's coherent state is locked; that is, the phase reference of Alice's and Bob's coherent state is the same.
\item
The phase-sifting condition, $\phi_a = \phi_b$ can be satisfied with an acceptable probability.
\end{enumerate}
We will discuss how to remove these two requirements in Appendix~\ref{Sc:SecurePost}, and, hence, make the protocol practical.

\subsection{Phase PostCompensation}\label{Sc:SecurePost}
Here, we modify the PM-QKD protocol to remove the requirement of phase locking between Alice and Bob, and also relax the postselection condition of $|\phi_a-\phi_b|=0$ or $\pi$. The method is shown in Fig.~\ref{fig:phasesel}.

First, we relax the postselection condition of $|\phi_a-\phi_b|=0$ or $\pi$ by dividing the phase interval $[0,2\pi)$ into $M$ slices $\{\Delta_j\}$ for $0\le j \le M-1$, where $\Delta_j = [2\pi j/M, 2\pi (j+1)/M)$. We change the phase postselection condition of $\phi_a = \phi_b$ to slice postselection $j_a=j_b$. Here $j_a, j_b$ are defined as
\begin{equation} \label{eq:jajb}
\begin{aligned}
j_a &= [\frac{\phi_a}{2\pi} M] \mod M, \\
j_b &= [\frac{\phi_b}{2\pi} M] \mod M,
\end{aligned}
\end{equation}
where $[\cdot]$ is the round function to output the nearest integer. That is, Alice announces the $j_a$ th slice $\Delta_{j_a}$, which her random phase $\phi_a$ falls into. Similarly, Bob announces $j_b$. The sifting condition is $|j_b - j_a| = 0$ or $M/2$. This technique is first used in the phase-encoding MDI-QKD \cite{Ma2012Alternative}. Announcing the phase slices $j_a$ and $j_b$ rather than the exact phase $\phi_a$ and $\phi_b$ will only leak less information to Eve. Thus, all the security-proof results from previous subsections apply here. Applying Eq.~\eqref{eqn:Shor-Preskill key rate}, we also need to add a phase-sifting factor $2/M$. Therefore, the final key rate is given by
\begin{equation} \label{eq:keyRateapp}
\begin{aligned}
R_{PM} &\ge \frac{2}{M} Q_\mu [ 1 - f H(E_{\mu}^Z) - H(E_{\mu}^X) ], \\
\end{aligned}
\end{equation}
where $Q_\mu$ is the total gain of the pulses, $E_{\mu}^Z$ is the overall QBER, the phase error rate $E_{\mu}^X$ is given by Eq.~\eqref{eqn:EX Fock}, and $f$ is the error correction efficiency. We need to point out that, when $M$ is too small, the misalignment caused by phase slices $\Delta_{j_a}$ and $\Delta_{j_b}$ is big and results in a large QBER $E_{\mu}^Z$.

Second, we remove the requirement of phase locking. In the phase postselection step, either Alice or Bob announces that the random phase is sufficient. Without loss of generality, we assume that Alice announces the phase and Bob performs the sifting of $|j_b - j_a| = 0$ or $M/2$. With less information announced, the key rate still holds.

Suppose that the difference between Alice's and Bob's phase references, denoted by $\phi_0\in[0,2\pi)$, is fixed but unknown. In sifting, Bob needs to figure out the value of $\phi_0$. Define an offset, where $0\le j_0\le M-1$, as
\begin{equation} \label{eq:j0}
\begin{aligned}
j_0 \equiv [\frac{\phi_0}{2\pi} M] \mod M,
\end{aligned}
\end{equation}
During sifting, Bob can estimate the offset $j_0$ for each pulse. The estimation accuracy would not affect the security. In fact, in the security proofs, we assume that Eve knows the phase references ahead. Suppose Bob compensates the offset by $j_d$, and he has the freedom to choose $j_d$ from $\{0,1,\cdots,M-1\}$. In a practical scenario, normally the phase references (and, hence, the offset $j_0$) change slowly with time. Then, Bob can figure out the proper phase compensation offset $j_d$ by minimizing the QBER from random sampling, as shown in Fig.~\ref{fig:phasesel}.

For the case where the phase reference difference $\phi_0$ varies, we can treat the changing caused by Eve. Such deviation will introduce bit errors, but not help Eve learn key information, since the variation of $\phi_0$ is independent of the key $\kappa_a, \kappa_b$. Note that in the security proof, we assume that Eve knows the phase references accurately.

With all the modifications above, we propose the following practical version of PM-QKD protocol.

\textbf{\uline{PM-QKD protocol with phase postcompensation}}
\begin{enumerate}
\item
State preparation: Alice randomly generates a key bit $\kappa_a\in\{0,1\}$, a random phase $\phi_a \in[0,2\pi)$, and a random intensity $\mu_a\in\{\mu/2, \mu_1/2, \mu_2/2, \cdots\}$. She prepares a coherent-state optical pulse $\ket{\sqrt{\mu_a}e^{i(\phi_a + \pi\kappa_a)}}_A$. Similarly, Bob generates a key bit $\kappa_b$ and prepares $\ket{\sqrt{\mu_b}e^{i(\phi_b + \pi\kappa_b)}}_B$.

\item
Measurement: The optical pulses, systems $A$ and $B$, are sent to an untrusted party, Eve, who is supposed to perform interference measure and record which detector ($L$ or $R$) clicks.

\item
Announcement: Eve announces the detection results, $L$/$R$ click or failure, for each round.

\item
Sifting: When Eve announces an $L$/$R$ click, Alice and Bob keep the key bits, $\kappa_a$ and $\kappa_b$. Bob flips his key bit $\kappa_b$ if Eve's announcement is $R$ click. After Eve's announcement, Alice and Bob announce their encoded intensities $\mu_a$ and $\mu_b$, respectively. They keep the bits if $\mu_a = \mu_b$.

\item
Parameter estimation:
Alice and Bob run the above procedures many times and then run the following procedures.

\begin{enumerate}
\item
For each bit, Alice announces the phase slice index $j_a$, given in Eq.~\eqref{eq:jajb}, and she randomly samples a certain amount of key bits and announces them for QBER testing.

\item
In the phase postcompensation method, given an offset compensation $j_d\in \{0,1,...,M/2 -1\}$, Bob sifts the sampled bits with the phase postselection condition $|j_b + j_d - j_a| \mod M = 0$ or $M/2$. For the case of $M/2$, Bob flips the key bit $\kappa_b$. After sifting, Bob calculates the QBER $E^Z$ with Alice's sampling key bits. Bob tries all possible $j_d\in\{0,1,\cdots,M-1\}$, and figures out the proper $j_d$ to minimize the sampling QBER. Using the phase postselection condition with the proper $j_d$, Bob sifts (and flips if needed) the unsampled bits and announces the locations to Alice. Alice sifts her key bits accordingly.

\item
Alice and Bob analyze the overall gain $Q_{\mu_i}$ and QBER $E^Z_{\mu_i}$ for different values of intensities $\mu_a =\mu_b =\mu_i/2$. They estimate the phase error rate $E^{(X)}_{\mu}$ by Eq.~\eqref{eqn:EX Fock}.
\end{enumerate}

\item
Key distillation: Alice and Bob apply error correction and privacy amplification. The key rate $r$  is given by Eq.~\eqref{eq:keyRateapp}.
\end{enumerate}

Here, in the phase postcompensation, Bob does not need to fix a $j_d$ for the whole experiment. Instead, he can adjust $j_d$ in real time. Given the total number of slices $M$, define the phase-stable time, $\tau_M$, to be the time period during which the phase fluctuation is smaller than $\pi/M$. Roughly speaking, Bob adjusts $j_d$ for every $\tau_M$. In practice, Bob would randomly sample some detections for tests. Here, in order to obtain a good estimation of $j_d$, Bob should sample enough detections for the offset test within $\tau_M$. Then, the phase postcompensation method would put a requirement on the detection rate.

Now, we consider a practical example. Considering the slice number $M=16$, the tolerable fluctuation should be about $\pi/M=0.196$ rad. The phase of a continuous-wave (CW) laser pulse usually fluctuates randomly, and its behavior can be modeled as a random walk. In an experimental work of testing the phase drift for a 36.5-km Mach-Zehnder interferometer \cite{minavr2008phase}, within a time duration of $0.2$ ms, the mean phase fluctuation is about $0.15$ rad. In the data presented in recent TF-QKD \cite{Lucamarini2018TF}, via the total transmission distance around $200$ km, the phase drift rate is about $3$ rad ms$^{-1}$. For a time period of $0.05$ ms, the mean phase fluctuation is about $0.15$ rad, which is less than $\pi/M$. Then, in this case, one can set $\tau_M=0.05$ ms for testing. Considering a GHz QKD system, there are $5\times10^4$ signals sent out within $\tau_M$. Using the same parameters for a longer transmission distance of $100$ km from Alice to Eve (same distance for Bob) using standard telecom fiber (0.2 dB/km), the transmission loss is $\eta = 20\,dB = 10^{-2}$, and there will be $500$ data left for post-processing, which is sufficient for sampling tests. From here, one can see that our phase postcompensation method is feasible with current technology. For a longer transmission distance, the phase postcompensation method may not be sufficient. One can use the phase calibration method as an alternative, as introduced in the main text.

\section{Model and Simulation} \label{Sc:SimuPara}
Here, we show the details for the simulations of different QKD schemes. We mainly follow the simulation model used in the literature \cite{Ma2005Practical,Ma2008PhD}. We calculate the detection probabilities of PM-QKD in Appendix~\ref{Sc:DetProb}, and then evaluate the yields, gains, and error rates in Appendix~\ref{Sc:GainQBER}. In simulations, we also compare the performance of PM-QKD with the decoy-state BB84 \cite{Bennett1984Quantum}, MDI-QKD \cite{Lo2012Measurement}, and the linear key-rate bound \cite{Pirandola2017Fundamental}. We list the used formulas and simulation parameters in Appendix~\ref{Sc:Others}.

\subsection{Detection probabilities} \label{Sc:DetProb}
The channel is assumed to be a pure loss one and symmetric for Alice and Bob with transmittance $\eta$ (with detector efficiency $\eta_d$ taken into account). We suppose that the lights from Alice and Bob faithfully interfere and are measured with dark-count rate $p_d$. Without loss of generality, we only consider the case where $\kappa_a = \kappa_b = 0$, $j_a = j_b = 0$.

In PM-QKD, the global phases $\phi_a, \phi_b$ are divided into $M$ slices. Without phase locking, the phase references for Alice and Bob can differ. Therefore, even if the announced slices meet, $|j_a - j_b| = 0$, there may be a considerable difference $\phi_\delta \equiv \phi_b-\phi_a$ between the two global phases $\phi_a, \phi_b$.

First, for a fixed $\phi_\delta$, we calculate the detection probability in the single-photon case. After Alice and Bob's encoding, the state becomes
\begin{equation}
(e^{i\phi_a}a^\dagger + e^{i\phi_b}b^\dagger)\ket{00}_{A0,B0} = (a^\dagger + e^{i\phi_\delta}b^\dagger)\ket{0}.
\end{equation}
Here, the transmittance $\eta$ includes channel losses and detection efficiencies, which transfer $a^\dagger,b^\dagger$ according to
\begin{equation}
\begin{aligned}
a^\dagger &\rightarrow \sqrt{\eta}a^\dagger + \sqrt{1-\eta}s^\dagger, \\
b^\dagger &\rightarrow \sqrt{\eta}b^\dagger + \sqrt{1-\eta}t^\dagger,
\end{aligned}
\end{equation}
where $s^\dagger, t^\dagger$ are the modes coupled to the environment. Before Eve's interference, the state is
\begin{equation}
(\sqrt{\eta}(a^\dagger + e^{i\phi_\delta}b^\dagger) + \sqrt{1 - \eta}(s^\dagger + e^{i\phi_\delta}t^\dagger))\ket{0},
\end{equation}
and after Eve's interference, the state becomes
\begin{equation}
\{\sqrt{\dfrac{\eta}{2}}[ (1+e^{i\phi_\delta})l^\dagger + (1-e^{i\phi_\delta})r^\dagger ] + \sqrt{1 - \eta}(s^\dagger + e^{i\phi_\delta}t^\dagger) \}\ket{0},
\end{equation}
where $l^\dagger$ and $r^\dagger$ are the creation operators for the modes to $L$- and $R$-detectors, respectively. Then, the detection probabilities are given by
\begin{equation} \label{eq:p1}
\begin{aligned}
p_0^{(1)} &= 1 - \eta, \\
p_l^{(1)} &= \eta \cos^2(\dfrac{\phi_\delta}{2}), \\
p_r^{(1)} &= \eta \sin^2(\dfrac{\phi_\delta}{2}), \\
p_{lr}^{(1)} &= 0,
\end{aligned}
\end{equation}
where $p_0^{(1)}$, $p_l^{(1)}$, $p_r^{(1)}$, and $p_{lr}^{(1)}$ are, respectively, the probabilities for no click, $L$ click, $R$ click, and double click for the single-photon case.

Then, we consider the $k$-photon case, for which we regard as $k$ identical and independent click events of the single-photon case; the click probabilities are given by
\begin{equation}\label{eq:pk}
\begin{aligned}
p_0^{(k)} &= (p_0^{(1)})^k,\\
p_l^{(k)} &= (p_0^{(1)} + p_l^{(1)})^k - (p_0^{(1)})^k, \\
p_r^{(k)} &= (p_0^{(1)} + p_r^{(1)})^k - (p_0^{(1)})^k, \\
p_{lr}^{(k)} &= 1 - p_0^{(k)} - p_l^{(k)} - p_r^{(k)} \\
&= 1 + (p_0^{(1)})^k -(p_0^{(1)} + p_l^{(1)})^k -  (p_0^{(1)} + p_r^{(1)})^k,
\end{aligned}
\end{equation}
where $p_0^{(k)}$, $p_l^{(k)}$, $p_r^{(k)}$, and $p_{lr}^{(k)}$ are, respectively, the probabilities for no-click, $L$ click, $R$ click, and double click for the $k$-photon case.

Now, we take into account the effects caused by the detector dark counts $p_d$. Since the dark counts are independent of the photon-click events, we can draw a table for all the cases, as shown in Table \ref{tab:darkprob}. The final click probabilities for the $k$-photon case are given by
\begin{equation}\label{eq:Pk}
\begin{aligned}
P_0^{(k)} &= (1-p_d)^2 p_0^{(k)}, \\
P_L^{(k)} &= (1-p_d)^2 p_l^{(k)} + p_d(1-p_d)(p_0^{(k)}+p_l^{(k)}) \\
&= p_d(1-p_d)p_0^{(k)}+(1-p_d) p_l^{(k)}, \\
P_R^{(k)} &= (1-p_d)^2 p_r^{(k)} + p_d(1-p_d)(p_0^{(k)}+p_r^{(k)}) \\
&= p_d(1-p_d)p_0^{(k)}+(1-p_d) p_r^{(k)}, \\
P_{LR}^{(k)} &= (1-p_d)^2 p_{lr}^{(k)} + p_d(1-p_d)(p_l^{(k)} + p_r^{(k)} + 2 p_{lr}^{(k)}) + p_d^2 \\
&= (1-p_d^2) p_{lr}^{(k)} + p_d(1-p_d)(p_l^{(k)} + p_r^{(k)}) + p_d^2.
\end{aligned}
\end{equation}

\begin{table*}[htbp]
\begin{tabular}{cc|cccc}

\hline
	\multicolumn{2}{c|}{\textbf{Dark count condition}} & \multicolumn{4}{c}{\textbf{Contributions to the overall click probabilities}} \\

\hline
$L$-dark count & $R$-dark count & No-click $P^{(k)}_0$ & $L$ click $P^{(k)}_L$ & $R$ click $P^{(k)}_R$ & Double-click $P^{(k)}_{LR}$ \\
\hline

$(1-p_d)$ & $(1-p_d)$ & $p^{(k)}_0$ & $p^{(k)}_l$ & $p^{(k)}_r$ & $p^{(k)}_{lr}$ \\

$(1-p_d)$ & $p_d$ & 0 & 0 & $p^{(k)}_0 + p^{(k)}_r$ & $p^{(k)}_l + p^{(k)}_{lr}$ \\

$p_d$ & $(1-p_d)$ & 0 & $p^{(k)}_0 + p^{(k)}_l$ & 0 & $p^{(k)}_r + p^{(k)}_{lr}$ \\

$p_d$ & $p_d$ & 0 & 0 & 0 & 1 \\
\hline
\end{tabular}
\caption{Probability table of clicks with dark counts present.} \label{tab:darkprob}
\end{table*}

Similarly, we can derive the formulas for detection probabilities with coherent state inputs. The states sent out by Alice and Bob are $\ket{\sqrt{\mu_a} e^{i\phi_a}}$ and $\ket{\sqrt{\mu_b} e^{i\phi_b}}$, respectively, where $\mu_a=\mu_b=\mu/2$. After channel and detection losses, the states become $\ket{\sqrt{\eta\mu_a} e^{i\phi_a}}$ and $\ket{\sqrt{\eta\mu_b} e^{i\phi_b}}$. By going through the BS, the states become
\begin{equation}
\begin{aligned}
\ket{\alpha_L} & = \ket{\dfrac{\sqrt{\eta\mu}}{2} (e^{i\phi_a} + e^{i\phi_b})}= \ket{\dfrac{\sqrt{\eta\mu}}{2} e^{i\phi_a}(1 + e^{i\phi_\delta})  }, \\
\ket{\alpha_R} & = \ket{\dfrac{\sqrt{\eta\mu}}{2} (e^{i\phi_a} - e^{i\phi_b})}= \ket{\dfrac{\sqrt{\eta\mu}}{2} e^{i\phi_a}(1 - e^{i\phi_\delta})  }.  \\
\end{aligned}
\end{equation}
Then, the detection click probabilities are
\begin{equation} \label{eq:Pmu}
\begin{aligned}
P_\mu(\bar{L}) &= (1 - p_d)\exp{(-|\alpha_L|^2)} \\
&= (1 - p_d) \exp{(\, -\eta\mu\, cos^2(\dfrac{\phi_\delta}{2}) \,)}, \\
P_\mu(L)&= 1-P_\mu(\bar{L}), \\
P_\mu(\bar{R}) &= ( 1 - p_d )\exp{(-|\alpha_R|^2)} \\
&= (1 - p_d) \exp{(\, -\eta\mu\, sin^2(\dfrac{\phi_\delta}{2}) \,)}, \\
P_\mu(R) &= 1-P_\mu(\bar{R}),
\end{aligned}
\end{equation}
where $P_\mu(L)$ and $P_\mu(\bar{L})$ are the probabilities of the $L$ click and no $L$ click, respectively, and $P_\mu(R)$ and $P_\mu(\bar{R})$ are for the $R$-detector. These probabilities are similar to Eq.~\eqref{eq:Pk}. The difference is that the probabilities in Eq.~\eqref{eq:Pk} are mutually exclusive, while in Eq.~\eqref{eq:Pmu}, the probabilities $P_\mu(L)$ and $P_\mu(R)$ are independent.

All the above probability formulas are functions of the phase difference $\phi_\delta$. In the simulation, one needs to integrate $\phi_\delta$ over its probability distribution $f_{\phi_\delta}(\phi)$. Recall that the phase reference deviation between Alice and Bob is $\phi_0$. Here, with the phase postcompensation method introduced in Appendix~\ref{Sc:SecurePost}, we assume that $\phi_0$ is uniformly distributed in $[-\pi /M,\pi /M)$, denoted by
\begin{equation}
\phi_0 \sim U[-\pi /M,\pi /M).
\end{equation}
In this case, $\phi_a$ and $\phi_b$ are uniformly distributed in $[0,2\pi/M)$ and $[\phi_0, 2\pi/M + \phi_0)$, respectively,
\begin{equation}
\begin{aligned}
\phi_a &\sim U[0, 2\pi/M), \\
\phi_b &\sim U[\phi_0, 2\pi/M + \phi_0).
\end{aligned}
\end{equation}
For a fixed $\phi_0$, the probability distribution of the phase difference $\phi_\delta$ is given by
\begin{equation} \label{eq:fdelta}
f_{\phi_\delta}^{(\phi_0)}(\phi) =
\begin{cases}
(\frac{M}{2\pi})^2 [\phi + (\frac{2\pi}{M} - \phi_0)] \quad &\phi\in [\phi_0 - \frac{2\pi}{M}, \phi_0) ,\\
(\frac{M}{2\pi})^2 [-\phi + (\frac{2\pi}{M} + \phi_0)] \quad &\phi\in [\phi_0 , \phi_0 + \frac{2\pi}{M}) ,\\
0 & \textit{otherwise}.
\end{cases}
\end{equation}

Furthermore, in the simulation, one needs to take another integration of $\phi_0$ over $[-\pi /M,\pi /M)$ to get the yields, gain, and error rates.

\subsection{Yields, gain, and error rates} \label{Sc:GainQBER}
We first analyze the yield $\{Y_k\}$ for $k\ge0$ and the total gain $Q_\mu$ according to the probability formulas given in Appendix~\ref{Sc:DetProb}. Note that, in PM-QKD, both no-click and double-click events are regarded as failure detection. The yield $Y_k$ of the $k$-photon state is given by Eqs.~\eqref{eq:p1}, \eqref{eq:pk}, and \eqref{eq:Pk},
\begin{widetext}
\begin{equation} \label{eq:Yk}
\begin{aligned}
Y_k &= P_L^{(k)} + P_R^{(k)} \\
&= (1-p_d) (p_l^{(k)} + p_r^{(k)}) + 2p_d(1-p_d) p_0^{(k)} \\
 &= (1-p_d) [ (1-\eta\, \sin^2(\dfrac{\phi_\delta}{2}))^k + (1-\eta\, \cos^2(\dfrac{\phi_\delta}{2}))^k ] - 2(1-p_d)^2 (1-\eta)^k \\
 &\approx ( 1 - p_d) ( 1 + (1 - \eta)^k ) - 2( 1 - p_d )^2 ( 1 - \eta )^k \\
 &= (1-p_d) [ 1 - (1 - 2p_d)(1 -\eta)^k ] \\
 &\approx 1 - (1 - 2p_d)(1 -\eta)^k. \\
\end{aligned}
\end{equation}
\end{widetext}

Here, in the first approximation, we take the first order of a small $\phi_\delta$ and ignore $\sin^2(\dfrac{\phi_\delta}{2})=0$. In the second approximation, we omit the higher-order term $p_d(1-(1-2p_d)(1-\eta)^k)=O(p_d^2+p_d\eta)$, since normally $p_d$ is small.

The total gain $Q_\mu$ is given by Eq.~\eqref{eq:Pmu},
\begin{widetext}
\begin{equation} \label{eq:Qmu}
\begin{aligned}
Q_\mu &= P(L)P(\bar{R}) + P(\bar{L})P(R) \\
  &= (1 - p_d) \exp{(\, -\eta\mu\, \sin^2(\dfrac{\phi_\delta}{2}) \,)}[1 - (1 - p_d) \exp{(\, -\eta\mu\, \cos^2(\dfrac{\phi_\delta}{2}) \,)} ] \\&\quad\quad + (1 - p_d) \exp{(\, -\eta\mu\, \cos^2(\dfrac{\phi_\delta}{2}) \,)}[1 - (1 - p_d) \exp{(\, -\eta\mu\, \sin^2(\dfrac{\phi_\delta}{2}) \,)} ] \\
  &\approx  (1-p_d)[1 - (1-p_d)\exp{(-\eta\mu_b)}] + p_d(1-p_d)\exp{(-\eta\mu_b)} \\
  &= (1-p_d) [ 1 - (1 - 2p_d)e^{-\eta\mu} ] \\
  &\approx 1 -e^{-\eta\mu} +2p_de^{-\eta\mu},
\end{aligned}
\end{equation}
\end{widetext}
where the two approximations are the same as the ones used in Eq.~\eqref{eq:Yk}. Since $\phi_\delta$ does not affect the yields $Y_k$ and gain $Q_\mu$ much with the first-order approximation, the average yields $Y_k$ and gain $Q_\mu$ over $\phi_\delta\in[-\pi /M,\pi /M)$ can be regarded as the ones when $\phi_\delta=0$. Also, the results of Eqs.~\eqref{eq:Yk} and \eqref{eq:Qmu} are consistent with the ones presented in the regular QKD model \cite{Ma2005Practical}, as shown in Eqs.~\eqref{eq:YemuBB84} and \eqref{eq:QEmuBB84}.

Then, we calculate the bit error rate for the $k$-photon signal $e^Z_k$ and the QBER $E^Z_\mu$ with coherent states input $\mu_a = \mu_b = \mu/2$. The bit error rate $e^Z_k(\phi_\delta)$ is given by
\begin{equation} \label{eq:ek}
e^Z_k(\phi_\delta) = \dfrac{P_R^{(k)}}{P_L^{(k)} + P_{R}^{(k)}},
\end{equation}
where $P_L^{(k)}$ and $P_{R}^{(k)}$ are given by Eqs.~\eqref{eq:p1}, \eqref{eq:pk}, and \eqref{eq:Pk}.
The average bit error rate of $e^Z_k$ over $\phi_\delta$ is given by the following integral,
\begin{equation} \label{eqn:eZkint}
e^Z_k =\dfrac{M}{2\pi} \int^{\pi/M}_{-\pi/M} d\phi_0 \int^{3\pi/M}_{-3\pi/M} d\phi \; f_{\phi_\delta}^{(\phi_0)} (\phi) e^Z_k(\phi), \\
\end{equation}
where $f_{\phi_\delta}^{(\phi_0)} (\phi)$ is given by Eq.~\eqref{eq:fdelta}. For the case $k=1$, we explicitly calculate the error rate $e^Z_1$, given by Eqs.~\eqref{eq:p1}, \eqref{eq:pk}, \eqref{eq:Pk}, and \eqref{eq:ek},
\begin{equation}
\begin{aligned} \label{eqn:eZ1}
e^Z_1(\phi_\delta) &= \dfrac{(1-p_d)^2 p_l^{(1)} + p_d(1-p_d)(p_0^{(1)}+p_l^{(1)})}{(1-p_d)^2 \eta + p_d(1-p_d)(2-\eta)} \\
&= \dfrac{\sin^2(\phi_\delta/2) \eta + p_d(1-\eta)}{\eta+2p_d(1-\eta)}, \\
\end{aligned}
\end{equation}
and integrate Eq.~\eqref{eqn:eZ1} with Eq.~\eqref{eqn:eZkint},
\begin{equation} \label{eq:e1}
\begin{aligned}
e^Z_1 &=\dfrac{M}{2\pi} \int^{\pi/M}_{-\pi/M} d\phi_0 \int^{3\pi/M}_{-3\pi/M} d\phi\; f_{\phi_\delta}^{(\phi_0)} (\phi) e^Z_1(\phi) \\
&= \dfrac{ e_{\delta} \eta + p_d(1-\eta)}{\eta+2p_d(1-\eta)}, \\
\end{aligned}
\end{equation}
\vskip 8mm
where
\begin{equation} \label{eq:edelta}
\begin{aligned}
e_{\delta} = \frac{\pi }{M}-\frac{M^2 }{\pi ^2}\sin ^3\left(\frac{\pi }{M}\right)
\end{aligned}
\end{equation}
can be regarded as the misalignment error rate. Equation~\eqref{eq:e1} can be understood that, if a signal causes a click, the error rate is $e_{\delta}$; and if a dark count causes a click, the error rate is $e_0=1/2$. Since the double-click events are discarded, when both a signal and a dark count cause clicks, the error rate is $e_{\delta}$. Then, we can approximate $e^Z_k$ with the same spirit, for the contributions in $Y_k$, given in Eq.~\eqref{eq:Yk},
\begin{equation} \label{eq:ekapp}
\begin{aligned}
e^Z_k &\approx \frac{p_d(1 -\eta)^k+e_\delta[1 -(1 -\eta)^k]}{Y_k}. \\
\end{aligned}
\end{equation}
Here, note that in Eq.~\eqref{eq:ekapp}, we ignore the double clicks caused by multiphoton signals, which would further reduce the misalignment error and, hence, the value of $e^Z_k$. The QBER $E^Z_\mu(\phi_\delta)$ for a given $\phi_\delta$ is given by substituting Eq.~\eqref{eq:Pmu},
\begin{widetext}
\begin{equation}\label{eq:EZ}
\begin{aligned}
E^Z_\mu(\phi_\delta) &= \dfrac{P(\bar{L})P(R)}{P(\bar{L})P(R) + P(L)P(\bar{R})} \\
&= \frac{1}{Q_\mu}{(1 - p_d) \exp{(\, -\eta\mu\, \cos^2(\dfrac{\phi_\delta}{2}))}\{1-(1 - p_d) \exp[-\eta\mu \sin^2(\frac{\phi_\delta}{2})]\}} \\
&= \frac{1}{Q_\mu} (1-p_d) \{\exp{[\, -\eta\mu\, \cos^2(\dfrac{\phi_\delta}{2})]}-(1 - p_d) e^{-\eta\mu}\} \\
&\approx \frac{e^{-\eta\mu}}{Q_\mu} [p_d + \eta\mu\sin^2(\frac{\phi_\delta}{2})] \\
\end{aligned}
\end{equation}
\end{widetext}
similar to Eq.~\eqref{eq:e1},
\begin{equation} \label{eq:Emusimu}
\begin{aligned}
E^Z_\mu &= \dfrac{M}{2\pi}\int^{\pi/M}_{-\pi/M} d\phi_0 \int^{3\pi/M}_{-3\pi/M} d\phi\; f_{\phi_\delta}^{(\phi_0)} (\phi) E^Z_\mu(\phi_\delta) \\
&\approx \frac{(p_d + \eta\mu e_\delta)e^{-\eta\mu}}{Q_\mu},
\end{aligned}
\end{equation}
where $e_\delta$ is given in Eq.~\eqref{eq:edelta}. The results of Eqs.~\eqref{eq:Emusimu} and \eqref{eq:e1} are consistent with the one presented in the regular QKD model \cite{Ma2005Practical}, as shown in Eqs.~\eqref{eq:YemuBB84} and \eqref{eq:QEmuBB84}, with a slight difference. The difference is caused by how the double-click events are processed. In BB84, Alice and Bob need to randomly assign a bit to double clicks, while in PM-QKD, double clicks are discarded.

Now, we can evaluate the key rate with the above model. Let us restate the key-rate formula, Eq.~\eqref{eq:keyRateapp},
\begin{equation}
R_{PM} \ge \frac{2}{M} Q_\mu [ 1 - f H(E_{\mu}^Z) - H(E_{\mu}^X) ],
\end{equation}
where the phase error rate is given by Eq.~\eqref{eqn:EX Fock},
\begin{equation} \label{eq:EmuX}
\begin{aligned}
E_{\mu}^X &= \sum_{k=0}^{\infty} q_{2k+1} e^Z_{2k+1} + \sum_{k=0}^{\infty} q_{2k} (1 - e^Z_{2k}), \\
&\le q_0 e^Z_0 +\sum_{k=0}^\infty e_{2k+1}^Z q_{2k+1} + (1 - q_0 - q_{odd}).
\end{aligned}
\end{equation}
In simulation, $Q_\mu$ and $E_{\mu}^Z$ are given by Eqs.~\eqref{eq:Qmu} and \eqref{eq:Emusimu}. The fractions, $\{q_k\}$ with $k\ge0$, of different photon components $k$ contributing to the valid detections ($L/R$ clicks) are given by
\begin{equation} \label{eq:q}
\begin{aligned}
q_k = \dfrac{P(k)Y_k}{Q_\mu} = e^{-\mu}\dfrac{\mu^k}{k!}\dfrac{Y_k}{Q_\mu}.
\end{aligned}
\end{equation}
For the odd-photon-number component, we have
\begin{equation} \label{eq:qevenodd}
\begin{aligned}
q_{odd} &= \sum_{k=0}^{\infty}  q_{2k+1}  \\
&= \frac{1}{Q_\mu}\sum_{k=0}^{\infty} \frac{Y_{2k+1} \mu^{2k+1} e^{-\mu}}{(2k+1)!} \\
&= \frac{e^{-\mu}}{Q_\mu}\sum_{k=0}^{\infty} \{ \dfrac{\mu^{2k+1}}{(2k+1)!} - (1-2p_d)\dfrac{[(1-\eta)\mu]^{2k+1}}{(2k+1)!} \} \\
&= \frac{e^{-\mu}}{Q_\mu}\{\sinh(\mu) - (1-2p_d)\sinh[(1-\eta)\mu] \}.
\end{aligned}
\end{equation}

To simplify the simulation, we explicitly calculate the $e^Z_k$ and $q_k$ for $0\le k\le 5$. Further zooming into Eq.~\eqref{eq:EmuX}, we have
\begin{equation} \label{eq:EmuXzoom}
\begin{aligned}
E_{\mu}^X &= \sum_{k=0}^{\infty} q_{2k+1} e^Z_{2k+1} + \sum_{k=0}^{\infty} q_{2k} (1 - e^Z_{2k}), \\
&\le q_0 e^Z_0 +(q_1 e^Z_1 + q_3 e^Z_3 + q_5 e^Z_5) + (1 - q_0 - q_1 - q_3 -q_5),
\end{aligned}
\end{equation}
where $e_k^Z$ and $q_k$ are given by Eqs.~\eqref{eq:ek} and \eqref{eq:q}, respectively.

Here, we attach the MALTAB code of PM-QKD key rate for reference.
\twocolumngrid
\begin{lstlisting}[language={Matlab}]
function [RPM] = PMKey(eta,mu,pd,M,Ef)

% k-photon yields Yk by Eq.(B13)
Y0 = 2*pd;
Y1 = 1 - (1 - 2*pd).*(1 - eta);
Y3 = 1 - (1 - 2*pd).*(1 - eta)^3;
Y5 = 1 - (1 - 2*pd).*(1 - eta)^5;

% total gain Q by Eq.(B14)
Qmu = 1 - (1 - 2*pd).*exp(-mu.*eta);

% k-photon clicked fractions by Eq.(B25)
q0 = Y0.*exp(-mu)./Qmu;
q1 = Y1.*mu.*exp(-mu)./Qmu;
q3 = Y3.*(mu.^3).*exp(-mu)./(factorial(3)*Qmu);
q5 = Y5.*(mu.^5).*exp(-mu)./(factorial(5)*Qmu);

% k-photon bit error rate by Eq.(B20)
e0 = 0.5;
edelta = pi/M - (M/pi)^2 * (sin(pi/M))^3;
e1Z = ( pd*(1-eta)   + edelta*(1 - (1-eta)  ) )/Y1;
e3Z = ( pd*(1-eta)^3 + edelta*(1 - (1-eta)^3) )/Y3;
e5Z = ( pd*(1-eta)^5 + edelta*(1 - (1-eta)^5) )/Y5;

% QBER by Eq.(B22)
% EZPM = (2*pi/M) * ( pd + eta*mu*edelta ).*exp(-eta*mu)./Qmu; % original errneous code
EZPM = ( pd + eta*mu*edelta ).*exp(-eta*mu)./Qmu;

% phase error rate EX by Eq.(B27)
EX = q0.*e0 + q1.*e1Z + q3.*e3Z + q5.*e5Z + (1 - q0 - q1 - q3 - q5);
EX = min(EX,0.5);

% key rate by Eq.(B23)
rPM = -Ef.*h(EZPM) + 1 - h(EX);
RPM = (2/M).*Qmu.*rPM;
RPM = max(RPM,0);
end

function [entropy] = h(p)
entropy = -p.*log2(p)-(1-p).*log2(1-p);
entropy(p<=0 | p>=1) = 0;
end
\end{lstlisting}

\subsection{Simulation formulas for BB84 and MDI-QKD protocol and simulation parameters} \label{Sc:Others}
We compare our derived key rate with that of the prepare-and-measure BB84 \cite{Bennett1984Quantum} protocol, whose key rate is given by the GLLP-decoy method \cite{gottesman04}.

The key rate of the decoy-state BB84 protocol is given by \cite{Lo2005Decoy}
\begin{equation}
R_{BB84} = \dfrac{1}{2} Q_\mu\{- f H(E_{\mu}^Z) + q_1[1 - H(e^X_1)] \}, \\
\end{equation}
where $1/2$ is the basis sifting factor.

In the simulation, the yield and error rates of the $k$-photon component are given by \cite{Ma2008PhD}
\begin{equation} \label{eq:YemuBB84}
\begin{aligned}
Y_k &= 1-(1-Y_0)(1-\eta)^{k}, \\
e_k &= e_d + \frac{(e_0-e_d)Y_0}{Y_k}, \\
\end{aligned}
\end{equation}
where $e_d$ is the intrinsic misalignment error rate caused by a phase reference mismatch. The gain and QBER are given by
\begin{equation} \label{eq:QEmuBB84}
\begin{aligned}
Q_\mu &= \sum_{k=0}^{\infty} \frac{\mu^k e^{-\mu}}{k!} Y_k, \\
&= 1-(1-Y_0)e^{-\eta\mu}, \\
E_\mu &= \sum_{k=0}^{\infty} \frac{\mu^k e^{-\mu}}{k!} e_kY_k, \\
&= e_d + \frac{(e_0-e_d)Y_0}{Q_\mu}, \\
\end{aligned}
\end{equation}
where $Y_0=2p_d$ and $e_0 = 1/2$.

The key rate of MDI-QKD is given by \cite{Lo2012Measurement}
\begin{equation}
\begin{aligned}
R_{MDI} = \dfrac{1}{2}\{ Q_{11} [1 - H(e_{11})] -f Q_{rect} H(E_{rect}) \},
\end{aligned}
\end{equation}
where $Q_{11} = \mu_a\mu_b e^{-\mu_a-\mu_b}Y_{11}$ and $1/2$ is the basis sifting factor. We take this formula from Eq.~(B27) in Ref.~\cite{Ma2012Alternative}. In simulation, the gain and error rates are given by

\begin{widetext}
\begin{equation}
\begin{aligned}
Y_{11} & = (1-p_d)^2 [ \dfrac{\eta_a\eta_b}{2} + (2\eta_a + 2\eta_b -3\eta_a\eta_b)p_d + 4(1-\eta_a)(1-\eta_b)p_d^2 ], \\
e_{11} & = e_0Y_{11} - (e_0- e_d)(1-p_d^2)\dfrac{\eta_a\eta_b}{2}, \\
Q_{rect} & = Q_{rect}^{(C)} + Q_{rect}^{(E)}, \\
Q_{rect}^{(C)} & = 2(1-p_d)^2 e^{-\mu^\prime/2} [1 - (1-p_d)e^{-\eta_a\mu_a/2}][1 - (1-p_d)e^{-\eta_b\mu_b/2}], \\
Q_{rect}^{(E)} & = 2p_d(1-p_d)^2 e^{-\mu^\prime/2}[I_0(2x) - (1-p_d)e^{-\mu^\prime/2}]; \\
E_{rect} Q_{rect} & = e_d Q_{rect}^{(C)} + (1 - e_d) Q_{rect}^{(E)}, \\
\end{aligned}
\end{equation}
\end{widetext}
Here,
\begin{equation} \label{eqn:MDI mu x}
\begin{aligned}
\mu^\prime & = \eta_a \mu_a + \eta_b \mu_b, \\
x & = \frac12\sqrt{\eta_a \mu_a \eta_b \mu_b}, \\
\end{aligned}
\end{equation}
where $\mu^\prime$ denotes the average number of photons reaching Eve's beam splitter, and  $\mu_a = \mu_b = \mu/2, \eta_a = \eta_b = \eta$. We take these formulas from Eqs.~(A9),~(A11),~(B7),~and (B28)-(B31) in Ref.~\cite{Ma2012Alternative}.

In 2014, Takeoka\textit{~et~al.} derived an upper bound of the key rate of the point-to-point-type QKD protocols \cite{takeoka2014fundamental},
\begin{equation}
R_{TGW} = -\log_2(\dfrac{1-\eta}{1+\eta}).
\end{equation}
Later, Pirandola\textit{~et~al.} established a tight upper bound \cite{Pirandola2017Fundamental},
\begin{equation}\label{eq:plob}
R_{PLOB} = - \log_2(1-\eta),
\end{equation}
which is the linear key-rate bound used in the main text. Note that Eq.~\eqref{eq:plob} is the secret key capacity of the pure loss channel, as it coincides with the lower bound previously known \cite{pirandola2009direct}.

\section{Beam-splitting Attack: Invalidation of the photon number channel model} \label{sc:attack}

Here, we argue that the tagging technique does not work for PM-QKD, by showing that the lower bound of the key rate from naively employing the tagging technique can be higher than an upper bound derived from a specific attack, in some parameter regime. The failure of the tagging technique stems from the failure of the photon number channel model used in the security proof \cite{Ma2008PhD}. In other words, the photon number channel model does not hold for the case when the random phases are announced during postprocessing.

\subsection{Key rate with or without the photon number channel}
Assuming the existence of the photon number channel model in PM-QKD, we can apply the tagging method and obtain a key-rate formula following the GLLP analysis \cite{gottesman04,Lo2005Decoy},
\begin{equation}\label{eq:keygllp}
r_{GLLP}=q_1[1-H(e^X_1)]- f H(E_{\mu}^Z),
\end{equation}
where $q_1$ is the single-photon detection ratio, $e^X_1$ and $E_{\mu}^Z$ are the single-photon phase error rate and total QBER, and $f$ is the error correction efficiency.

As for comparison, our PM-QKD key-rate formula is given by Eq.~\eqref{eqn:Shor-Preskill key rate},
\begin{equation}
r_{PM} = 1 - H(E^Z_\mu) - H(E^X_\mu),
\end{equation}
where $E_{\mu}^X$ is bounded by Eq.~\eqref{eq:EmuXzoom}
\begin{equation}
\begin{aligned}
E_{\mu}^X &= \sum_{k=0}^{\infty} q_{2k+1} e^Z_{2k+1} + \sum_{k=0}^{\infty} q_{2k} (1 - e^Z_{2k}), \\
&\le q_0 e^Z_0 +(q_1 e^Z_1 + q_3 e^Z_3 + q_5 e^Z_5) + (1 - q_0 - q_1 - q_3 -q_5).
\end{aligned}
\end{equation}

\subsection{Beam-splitting attack}
Now, we consider a beam-splitting attack (BS-attack), where Eve sets beam splitters with transmittance $\eta$ on both Alice and Bob's side to simulate a lossy channel. She intercepts the pulses on $A$ and $B$, regardless of the keys and bases of the pulses, and stores the reflected lights on her quantum memories, modes $A0$ and $B0$. Eve interferes the two transmitted pulses, modes $A1$ and $B1$, and announces the results. After Alice's and Bob's phase announcements, Eve performs an unambiguous state discrimination (USD) \cite{jaeger1995optimal} on the states on modes $A0$ and $B0$ separately to guess the key information. Since the states on $A0(B0)$ and $A1(B1)$ are uncorrelated because of the beam splitting of coherent states, the BS-attack will not introduce any error or other detectable effects. Apparently, this BS-attack is an individual attack. We will calculate Eve's successful probability for guessing the key bits unambiguously, and calculate the mutual information $I(\kappa_{a(b)}:E)$, from which we can derive an upper bound on the secure key rate.

In PM-QKD, Alice and Bob prepare coherent states $\{\ket{\sqrt{\mu_a}e^{i(\phi_a+\kappa_a)}}_A$ and $\ket{\sqrt{\mu_b}e^{i(\phi_b+\kappa_b)}}_B\}$, where $\mu_a = \mu_b = \mu/2$, $\phi_a,\phi_b\in[0,2\pi)$ are the random phases, and $\kappa_a,\kappa_b\in\{0,\pi\}$ are the key information. The states on system $A0,B0$ reflected by Eve's beam splitters are given by
\begin{equation}
\begin{aligned}
\ket{\sqrt{(1-\eta)\mu_a} e^{i(\phi_a+\kappa_a)}}_{A0}, \\
\ket{\sqrt{(1-\eta)\mu_b} e^{i(\phi_b+\kappa_b)}}_{B0}.
\end{aligned}
\end{equation}
The two remaining pulses on the system $A1,B1$, used for interference are given by
\begin{equation}
\begin{aligned}
\ket{\sqrt{\eta\mu_a} e^{i (\phi_a+\kappa_a)}}_{A1}, \\
\ket{\sqrt{\eta\mu_b} e^{i (\phi_b+\kappa_b)}}_{B1}.
\end{aligned}
\end{equation}
Suppose Eve interferes systems $A1$ and $B1$; then the total gain $Q_\mu$ is given by
\begin{equation}
\label{eqn:gain}
Q_\mu = 1 - e^{-\eta\mu}.
\end{equation}
Here, we consider the signal with $\phi_a = \phi_b$, and Eve holds perfect single-photon detectors.
Without loss of generality, we consider Eve's attack on system $A0$.
If Eve cannot get the sifted phase information $\phi_a$, then no matter the value of key phases $\kappa_a$, the state on $A0$ from Eve's perspective is
\begin{widetext}
\begin{equation}
\dfrac{1}{2\pi}\int_{0}^{2\pi} d\phi_a \ket{\sqrt{(1-\eta)\mu_a}e^{i(\phi_a+\kappa_a)}}_{A0}\bra{\sqrt{(1-\eta)\mu_a}e^{i(\phi_a+\kappa_a)}} = \sum_{k=0}^{\infty} P^{(1-\eta)\mu_a}(k) \ket{k}_{A0}\bra{k},
\end{equation}
\end{widetext}
where
\begin{equation}
\begin{aligned}
P^{(1-\eta)\mu_a}(k) = \frac{[(1-\eta)\mu_a]^k}{k!}e^{-(1-\eta)\mu_a} .
\end{aligned}
\end{equation}
In this case, Eve cannot obtain any information about $\kappa_a$. While in PM-QKD, the phase $\phi_a$ will be announced. In this case, Eve can first rotate the states on system $A0$ by $-\phi_a$. Now, the state becomes
\begin{equation}
\begin{aligned}
\dfrac{1}{2}&(\ket{\sqrt{(1-\eta)\mu_a}}_{A0}\bra{\sqrt{(1-\eta)\mu_a}} \\
&+ \ket{-\sqrt{(1-\eta)\mu_a}}_{A0}\bra{-\sqrt{(1-\eta)\mu_a}}).
\end{aligned}
\end{equation}
Then, Eve only needs to determine whether the phase $\kappa_a$ is $0$ or $\pi$ to obtain the key information. As shown in Ref.~\cite{jaeger1995optimal}, if two pure states $\ket{p}$ and $\ket{q}$ are prepared with the same \textit{a priori} $1/2$, the maximum probability of \emph{unambiguous} discrimination is
\begin{equation}
P_{des} = 1 - |\braket{p|q}|.
\end{equation}
Therefore, the probability of successfully distinguishing the states $\ket{\sqrt{(1-\eta)\mu_a}}$ and $\ket{-\sqrt{(1-\eta)\mu_a}}$ is given by
\begin{equation}
P_{suc} = 1 - |\braket{\sqrt{(1-\eta)\mu_a}|-\sqrt{(1-\eta)\mu_a}}| = 1 - e^{-(1-\eta)\mu}.
\end{equation}
Note that Eve already know the value $\kappa_a - \kappa_b$ with interference results on $A1, B1$. Eve only needs to learn either of $\kappa_a$ or $\kappa_b$. Thus, Eve's successful unambiguous measurement probability is
\begin{equation}
P_{BS} = 1- (1-P_{suc})^2 = 1 - e^{-2(1-\eta)\mu}.
\end{equation}
If the light intensity $\mu$ is large enough, $P_{BS} = 1$, then Eve can learn the key information with a high probability.

With unambiguous state discrimination, the mutual information of Eve's measurement result (denoted by variable $E$) and Alice's and Bob's key $\kappa_{a(b)}$ is
\begin{equation}
I(\kappa_{a(b)}:E) = P_{BS} =1 - e^{-2(1-\eta)\mu}.
\end{equation}
Suppose that there is no extra error induced by other factors; then the BS-attack provides an upper bound on the key rate,
\begin{equation} \label{eqn: rBS}
\begin{aligned}
r_{BS} & = [1-I(\kappa_{a(b)}:E)] \\
& = e^{-2(1-\eta)\mu} \\
& = \frac{e^{-2\mu}}{(1-Q_\mu)^2}.
\end{aligned}
\end{equation}
Here, the third equality is based on Eq.~\eqref{eqn:gain}.

\subsection{Comparison}
Now we make a simulation to compare the key rate in the following cases: with the ``tagging'' method, by our security proof, and the upper bound under BS-attack.
In the discussion, we neglect the misalignment error and dark counts.

First, under BS-attack, the yield is given by
\begin{equation}
Y_k = 1 - (1-\eta)^k,
\end{equation}
and the fraction of the $k$-photon component is
\begin{equation}
q_k = (e^{-\mu}\dfrac{\mu^k}{k!})\dfrac{Y_k}{Q_\mu} = [1 - (1-\eta)^k]\dfrac{e^{-\mu}\mu^k}{k! Q_\mu}.
\end{equation}
Note that $q_0=0$.

Second, in the BS-attack scenario, Eve's operations will not introduce any error; hence,
\begin{equation}
\begin{aligned}
e^Z_k &= 0, \quad\forall k \ge 0, \\
E^Z_\mu &= 0. \\
\end{aligned}
\end{equation}
Then Eqs.~\eqref{eq:keygllp}, ~\eqref{eqn:Shor-Preskill key rate}, and ~\eqref{eq:EmuXzoom} can be simplified to
\begin{equation}
\begin{aligned} \label{eqn: runderBS}
r_{GLLP} &= q_1 \\
&=  \eta\mu e^{-\mu}, \\
r_{PM} &= [1 - H(E^X)] \\
&\ge [1 - H(1 - q_1 - q_3 - q_5)].
\end{aligned}
\end{equation}

We compare the key-rate formula $r_{GLLP}, r_{PM}$ in Eq.~\eqref{eqn: runderBS} with the key-rate upper bound by BS-attack $r_{BS}$ in Eq.~\eqref{eqn: rBS}. We set the total light intensity to be a typical value $\mu=0.5$, and we adjust $\eta$ to compare the key-rate performance. As shown in Fig.~\ref{fig:BSattackMu05}, the GLLP ``tagging'' formula for the ``single-photon'' component cannot hold under the BS-attack for a transmittance $\eta<0.6$.
If we fix the transmittance $\eta=0.2$ and adjust $\mu$, Fig.~\ref{fig:BSattackEta02} shows that the GLLP tagging formula cannot hold under the BS-attack.

\begin{figure}[hbtp]
\centering
\resizebox{8cm}{!}{\includegraphics{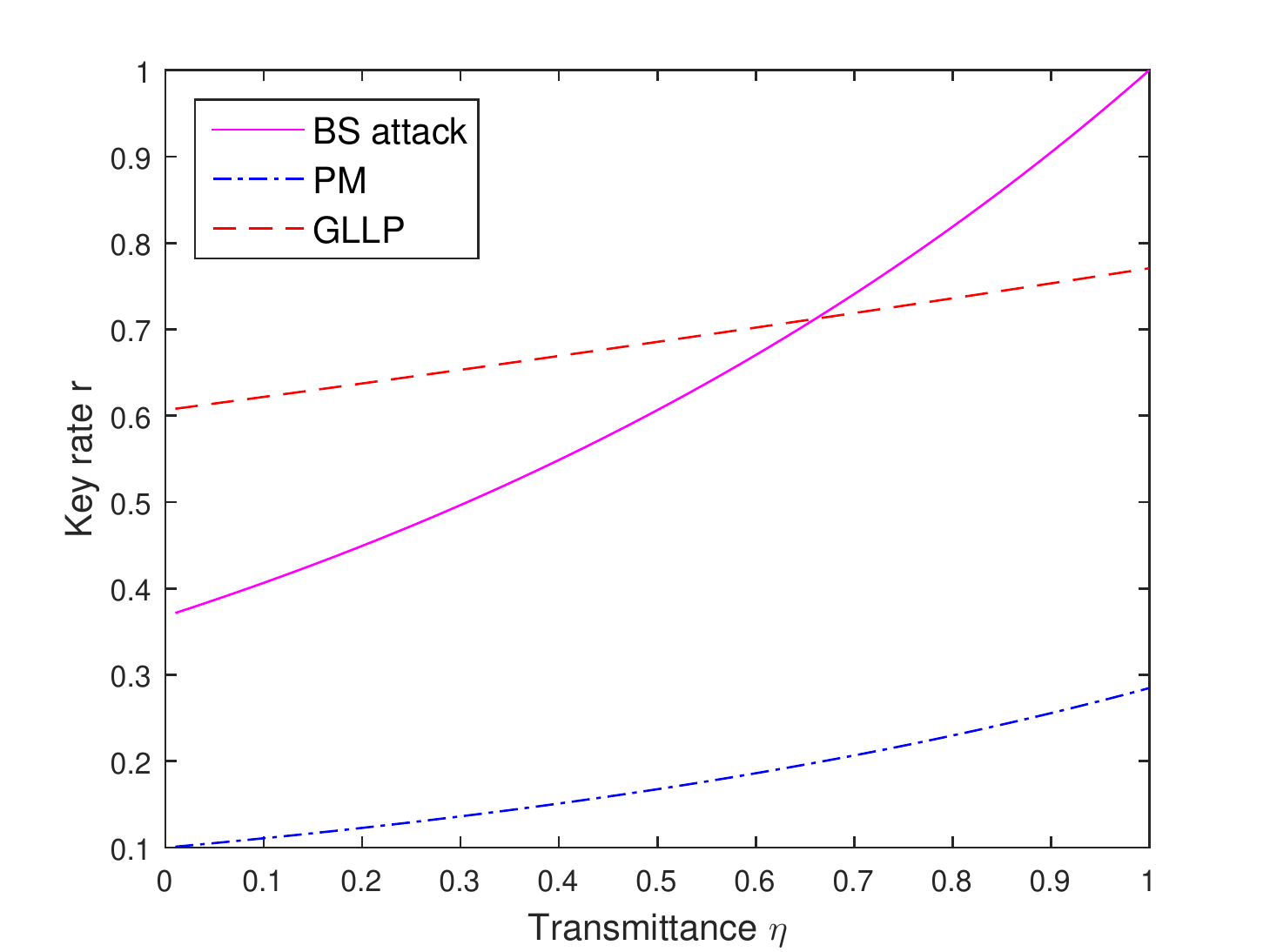}}
\caption{Key-rate comparison with fixed light intensity $\mu=0.5$. Here we can see that the GLLP tagging formula for the single-photon component cannot hold under the BS-attack when the transmittance $\eta<0.6$.} \label{fig:BSattackMu05}
\end{figure}

\begin{figure}[hbtp]
\centering
\resizebox{8cm}{!}{\includegraphics{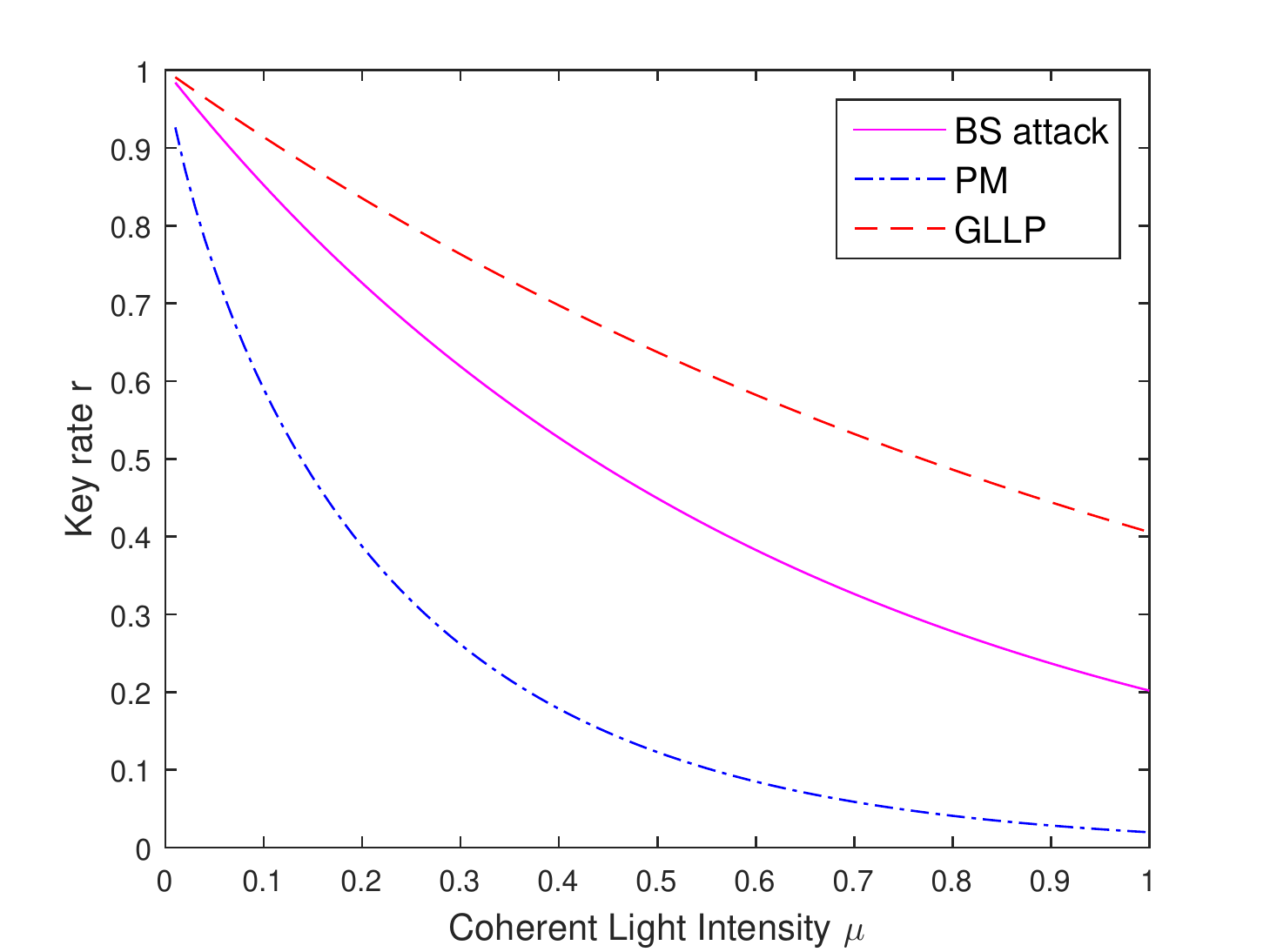}}
\caption{Key-rate comparison with fixed transmittance $\eta=0.2$. The GLLP tagging formula for the single-photon component cannot hold under the BS-attack.} \label{fig:BSattackEta02}
\end{figure}

The BS-attack results serve to invalidate the tagging method even the single-photon component cannot exist, since the announcement of phase $\phi_a,\phi_b$ will leak more information of the key bits $\kappa_a $ and $\kappa_b$.

Recently, Wang\textit{~et~al.} proposed an eavesdropping strategy to the TF-QKD protocol \cite{Wang2018Effective}. The attack can also be performed against the PM-QKD protocol. Under such an attack,
Eve can learn all the key bits. The key rate provided by the GLLP formula, Eq.~\eqref{eq:keygllp}, is 0.5 (for all the clicked signals) while that given by our security proof is strictly 0. This also shows the invalidation of the GLLP key-rate formula, and our security proof is still valid under such an attack.



\section{Comparison with other phase-encoding MDI-QKD}
Here, we compare three related phase-encoding MDI-QKD protocols, including the previous phase-encoding MDI-QKD (Scheme I in Ref.~\cite{Tamaki2012PhaseMDI}), the recently proposed TF-QKD \cite{Lucamarini2018TF}, and our protocol (PM-QKD). For the simplicity of the statements, we assume that the phase reference of Alice and Bob is locked in all the protocols.
\begin{figure}[htbp]
\centering
\includegraphics[width=8cm]{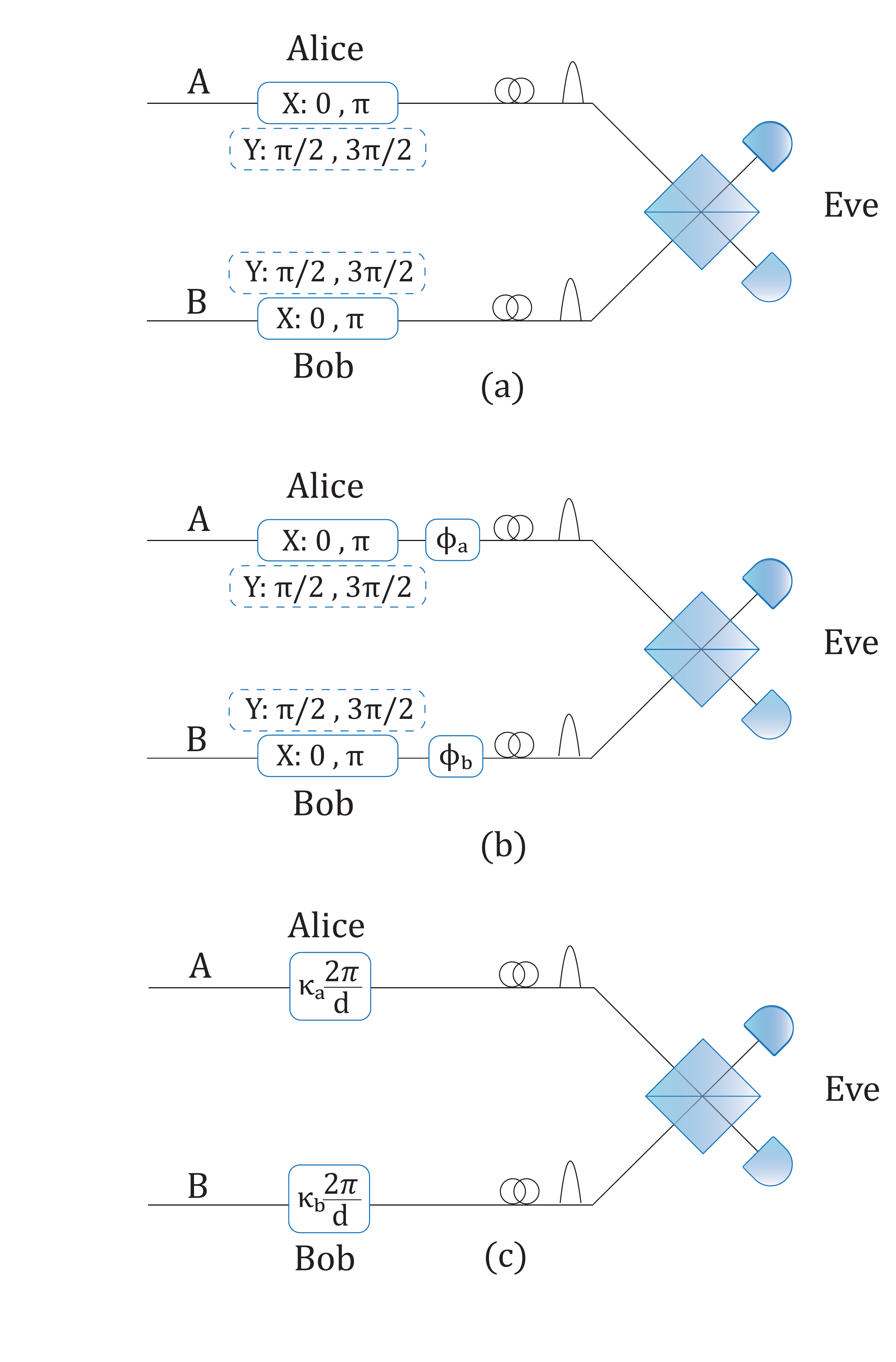}
\caption{(a) Schematic diagram of phase-encoding MDI-QKD (Scheme I in Ref.~\cite{Tamaki2012PhaseMDI}). (b) Schematic diagram of TF-QKD, which is the decoy-state version of the scheme in (a).  (c) Schematic diagram of $d$-phase PM-QKD.}
\label{fig:mdivstfvspm}
\end{figure}

\subsection{Phase-encoding MDI-QKD} \label{AppSub:PEMDI}
In the phase encoding MDI-QKD (Scheme I in Ref.~\cite{Tamaki2012PhaseMDI}), Alice generates two random bits $\kappa_a, \beta_a$ as the key and the random choice of the $X$ or $Y$ basis, respectively, and then generates a coherent pulse $\ket{\sqrt{\mu/2} e^{i(\pi\kappa_a + \pi\beta_a/2)} }_A$. Similarly, Bob generates $\ket{\sqrt{\mu/2} e^{i(\pi\kappa_b + \pi\beta_b/2)} }_B$. They send their coherent pulses to Eve, who is supposed to perform an interference measurement and announce the detection results [Fig.~\ref{fig:mdivstfvspm}(a)]. After Eve's announcement, Alice and Bob announce the basis information $\beta_a, \beta_b$ and perform basis sifting.


Note that there is only single-photon detection in this protocol. Therefore, the number of clicked signals scales with $O(\sqrt{\eta})$, where $\eta$ is the total transmittance between Alice and Bob. However, there is a considerable source flaw caused by using a coherent state as an approximation of an ideal single photon state, resulting in a final key rate of $O(\eta)$.

\subsection{TF-QKD}
To avoid the performance deterioration caused by the source flaw in the phase-encoding MDI-QKD protocol above, a natural idea is to apply the decoy-state method \cite{Lo2005Decoy}, that is, to perform phase randomization, and to estimate the fraction of $q_1$ and the error rate $e_1$ of the single-photon component.

Recently, Lucamarini\textit{~et~al.} modified the phase-encoding QKD protocol, namely TF-QKD \cite{Lucamarini2018TF}. As shown in Fig.~\ref{fig:mdivstfvspm}(b), Alice generates two random bits $\kappa_a, \beta_a$ as the key, chooses randomly the $X$ or $Y$ basis, and modulates another random phase $\phi_a$ for the phase randomization in the decoy-state method. She generates a coherent pulse $\ket{\sqrt{\mu/2} e^{i(\pi\kappa_a + \pi\beta_a/2 + \phi_a)} }_A$. Similarly, Bob generates a coherent pulse $\ket{\sqrt{\mu/2} e^{i(\pi\kappa_b + \pi\beta_b/2 + \phi_b)} }_B$. They send their coherent pulses to Eve, who is supposed to perform the interference measurement and announce the detection results. After Eve's announcement, Alice and Bob announce $\beta_a, \beta_b; \phi_a, \phi_b$ and perform basis sifting and phase sifting.


\subsection{$d$-phase phase-matching QKD}

In the $d$-phase PM-QKD, Alice first generates a random integer $\kappa_a\in\{0, 1,..., d-1\}$ and then prepares a coherent pulse $\ket{\sqrt{\mu/2} \exp{(i\kappa_a \dfrac{2\pi}{d})}}_A$. Similarly, Bob generates $\kappa_b\in\{0,1,..., d-1\}$ and prepares a similar state $\ket{\sqrt{\mu/2} \exp{(i\kappa_b \dfrac{2\pi}{d}})}_B$. They send their coherent pulses to Eve, who is supposed to perform an interference measurement and announce the detection results. If Eve announces detector $L/R$ clicks, Alice and Bob keep the numbers $\kappa_a, \kappa_b$ and which detector clicks.

After many rounds of the steps above, Alice and Bob randomly select some of the data, announce them, and calculate the $L/R$ detection probability for each case. By the estimated probabilities, they calculate the key rate and extract private keys.

As we can see, after Eve's announcement, Alice and Bob will share some mutual information. While Eve cannot fully learn the variable $\kappa_a$, the information-secure key can be generated between Alice and Bob.

The PM-QKD protocol in the main text and Appendix \ref{Sc:SecureProof} corresponds to the case of $d=2$, which is combined with a random phase announcement and the decoy-state method. In Appendix \ref{Sc:SecureProof}, we completed the proof for the PM-QKD of $d=2$ with a random phase announcement. The security proof for this generalized case is left for future work.

\subsection{Comparison of different protocols}

In all three protocols, phase-encoding MDI-QKD \cite{Tamaki2012PhaseMDI}, TF-QKD \cite{Lucamarini2018TF}, and PM-QKD, single-detection clicks on the untrusted node are used as successful measurement events, whose rate scales with $O(\sqrt{\eta})$. Technically, this is the precondition of the key rate having square-root scaling\cite{Lucamarini2018TF}.
Nevertheless, in the phase-encoding MDI scheme, without phase randomization, the single-photon source can be approximated by two weak coherent states, which decreases significantly the key rate of the protocol to $O(\eta)$. In TF-QKD, a single-qubit view is taken and a higher key rate is targeted. Unfortunately, the phase announcement invalidates the existence of the photon number channel model \cite{Ma2008PhD}, and, hence, no enhancement can be claimed by directly applying the decoy-state method onto the phase-encoding MDI-QKD scheme.

Here, in the phase-matching QKD scheme, we switch from the qubit-based view to the optical-mode-based view. As shown above, our proposed phase-matching QKD follows the phase-encoding MDI-QKD scheme and the TF-QKD scheme, by modifying the encoding and basis choice. The name ``phase-matching (MDI-)QKD'' follows ``phase-encoding MDI-QKD''. We removed ``MDI'' (which is not the key point to our work) to make the name concise.

The PM-QKD protocol in the main text and Appendix \ref{Sc:SecureProof} is the one with $d=2$ and random phase announcement. Though the current practical implementations of phase-encoding MDI-QKD / TF-QKD and PM-QKD do resemble each other, the security scenarios are quite different. The TF-QKD can be taken as an extension of the BB84 protocol, which utilizes the single-photon source and is a discrete-variable QKD protocol. In the future, many discrete-variable QKD design techniques, such as a six-state protocol and reference-frame-independent protocol, can be employed in this framework. On the other hand, in PM-QKD, we focus more on the optical modes rather than single-photon states (qubits). It would be interesting to see whether the security proof techniques developed in continuous-variable QKD can be applied to PM-QKD.

\subsection{Recent related works on TF-QKD}
There are some recent works based on TF-QKD. Tamaki\textit{~et~al.} provided a security proof of the TF-QKD protocol \cite{tamaki2018information}. They modified the TF-QKD protocol by introducing a ``test mode'', where Alice and Bob do not announce the phase information and photon-number-channel model holds. The original TF-QKD protocol is called ``code mode''. Following the original security proof of phase-encoding MDI-QKD \cite{Tamaki2012PhaseMDI}, they estimated the phase error $E_X$ by considering the imbalance of different bases. With a fair sampling argument on the test mode and the code mode, the basis imbalance in the code mode can be estimated by the one in the test mode. A square-root-scaling key rate has been derived, but it is significantly lower than ours. Note that, there are still two bases in Tamaki\textit{~et~al.}'s protocol and security proof, which follows the qubit-based view of BB84. The two-basis requirement in the security proof implies that it cannot be applied to our PM-QKD protocol directly, which also highlights the difference between the phase-encoding MDI-QKD / TF-QKD and PM-QKD.

Wang\textit{~et~al.} proposed a ``sending or not sending'' TF-QKD protocol \cite{Wang2018Sending}, which aims to utilize the $Z$-basis of single-photon for key generation, following the viewpoint of TF-QKD. However, the definition of the $Z$-basis encoding seems confusing. In the original TF-QKD protocol, the definitions of the $X$ or $Y$ bases refer to the ancillary bits rather than the real single photon. Also, the $Z$-basis encoding on the real photon does not correspond to a $Z$-basis encoding on the ancillary bits. This protocol needs further studies.

\end{appendix}


\normalem
\bibliography{bibPosBB84}

\begin{thebibliography}{49}%
\makeatletter
\providecommand \@ifxundefined [1]{%
 \@ifx{#1\undefined}
}%
\providecommand \@ifnum [1]{%
 \ifnum #1\expandafter \@firstoftwo
 \else \expandafter \@secondoftwo
 \fi
}%
\providecommand \@ifx [1]{%
 \ifx #1\expandafter \@firstoftwo
 \else \expandafter \@secondoftwo
 \fi
}%
\providecommand \natexlab [1]{#1}%
\providecommand \enquote  [1]{``#1''}%
\providecommand \bibnamefont  [1]{#1}%
\providecommand \bibfnamefont [1]{#1}%
\providecommand \citenamefont [1]{#1}%
\providecommand \href@noop [0]{\@secondoftwo}%
\providecommand \href [0]{\begingroup \@sanitize@url \@href}%
\providecommand \@href[1]{\@@startlink{#1}\@@href}%
\providecommand \@@href[1]{\endgroup#1\@@endlink}%
\providecommand \@sanitize@url [0]{\catcode `\\12\catcode `\$12\catcode
  `\&12\catcode `\#12\catcode `\^12\catcode `\_12\catcode `\%12\relax}%
\providecommand \@@startlink[1]{}%
\providecommand \@@endlink[0]{}%
\providecommand \url  [0]{\begingroup\@sanitize@url \@url }%
\providecommand \@url [1]{\endgroup\@href {#1}{\urlprefix }}%
\providecommand \urlprefix  [0]{URL }%
\providecommand \Eprint [0]{\href }%
\providecommand \doibase [0]{http://dx.doi.org/}%
\providecommand \selectlanguage [0]{\@gobble}%
\providecommand \bibinfo  [0]{\@secondoftwo}%
\providecommand \bibfield  [0]{\@secondoftwo}%
\providecommand \translation [1]{[#1]}%
\providecommand \BibitemOpen [0]{}%
\providecommand \bibitemStop [0]{}%
\providecommand \bibitemNoStop [0]{.\EOS\space}%
\providecommand \EOS [0]{\spacefactor3000\relax}%
\providecommand \BibitemShut  [1]{\csname bibitem#1\endcsname}%
\let\auto@bib@innerbib\@empty
\bibitem [{\citenamefont {Bennett}\ and\ \citenamefont
  {Brassard}(1984)}]{Bennett1984Quantum}%
  \BibitemOpen
  \bibfield  {author} {\bibinfo {author} {\bibfnamefont {C.~H.}\ \bibnamefont
  {Bennett}}\ and\ \bibinfo {author} {\bibfnamefont {G.}~\bibnamefont
  {Brassard}},\ }\bibfield  {title} {\enquote {\bibinfo {title} {{Quantum
  Cryptography: Public Key Distribution and Coin Tossing}},}\ }in\ \href
  {https://doi.org/10.1016/j.tcs.2014.05.025} {\emph {\bibinfo {booktitle}
  {Proceedings of the IEEE International Conference on Computers, Systems and
  Signal Processing}}}\ (\bibinfo  {publisher} {IEEE Press},\ \bibinfo
  {address} {New York},\ \bibinfo {year} {1984})\ pp.\ \bibinfo {pages}
  {175--179}\BibitemShut {NoStop}%
\bibitem [{\citenamefont {Ekert}(1991)}]{Ekert1991Quantum}%
  \BibitemOpen
  \bibfield  {author} {\bibinfo {author} {\bibfnamefont {Artur~K.}\
  \bibnamefont {Ekert}},\ }\bibfield  {title} {\enquote {\bibinfo {title}
  {Quantum cryptography based on bell's theorem},}\ }\href {\doibase
  10.1103/PhysRevLett.67.661} {\bibfield  {journal} {\bibinfo  {journal} {Phys.
  Rev. Lett.}\ }\textbf {\bibinfo {volume} {67}},\ \bibinfo {pages} {661--663}
  (\bibinfo {year} {1991})}\BibitemShut {NoStop}%
\bibitem [{\citenamefont {Mayers}(2001)}]{Mayers2001Unconditional}%
  \BibitemOpen
  \bibfield  {author} {\bibinfo {author} {\bibfnamefont {D.}~\bibnamefont
  {Mayers}},\ }\bibfield  {title} {\enquote {\bibinfo {title} {Unconditional
  security in quantum cryptography},}\ }\href
  {https://dl.acm.org/citation.cfm?id=382781} {\bibfield  {journal} {\bibinfo
  {journal} {J. ACM}\ }\textbf {\bibinfo {volume} {48}},\ \bibinfo {pages}
  {351--406} (\bibinfo {year} {2001})}\BibitemShut {NoStop}%
\bibitem [{\citenamefont {Lo}\ and\ \citenamefont
  {Chau}(1999)}]{Lo1999Unconditional}%
  \BibitemOpen
  \bibfield  {author} {\bibinfo {author} {\bibfnamefont {H.~K.}\ \bibnamefont
  {Lo}}\ and\ \bibinfo {author} {\bibfnamefont {H.~F.}\ \bibnamefont {Chau}},\
  }\bibfield  {title} {\enquote {\bibinfo {title} {Unconditional security of
  quantum key distribution over arbitrarily long distances},}\ }\href
  {http://science.sciencemag.org/content/283/5410/2050} {\bibfield  {journal}
  {\bibinfo  {journal} {Science}\ }\textbf {\bibinfo {volume} {283}},\ \bibinfo
  {pages} {2050} (\bibinfo {year} {1999})}\BibitemShut {NoStop}%
\bibitem [{\citenamefont {Shor}\ and\ \citenamefont
  {Preskill}(2000)}]{Shor2000Simple}%
  \BibitemOpen
  \bibfield  {author} {\bibinfo {author} {\bibfnamefont {P.~W.}\ \bibnamefont
  {Shor}}\ and\ \bibinfo {author} {\bibfnamefont {J}~\bibnamefont {Preskill}},\
  }\bibfield  {title} {\enquote {\bibinfo {title} {Simple proof of security of
  the bb84 quantum key distribution protocol},}\ }\href
  {https://journals.aps.org/prl/abstract/10.1103/PhysRevLett.85.441} {\bibfield
   {journal} {\bibinfo  {journal} {Phys. Rev. Lett.}\ }\textbf {\bibinfo
  {volume} {85}},\ \bibinfo {pages} {441} (\bibinfo {year} {2000})}\BibitemShut
  {NoStop}%
\bibitem [{\citenamefont {Bennett}\ \emph {et~al.}(1992)\citenamefont
  {Bennett}, \citenamefont {Bessette}, \citenamefont {Brassard}, \citenamefont
  {Salvail},\ and\ \citenamefont {Smolin}}]{bennett1992experimental}%
  \BibitemOpen
  \bibfield  {author} {\bibinfo {author} {\bibfnamefont {Charles~H}\
  \bibnamefont {Bennett}}, \bibinfo {author} {\bibfnamefont {Fran{\c{c}}ois}\
  \bibnamefont {Bessette}}, \bibinfo {author} {\bibfnamefont {Gilles}\
  \bibnamefont {Brassard}}, \bibinfo {author} {\bibfnamefont {Louis}\
  \bibnamefont {Salvail}}, \ and\ \bibinfo {author} {\bibfnamefont {John}\
  \bibnamefont {Smolin}},\ }\bibfield  {title} {\enquote {\bibinfo {title}
  {Experimental quantum cryptography},}\ }\href
  {https://doi.org/10.1007/BF00191318} {\bibfield  {journal} {\bibinfo
  {journal} {J. Cryptol.}\ }\textbf {\bibinfo {volume} {5}},\ \bibinfo {pages}
  {3--28} (\bibinfo {year} {1992})}\BibitemShut {NoStop}%
\bibitem [{\citenamefont {Liao}\ \emph {et~al.}(2017)\citenamefont {Liao},
  \citenamefont {Cai}, \citenamefont {Liu}, \citenamefont {Zhang},
  \citenamefont {Li}, \citenamefont {Ren}, \citenamefont {Yin}, \citenamefont
  {Shen}, \citenamefont {Cao}, \citenamefont {Li} \emph
  {et~al.}}]{liao2017satellite}%
  \BibitemOpen
  \bibfield  {author} {\bibinfo {author} {\bibfnamefont {Sheng-Kai}\
  \bibnamefont {Liao}}, \bibinfo {author} {\bibfnamefont {Wen-Qi}\ \bibnamefont
  {Cai}}, \bibinfo {author} {\bibfnamefont {Wei-Yue}\ \bibnamefont {Liu}},
  \bibinfo {author} {\bibfnamefont {Liang}\ \bibnamefont {Zhang}}, \bibinfo
  {author} {\bibfnamefont {Yang}\ \bibnamefont {Li}}, \bibinfo {author}
  {\bibfnamefont {Ji-Gang}\ \bibnamefont {Ren}}, \bibinfo {author}
  {\bibfnamefont {Juan}\ \bibnamefont {Yin}}, \bibinfo {author} {\bibfnamefont
  {Qi}~\bibnamefont {Shen}}, \bibinfo {author} {\bibfnamefont {Yuan}\
  \bibnamefont {Cao}}, \bibinfo {author} {\bibfnamefont {Zheng-Ping}\
  \bibnamefont {Li}},  \emph {et~al.},\ }\bibfield  {title} {\enquote {\bibinfo
  {title} {Satellite-to-ground quantum key distribution},}\ }\href
  {https://www.nature.com/articles/nature23655} {\bibfield  {journal} {\bibinfo
   {journal} {Nature}\ }\textbf {\bibinfo {volume} {549}},\ \bibinfo {pages}
  {43} (\bibinfo {year} {2017})}\BibitemShut {NoStop}%
\bibitem [{Note1()}]{Note1}%
  \BibitemOpen
  \bibinfo {note} {In practice, single-photon sources are often replaced with
  weak coherent state sources or heralded single-photon sources. Nevertheless,
  only the single-photon components are used for secure key
  distribution.}\BibitemShut {Stop}%
\bibitem [{\citenamefont {Curty}\ \emph {et~al.}(2004)\citenamefont {Curty},
  \citenamefont {Lewenstein},\ and\ \citenamefont
  {L\"utkenhaus}}]{Curty2004Entanglement}%
  \BibitemOpen
  \bibfield  {author} {\bibinfo {author} {\bibfnamefont {Marcos}\ \bibnamefont
  {Curty}}, \bibinfo {author} {\bibfnamefont {Maciej}\ \bibnamefont
  {Lewenstein}}, \ and\ \bibinfo {author} {\bibfnamefont {Norbert}\
  \bibnamefont {L\"utkenhaus}},\ }\bibfield  {title} {\enquote {\bibinfo
  {title} {Entanglement as a precondition for secure quantum key
  distribution},}\ }\href {\doibase 10.1103/PhysRevLett.92.217903} {\bibfield
  {journal} {\bibinfo  {journal} {Phys. Rev. Lett.}\ }\textbf {\bibinfo
  {volume} {92}},\ \bibinfo {pages} {217903} (\bibinfo {year}
  {2004})}\BibitemShut {NoStop}%
\bibitem [{\citenamefont {Takeoka}\ \emph {et~al.}(2014)\citenamefont
  {Takeoka}, \citenamefont {Guha},\ and\ \citenamefont
  {Wilde}}]{takeoka2014fundamental}%
  \BibitemOpen
  \bibfield  {author} {\bibinfo {author} {\bibfnamefont {Masahiro}\
  \bibnamefont {Takeoka}}, \bibinfo {author} {\bibfnamefont {Saikat}\
  \bibnamefont {Guha}}, \ and\ \bibinfo {author} {\bibfnamefont {Mark~M}\
  \bibnamefont {Wilde}},\ }\bibfield  {title} {\enquote {\bibinfo {title}
  {Fundamental rate-loss tradeoff for optical quantum key distribution},}\
  }\href {https://www.nature.com/articles/ncomms6235} {\bibfield  {journal}
  {\bibinfo  {journal} {Nat. Commun.}\ }\textbf {\bibinfo {volume} {5}},\
  \bibinfo {pages} {5235} (\bibinfo {year} {2014})}\BibitemShut {NoStop}%
\bibitem [{\citenamefont {Pirandola}\ \emph {et~al.}(2017)\citenamefont
  {Pirandola}, \citenamefont {Laurenza}, \citenamefont {Ottaviani},\ and\
  \citenamefont {Banchi}}]{Pirandola2017Fundamental}%
  \BibitemOpen
  \bibfield  {author} {\bibinfo {author} {\bibfnamefont {Stefano}\ \bibnamefont
  {Pirandola}}, \bibinfo {author} {\bibfnamefont {Riccardo}\ \bibnamefont
  {Laurenza}}, \bibinfo {author} {\bibfnamefont {Carlo}\ \bibnamefont
  {Ottaviani}}, \ and\ \bibinfo {author} {\bibfnamefont {Leonardo}\
  \bibnamefont {Banchi}},\ }\bibfield  {title} {\enquote {\bibinfo {title}
  {Fundamental limits of repeaterless quantum communications},}\ }\href
  {https://www.nature.com/articles/ncomms15043} {\bibfield  {journal} {\bibinfo
   {journal} {Nat. Commun.}\ }\textbf {\bibinfo {volume} {8}},\ \bibinfo
  {pages} {15043} (\bibinfo {year} {2017})}\BibitemShut {NoStop}%
\bibitem [{\citenamefont {Zukowski}\ \emph {et~al.}(1993)\citenamefont
  {Zukowski}, \citenamefont {Zeilinger}, \citenamefont {Horne},\ and\
  \citenamefont {Ekert}}]{EntSwap1993}%
  \BibitemOpen
  \bibfield  {author} {\bibinfo {author} {\bibfnamefont {M.}~\bibnamefont
  {Zukowski}}, \bibinfo {author} {\bibfnamefont {A.}~\bibnamefont {Zeilinger}},
  \bibinfo {author} {\bibfnamefont {M.~A.}\ \bibnamefont {Horne}}, \ and\
  \bibinfo {author} {\bibfnamefont {A.~K.}\ \bibnamefont {Ekert}},\ }\bibfield
  {title} {\enquote {\bibinfo {title} {``event-ready-detectors'' bell
  experiment via entanglement swapping},}\ }\href {\doibase
  10.1103/PhysRevLett.71.4287} {\bibfield  {journal} {\bibinfo  {journal}
  {Phys. Rev. Lett.}\ }\textbf {\bibinfo {volume} {71}},\ \bibinfo {pages}
  {4287--4290} (\bibinfo {year} {1993})}\BibitemShut {NoStop}%
\bibitem [{\citenamefont {Briegel}\ \emph {et~al.}(1998)\citenamefont
  {Briegel}, \citenamefont {D\"ur}, \citenamefont {Cirac},\ and\ \citenamefont
  {Zoller}}]{Briegel1998Repeater}%
  \BibitemOpen
  \bibfield  {author} {\bibinfo {author} {\bibfnamefont {H.-J.}\ \bibnamefont
  {Briegel}}, \bibinfo {author} {\bibfnamefont {W.}~\bibnamefont {D\"ur}},
  \bibinfo {author} {\bibfnamefont {J.~I.}\ \bibnamefont {Cirac}}, \ and\
  \bibinfo {author} {\bibfnamefont {P.}~\bibnamefont {Zoller}},\ }\bibfield
  {title} {\enquote {\bibinfo {title} {Quantum repeaters: The role of imperfect
  local operations in quantum communication},}\ }\href {\doibase
  10.1103/PhysRevLett.81.5932} {\bibfield  {journal} {\bibinfo  {journal}
  {Phys. Rev. Lett.}\ }\textbf {\bibinfo {volume} {81}},\ \bibinfo {pages}
  {5932--5935} (\bibinfo {year} {1998})}\BibitemShut {NoStop}%
\bibitem [{\citenamefont {Azuma}\ \emph {et~al.}(2015)\citenamefont {Azuma},
  \citenamefont {Tamaki},\ and\ \citenamefont {Lo}}]{azuma2015all}%
  \BibitemOpen
  \bibfield  {author} {\bibinfo {author} {\bibfnamefont {Koji}\ \bibnamefont
  {Azuma}}, \bibinfo {author} {\bibfnamefont {Kiyoshi}\ \bibnamefont {Tamaki}},
  \ and\ \bibinfo {author} {\bibfnamefont {Hoi-Kwong}\ \bibnamefont {Lo}},\
  }\bibfield  {title} {\enquote {\bibinfo {title} {All-photonic quantum
  repeaters},}\ }\href {https://www.nature.com/articles/ncomms7787} {\bibfield
  {journal} {\bibinfo  {journal} {Nat. Commun.}\ }\textbf {\bibinfo {volume}
  {6}},\ \bibinfo {pages} {6787} (\bibinfo {year} {2015})}\BibitemShut
  {NoStop}%
\bibitem [{\citenamefont {Lo}\ \emph {et~al.}(2012)\citenamefont {Lo},
  \citenamefont {Curty},\ and\ \citenamefont {Qi}}]{Lo2012Measurement}%
  \BibitemOpen
  \bibfield  {author} {\bibinfo {author} {\bibfnamefont {Hoi-Kwong}\
  \bibnamefont {Lo}}, \bibinfo {author} {\bibfnamefont {Marcos}\ \bibnamefont
  {Curty}}, \ and\ \bibinfo {author} {\bibfnamefont {Bing}\ \bibnamefont
  {Qi}},\ }\bibfield  {title} {\enquote {\bibinfo {title}
  {Measurement-device-independent quantum key distribution},}\ }\href {\doibase
  10.1103/PhysRevLett.108.130503} {\bibfield  {journal} {\bibinfo  {journal}
  {Phys. Rev. Lett.}\ }\textbf {\bibinfo {volume} {108}},\ \bibinfo {pages}
  {130503} (\bibinfo {year} {2012})}\BibitemShut {NoStop}%
\bibitem [{\citenamefont {Lucamarini}\ \emph {et~al.}(2018)\citenamefont
  {Lucamarini}, \citenamefont {Yuan}, \citenamefont {Dynes},\ and\
  \citenamefont {Shields}}]{Lucamarini2018TF}%
  \BibitemOpen
  \bibfield  {author} {\bibinfo {author} {\bibfnamefont {M.}~\bibnamefont
  {Lucamarini}}, \bibinfo {author} {\bibfnamefont {Z.L.}\ \bibnamefont {Yuan}},
  \bibinfo {author} {\bibfnamefont {J.F.}\ \bibnamefont {Dynes}}, \ and\
  \bibinfo {author} {\bibfnamefont {A.J.}\ \bibnamefont {Shields}},\ }\bibfield
   {title} {\enquote {\bibinfo {title} {Overcoming the rate--distance limit of
  quantum key distribution without quantum repeaters},}\ }\href
  {https://www.nature.com/articles/s41586-018-0066-6} {\bibfield  {journal}
  {\bibinfo  {journal} {Nature (London)}\ }\textbf {\bibinfo {volume} {557}},\
  \bibinfo {pages} {400--403} (\bibinfo {year} {2018})}\BibitemShut {NoStop}%
\bibitem [{\citenamefont {Tamaki}\ \emph {et~al.}(2012)\citenamefont {Tamaki},
  \citenamefont {Lo}, \citenamefont {Fung},\ and\ \citenamefont
  {Qi}}]{Tamaki2012PhaseMDI}%
  \BibitemOpen
  \bibfield  {author} {\bibinfo {author} {\bibfnamefont {Kiyoshi}\ \bibnamefont
  {Tamaki}}, \bibinfo {author} {\bibfnamefont {Hoi-Kwong}\ \bibnamefont {Lo}},
  \bibinfo {author} {\bibfnamefont {Chi-Hang~Fred}\ \bibnamefont {Fung}}, \
  and\ \bibinfo {author} {\bibfnamefont {Bing}\ \bibnamefont {Qi}},\ }\bibfield
   {title} {\enquote {\bibinfo {title} {Phase encoding schemes for
  measurement-device-independent quantum key distribution with basis-dependent
  flaw},}\ }\href {\doibase 10.1103/PhysRevA.85.042307} {\bibfield  {journal}
  {\bibinfo  {journal} {Phys. Rev. A}\ }\textbf {\bibinfo {volume} {85}},\
  \bibinfo {pages} {042307} (\bibinfo {year} {2012})}\BibitemShut {NoStop}%
\bibitem [{\citenamefont {Ma}\ and\ \citenamefont
  {Razavi}(2012)}]{Ma2012Alternative}%
  \BibitemOpen
  \bibfield  {author} {\bibinfo {author} {\bibfnamefont {Xiongfeng}\
  \bibnamefont {Ma}}\ and\ \bibinfo {author} {\bibfnamefont {Mohsen}\
  \bibnamefont {Razavi}},\ }\bibfield  {title} {\enquote {\bibinfo {title}
  {Alternative schemes for measurement-device-independent quantum key
  distribution},}\ }\href {\doibase 10.1103/PhysRevA.86.062319} {\bibfield
  {journal} {\bibinfo  {journal} {Phys. Rev. A}\ }\textbf {\bibinfo {volume}
  {86}},\ \bibinfo {pages} {062319} (\bibinfo {year} {2012})}\BibitemShut
  {NoStop}%
\bibitem [{\citenamefont {Ma}(2008)}]{Ma2008PhD}%
  \BibitemOpen
  \bibfield  {author} {\bibinfo {author} {\bibfnamefont {Xiongfeng}\
  \bibnamefont {Ma}},\ }\emph {\bibinfo {title} {Quantum Cryptography: From
  Theory to Practice}},\ \href {https://arxiv.org/abs/0808.1385} {Ph.D.
  thesis},\ \bibinfo  {school} {University of Toronto} (\bibinfo {year}
  {2008})\BibitemShut {NoStop}%
\bibitem [{\citenamefont {Bennett}(1992)}]{Bennett1992Quantum}%
  \BibitemOpen
  \bibfield  {author} {\bibinfo {author} {\bibfnamefont {Charles~H.}\
  \bibnamefont {Bennett}},\ }\bibfield  {title} {\enquote {\bibinfo {title}
  {Quantum cryptography using any two nonorthogonal states},}\ }\href {\doibase
  10.1103/PhysRevLett.68.3121} {\bibfield  {journal} {\bibinfo  {journal}
  {Phys. Rev. Lett.}\ }\textbf {\bibinfo {volume} {68}},\ \bibinfo {pages}
  {3121--3124} (\bibinfo {year} {1992})}\BibitemShut {NoStop}%
\bibitem [{\citenamefont {Guan}\ \emph {et~al.}(2015)\citenamefont {Guan},
  \citenamefont {Cao}, \citenamefont {Liu}, \citenamefont {Shen-Tu},
  \citenamefont {Pelc}, \citenamefont {Fejer}, \citenamefont {Peng},
  \citenamefont {Ma}, \citenamefont {Zhang},\ and\ \citenamefont
  {Pan}}]{Guan2015RRPDS}%
  \BibitemOpen
  \bibfield  {author} {\bibinfo {author} {\bibfnamefont {Jian-Yu}\ \bibnamefont
  {Guan}}, \bibinfo {author} {\bibfnamefont {Zhu}\ \bibnamefont {Cao}},
  \bibinfo {author} {\bibfnamefont {Yang}\ \bibnamefont {Liu}}, \bibinfo
  {author} {\bibfnamefont {Guo-Liang}\ \bibnamefont {Shen-Tu}}, \bibinfo
  {author} {\bibfnamefont {Jason~S.}\ \bibnamefont {Pelc}}, \bibinfo {author}
  {\bibfnamefont {M.~M.}\ \bibnamefont {Fejer}}, \bibinfo {author}
  {\bibfnamefont {Cheng-Zhi}\ \bibnamefont {Peng}}, \bibinfo {author}
  {\bibfnamefont {Xiongfeng}\ \bibnamefont {Ma}}, \bibinfo {author}
  {\bibfnamefont {Qiang}\ \bibnamefont {Zhang}}, \ and\ \bibinfo {author}
  {\bibfnamefont {Jian-Wei}\ \bibnamefont {Pan}},\ }\bibfield  {title}
  {\enquote {\bibinfo {title} {Experimental passive round-robin differential
  phase-shift quantum key distribution},}\ }\href {\doibase
  10.1103/PhysRevLett.114.180502} {\bibfield  {journal} {\bibinfo  {journal}
  {Phys. Rev. Lett.}\ }\textbf {\bibinfo {volume} {114}},\ \bibinfo {pages}
  {180502} (\bibinfo {year} {2015})}\BibitemShut {NoStop}%
\bibitem [{\citenamefont {Ferenczi}(2013)}]{Ferenczi2013}%
  \BibitemOpen
  \bibfield  {author} {\bibinfo {author} {\bibfnamefont {Agnes}\ \bibnamefont
  {Ferenczi}},\ }\emph {\bibinfo {title} {Security proof methods for quantum
  key distribution protocols}},\ \href {http://hdl.handle.net/10012/7468}
  {Ph.D. thesis} (\bibinfo {year} {2013})\BibitemShut {NoStop}%
\bibitem [{\citenamefont {Gottesman}\ \emph {et~al.}(2004)\citenamefont
  {Gottesman}, \citenamefont {Lo}, \citenamefont {L\"{u}tkenhaus},\ and\
  \citenamefont {Preskill}}]{gottesman04}%
  \BibitemOpen
  \bibfield  {author} {\bibinfo {author} {\bibfnamefont {Daniel}\ \bibnamefont
  {Gottesman}}, \bibinfo {author} {\bibfnamefont {Hoi-Kwong}\ \bibnamefont
  {Lo}}, \bibinfo {author} {\bibfnamefont {Norbert}\ \bibnamefont
  {L\"{u}tkenhaus}}, \ and\ \bibinfo {author} {\bibfnamefont {John}\
  \bibnamefont {Preskill}},\ }\bibfield  {title} {\enquote {\bibinfo {title}
  {Security of quantum key distribution with imperfect devices},}\ }\href
  {http://dl.acm.org/citation.cfm?id=2011586.2011587} {\bibfield  {journal}
  {\bibinfo  {journal} {Quantum Info. Comput.}\ }\textbf {\bibinfo {volume}
  {4}},\ \bibinfo {pages} {325--360} (\bibinfo {year} {2004})}\BibitemShut
  {NoStop}%
\bibitem [{\citenamefont {Hwang}(2003)}]{Hwang2003Decoy}%
  \BibitemOpen
  \bibfield  {author} {\bibinfo {author} {\bibfnamefont {Won-Young}\
  \bibnamefont {Hwang}},\ }\bibfield  {title} {\enquote {\bibinfo {title}
  {Quantum key distribution with high loss: Toward global secure
  communication},}\ }\href {\doibase 10.1103/PhysRevLett.91.057901} {\bibfield
  {journal} {\bibinfo  {journal} {Phys. Rev. Lett.}\ }\textbf {\bibinfo
  {volume} {91}},\ \bibinfo {pages} {057901} (\bibinfo {year}
  {2003})}\BibitemShut {NoStop}%
\bibitem [{\citenamefont {Lo}\ \emph {et~al.}(2005)\citenamefont {Lo},
  \citenamefont {Ma},\ and\ \citenamefont {Chen}}]{Lo2005Decoy}%
  \BibitemOpen
  \bibfield  {author} {\bibinfo {author} {\bibfnamefont {Hoi-Kwong}\
  \bibnamefont {Lo}}, \bibinfo {author} {\bibfnamefont {Xiongfeng}\
  \bibnamefont {Ma}}, \ and\ \bibinfo {author} {\bibfnamefont {Kai}\
  \bibnamefont {Chen}},\ }\bibfield  {title} {\enquote {\bibinfo {title} {Decoy
  state quantum key distribution},}\ }\href {\doibase
  10.1103/PhysRevLett.94.230504} {\bibfield  {journal} {\bibinfo  {journal}
  {Phys. Rev. Lett.}\ }\textbf {\bibinfo {volume} {94}},\ \bibinfo {pages}
  {230504} (\bibinfo {year} {2005})}\BibitemShut {NoStop}%
\bibitem [{\citenamefont {Wang}(2005)}]{Wang2005Decoy}%
  \BibitemOpen
  \bibfield  {author} {\bibinfo {author} {\bibfnamefont {Xiang-Bin}\
  \bibnamefont {Wang}},\ }\bibfield  {title} {\enquote {\bibinfo {title}
  {Beating the photon-number-splitting attack in practical quantum
  cryptography},}\ }\href {\doibase 10.1103/PhysRevLett.94.230503} {\bibfield
  {journal} {\bibinfo  {journal} {Phys. Rev. Lett.}\ }\textbf {\bibinfo
  {volume} {94}},\ \bibinfo {pages} {230503} (\bibinfo {year}
  {2005})}\BibitemShut {NoStop}%
\bibitem [{\citenamefont {Lim}\ \emph {et~al.}(2014)\citenamefont {Lim},
  \citenamefont {Korzh}, \citenamefont {Martin}, \citenamefont {Bussieres},
  \citenamefont {Thew},\ and\ \citenamefont {Zbinden}}]{lim2014detector}%
  \BibitemOpen
  \bibfield  {author} {\bibinfo {author} {\bibfnamefont {Charles Ci~Wen}\
  \bibnamefont {Lim}}, \bibinfo {author} {\bibfnamefont {Boris}\ \bibnamefont
  {Korzh}}, \bibinfo {author} {\bibfnamefont {Anthony}\ \bibnamefont {Martin}},
  \bibinfo {author} {\bibfnamefont {F{\'e}lix}\ \bibnamefont {Bussieres}},
  \bibinfo {author} {\bibfnamefont {Rob}\ \bibnamefont {Thew}}, \ and\ \bibinfo
  {author} {\bibfnamefont {Hugo}\ \bibnamefont {Zbinden}},\ }\bibfield  {title}
  {\enquote {\bibinfo {title} {Detector-device-independent quantum key
  distribution},}\ }\href {https://doi.org/10.1063/1.4903350} {\bibfield
  {journal} {\bibinfo  {journal} {Appl. Phys. Lett.}\ }\textbf {\bibinfo
  {volume} {105}},\ \bibinfo {pages} {221112} (\bibinfo {year}
  {2014})}\BibitemShut {NoStop}%
\bibitem [{\citenamefont {Gonz{\'a}lez}\ \emph {et~al.}(2015)\citenamefont
  {Gonz{\'a}lez}, \citenamefont {Reb{\'o}n}, \citenamefont {da~Silva},
  \citenamefont {Figueroa}, \citenamefont {Saavedra}, \citenamefont {Curty},
  \citenamefont {Lima}, \citenamefont {Xavier},\ and\ \citenamefont
  {Nogueira}}]{gonzalez2015quantum}%
  \BibitemOpen
  \bibfield  {author} {\bibinfo {author} {\bibfnamefont {P}~\bibnamefont
  {Gonz{\'a}lez}}, \bibinfo {author} {\bibfnamefont {L}~\bibnamefont
  {Reb{\'o}n}}, \bibinfo {author} {\bibfnamefont {T~Ferreira}\ \bibnamefont
  {da~Silva}}, \bibinfo {author} {\bibfnamefont {M}~\bibnamefont {Figueroa}},
  \bibinfo {author} {\bibfnamefont {C}~\bibnamefont {Saavedra}}, \bibinfo
  {author} {\bibfnamefont {M}~\bibnamefont {Curty}}, \bibinfo {author}
  {\bibfnamefont {G}~\bibnamefont {Lima}}, \bibinfo {author} {\bibfnamefont
  {GB}~\bibnamefont {Xavier}}, \ and\ \bibinfo {author} {\bibfnamefont {WAT}\
  \bibnamefont {Nogueira}},\ }\bibfield  {title} {\enquote {\bibinfo {title}
  {Quantum key distribution with untrusted detectors},}\ }\href
  {https://doi.org/10.1103/PhysRevA.92.022337} {\bibfield  {journal} {\bibinfo
  {journal} {Phys. Rev. A}\ }\textbf {\bibinfo {volume} {92}},\ \bibinfo
  {pages} {022337} (\bibinfo {year} {2015})}\BibitemShut {NoStop}%
\bibitem [{\citenamefont {Santarelli}\ \emph {et~al.}(1994)\citenamefont
  {Santarelli}, \citenamefont {Clairon}, \citenamefont {Lea},\ and\
  \citenamefont {Tino}}]{santarelli1994heterodyne}%
  \BibitemOpen
  \bibfield  {author} {\bibinfo {author} {\bibfnamefont {G}~\bibnamefont
  {Santarelli}}, \bibinfo {author} {\bibfnamefont {A}~\bibnamefont {Clairon}},
  \bibinfo {author} {\bibfnamefont {SN}~\bibnamefont {Lea}}, \ and\ \bibinfo
  {author} {\bibfnamefont {GM}~\bibnamefont {Tino}},\ }\bibfield  {title}
  {\enquote {\bibinfo {title} {Heterodyne optical phase-locking of
  extended-cavity semiconductor lasers at 9 ghz},}\ }\href
  {https://doi.org/10.1016/0030-4018(94)90567-3} {\bibfield  {journal}
  {\bibinfo  {journal} {Opt. Commun.}\ }\textbf {\bibinfo {volume} {104}},\
  \bibinfo {pages} {339--344} (\bibinfo {year} {1994})}\BibitemShut {NoStop}%
\bibitem [{\citenamefont {Qi}\ \emph {et~al.}(2007)\citenamefont {Qi},
  \citenamefont {Huang}, \citenamefont {Qian},\ and\ \citenamefont
  {Lo}}]{Qi2007experimental}%
  \BibitemOpen
  \bibfield  {author} {\bibinfo {author} {\bibfnamefont {Bing}\ \bibnamefont
  {Qi}}, \bibinfo {author} {\bibfnamefont {Lei-Lei}\ \bibnamefont {Huang}},
  \bibinfo {author} {\bibfnamefont {Li}~\bibnamefont {Qian}}, \ and\ \bibinfo
  {author} {\bibfnamefont {Hoi-Kwong}\ \bibnamefont {Lo}},\ }\bibfield  {title}
  {\enquote {\bibinfo {title} {Experimental study on the gaussian-modulated
  coherent-state quantum key distribution over standard telecommunication
  fibers},}\ }\href {\doibase 10.1103/PhysRevA.76.052323} {\bibfield  {journal}
  {\bibinfo  {journal} {Phys. Rev. A}\ }\textbf {\bibinfo {volume} {76}},\
  \bibinfo {pages} {052323} (\bibinfo {year} {2007})}\BibitemShut {NoStop}%
\bibitem [{\citenamefont {Qi}\ \emph {et~al.}(2015)\citenamefont {Qi},
  \citenamefont {Lougovski}, \citenamefont {Pooser}, \citenamefont {Grice},\
  and\ \citenamefont {Bobrek}}]{Qi2015generating}%
  \BibitemOpen
  \bibfield  {author} {\bibinfo {author} {\bibfnamefont {Bing}\ \bibnamefont
  {Qi}}, \bibinfo {author} {\bibfnamefont {Pavel}\ \bibnamefont {Lougovski}},
  \bibinfo {author} {\bibfnamefont {Raphael}\ \bibnamefont {Pooser}}, \bibinfo
  {author} {\bibfnamefont {Warren}\ \bibnamefont {Grice}}, \ and\ \bibinfo
  {author} {\bibfnamefont {Miljko}\ \bibnamefont {Bobrek}},\ }\bibfield
  {title} {\enquote {\bibinfo {title} {Generating the local oscillator
  ``locally'' in continuous-variable quantum key distribution based on coherent
  detection},}\ }\href {\doibase 10.1103/PhysRevX.5.041009} {\bibfield
  {journal} {\bibinfo  {journal} {Phys. Rev. X}\ }\textbf {\bibinfo {volume}
  {5}},\ \bibinfo {pages} {041009} (\bibinfo {year} {2015})}\BibitemShut
  {NoStop}%
\bibitem [{\citenamefont {Soh}\ \emph {et~al.}(2015)\citenamefont {Soh},
  \citenamefont {Brif}, \citenamefont {Coles}, \citenamefont {L\"utkenhaus},
  \citenamefont {Camacho}, \citenamefont {Urayama},\ and\ \citenamefont
  {Sarovar}}]{Soh2015self}%
  \BibitemOpen
  \bibfield  {author} {\bibinfo {author} {\bibfnamefont {Daniel B.~S.}\
  \bibnamefont {Soh}}, \bibinfo {author} {\bibfnamefont {Constantin}\
  \bibnamefont {Brif}}, \bibinfo {author} {\bibfnamefont {Patrick~J.}\
  \bibnamefont {Coles}}, \bibinfo {author} {\bibfnamefont {Norbert}\
  \bibnamefont {L\"utkenhaus}}, \bibinfo {author} {\bibfnamefont {Ryan~M.}\
  \bibnamefont {Camacho}}, \bibinfo {author} {\bibfnamefont {Junji}\
  \bibnamefont {Urayama}}, \ and\ \bibinfo {author} {\bibfnamefont {Mohan}\
  \bibnamefont {Sarovar}},\ }\bibfield  {title} {\enquote {\bibinfo {title}
  {Self-referenced continuous-variable quantum key distribution protocol},}\
  }\href {\doibase 10.1103/PhysRevX.5.041010} {\bibfield  {journal} {\bibinfo
  {journal} {Phys. Rev. X}\ }\textbf {\bibinfo {volume} {5}},\ \bibinfo {pages}
  {041010} (\bibinfo {year} {2015})}\BibitemShut {NoStop}%
\bibitem [{\citenamefont {Cao}\ \emph {et~al.}(2015)\citenamefont {Cao},
  \citenamefont {Zhang}, \citenamefont {Lo},\ and\ \citenamefont
  {Ma}}]{zhu2015discrete}%
  \BibitemOpen
  \bibfield  {author} {\bibinfo {author} {\bibfnamefont {Zhu}\ \bibnamefont
  {Cao}}, \bibinfo {author} {\bibfnamefont {Zhen}\ \bibnamefont {Zhang}},
  \bibinfo {author} {\bibfnamefont {Hoi-Kwong}\ \bibnamefont {Lo}}, \ and\
  \bibinfo {author} {\bibfnamefont {Xiongfeng}\ \bibnamefont {Ma}},\ }\bibfield
   {title} {\enquote {\bibinfo {title} {Discrete-phase-randomized coherent
  state source and its application in quantum key distribution},}\ }\href
  {http://stacks.iop.org/1367-2630/17/i=5/a=053014} {\bibfield  {journal}
  {\bibinfo  {journal} {New J. Phys.}\ }\textbf {\bibinfo {volume} {17}},\
  \bibinfo {pages} {053014} (\bibinfo {year} {2015})}\BibitemShut {NoStop}%
\bibitem [{\citenamefont {Droste}\ \emph {et~al.}(2013)\citenamefont {Droste},
  \citenamefont {Ozimek}, \citenamefont {Udem}, \citenamefont {Predehl},
  \citenamefont {H\"ansch}, \citenamefont {Schnatz}, \citenamefont {Grosche},\
  and\ \citenamefont {Holzwarth}}]{Droste2013optical}%
  \BibitemOpen
  \bibfield  {author} {\bibinfo {author} {\bibfnamefont {S.}~\bibnamefont
  {Droste}}, \bibinfo {author} {\bibfnamefont {F.}~\bibnamefont {Ozimek}},
  \bibinfo {author} {\bibfnamefont {Th.}\ \bibnamefont {Udem}}, \bibinfo
  {author} {\bibfnamefont {K.}~\bibnamefont {Predehl}}, \bibinfo {author}
  {\bibfnamefont {T.~W.}\ \bibnamefont {H\"ansch}}, \bibinfo {author}
  {\bibfnamefont {H.}~\bibnamefont {Schnatz}}, \bibinfo {author} {\bibfnamefont
  {G.}~\bibnamefont {Grosche}}, \ and\ \bibinfo {author} {\bibfnamefont
  {R.}~\bibnamefont {Holzwarth}},\ }\bibfield  {title} {\enquote {\bibinfo
  {title} {Optical-frequency transfer over a single-span 1840 km fiber link},}\
  }\href {\doibase 10.1103/PhysRevLett.111.110801} {\bibfield  {journal}
  {\bibinfo  {journal} {Phys. Rev. Lett.}\ }\textbf {\bibinfo {volume} {111}},\
  \bibinfo {pages} {110801} (\bibinfo {year} {2013})}\BibitemShut {NoStop}%
\bibitem [{\citenamefont {Carvacho}\ \emph {et~al.}(2015)\citenamefont
  {Carvacho}, \citenamefont {Cari\~ne}, \citenamefont {Saavedra}, \citenamefont
  {Cuevas}, \citenamefont {Fuenzalida}, \citenamefont {Toledo}, \citenamefont
  {Figueroa}, \citenamefont {Cabello}, \citenamefont {Larsson}, \citenamefont
  {Mataloni}, \citenamefont {Lima},\ and\ \citenamefont
  {Xavier}}]{Carvacho2015postselection}%
  \BibitemOpen
  \bibfield  {author} {\bibinfo {author} {\bibfnamefont {Gonzalo}\ \bibnamefont
  {Carvacho}}, \bibinfo {author} {\bibfnamefont {Jaime}\ \bibnamefont
  {Cari\~ne}}, \bibinfo {author} {\bibfnamefont {Gabriel}\ \bibnamefont
  {Saavedra}}, \bibinfo {author} {\bibfnamefont {\'Alvaro}\ \bibnamefont
  {Cuevas}}, \bibinfo {author} {\bibfnamefont {Jorge}\ \bibnamefont
  {Fuenzalida}}, \bibinfo {author} {\bibfnamefont {Felipe}\ \bibnamefont
  {Toledo}}, \bibinfo {author} {\bibfnamefont {Miguel}\ \bibnamefont
  {Figueroa}}, \bibinfo {author} {\bibfnamefont {Ad\'an}\ \bibnamefont
  {Cabello}}, \bibinfo {author} {\bibfnamefont {Jan-\AA{}ke}\ \bibnamefont
  {Larsson}}, \bibinfo {author} {\bibfnamefont {Paolo}\ \bibnamefont
  {Mataloni}}, \bibinfo {author} {\bibfnamefont {Gustavo}\ \bibnamefont
  {Lima}}, \ and\ \bibinfo {author} {\bibfnamefont {Guilherme~B.}\ \bibnamefont
  {Xavier}},\ }\bibfield  {title} {\enquote {\bibinfo {title}
  {Postselection-loophole-free bell test over an installed optical fiber
  network},}\ }\href {\doibase 10.1103/PhysRevLett.115.030503} {\bibfield
  {journal} {\bibinfo  {journal} {Phys. Rev. Lett.}\ }\textbf {\bibinfo
  {volume} {115}},\ \bibinfo {pages} {030503} (\bibinfo {year}
  {2015})}\BibitemShut {NoStop}%
\bibitem [{\citenamefont {Lipka}\ \emph {et~al.}(2017)\citenamefont {Lipka},
  \citenamefont {Parniak},\ and\ \citenamefont
  {Wasilewski}}]{Lipka2017optical}%
  \BibitemOpen
  \bibfield  {author} {\bibinfo {author} {\bibfnamefont {Micha{\l}}\
  \bibnamefont {Lipka}}, \bibinfo {author} {\bibfnamefont {Micha{\l}}\
  \bibnamefont {Parniak}}, \ and\ \bibinfo {author} {\bibfnamefont {Wojciech}\
  \bibnamefont {Wasilewski}},\ }\bibfield  {title} {\enquote {\bibinfo {title}
  {Optical frequency locked loop for long-term stabilization of broad-line dfb
  laser frequency difference},}\ }\href {\doibase 10.1007/s00340-017-6808-6}
  {\bibfield  {journal} {\bibinfo  {journal} {Appl. Phys. B}\ }\textbf
  {\bibinfo {volume} {123}},\ \bibinfo {pages} {238} (\bibinfo {year}
  {2017})}\BibitemShut {NoStop}%
\bibitem [{\citenamefont {Tang}\ \emph {et~al.}(2014)\citenamefont {Tang},
  \citenamefont {Yin}, \citenamefont {Chen}, \citenamefont {Liu}, \citenamefont
  {Zhang}, \citenamefont {Jiang}, \citenamefont {Zhang}, \citenamefont {Wang},
  \citenamefont {You}, \citenamefont {Guan}, \citenamefont {Yang},
  \citenamefont {Wang}, \citenamefont {Liang}, \citenamefont {Zhang},
  \citenamefont {Zhou}, \citenamefont {Ma}, \citenamefont {Chen}, \citenamefont
  {Zhang},\ and\ \citenamefont {Pan}}]{Tang2014MDI200}%
  \BibitemOpen
  \bibfield  {author} {\bibinfo {author} {\bibfnamefont {Yan-Lin}\ \bibnamefont
  {Tang}}, \bibinfo {author} {\bibfnamefont {Hua-Lei}\ \bibnamefont {Yin}},
  \bibinfo {author} {\bibfnamefont {Si-Jing}\ \bibnamefont {Chen}}, \bibinfo
  {author} {\bibfnamefont {Yang}\ \bibnamefont {Liu}}, \bibinfo {author}
  {\bibfnamefont {Wei-Jun}\ \bibnamefont {Zhang}}, \bibinfo {author}
  {\bibfnamefont {Xiao}\ \bibnamefont {Jiang}}, \bibinfo {author}
  {\bibfnamefont {Lu}~\bibnamefont {Zhang}}, \bibinfo {author} {\bibfnamefont
  {Jian}\ \bibnamefont {Wang}}, \bibinfo {author} {\bibfnamefont {Li-Xing}\
  \bibnamefont {You}}, \bibinfo {author} {\bibfnamefont {Jian-Yu}\ \bibnamefont
  {Guan}}, \bibinfo {author} {\bibfnamefont {Dong-Xu}\ \bibnamefont {Yang}},
  \bibinfo {author} {\bibfnamefont {Zhen}\ \bibnamefont {Wang}}, \bibinfo
  {author} {\bibfnamefont {Hao}\ \bibnamefont {Liang}}, \bibinfo {author}
  {\bibfnamefont {Zhen}\ \bibnamefont {Zhang}}, \bibinfo {author}
  {\bibfnamefont {Nan}\ \bibnamefont {Zhou}}, \bibinfo {author} {\bibfnamefont
  {Xiongfeng}\ \bibnamefont {Ma}}, \bibinfo {author} {\bibfnamefont {Teng-Yun}\
  \bibnamefont {Chen}}, \bibinfo {author} {\bibfnamefont {Qiang}\ \bibnamefont
  {Zhang}}, \ and\ \bibinfo {author} {\bibfnamefont {Jian-Wei}\ \bibnamefont
  {Pan}},\ }\bibfield  {title} {\enquote {\bibinfo {title}
  {Measurement-device-independent quantum key distribution over 200 km},}\
  }\href {\doibase 10.1103/PhysRevLett.113.190501} {\bibfield  {journal}
  {\bibinfo  {journal} {Phys. Rev. Lett.}\ }\textbf {\bibinfo {volume} {113}},\
  \bibinfo {pages} {190501} (\bibinfo {year} {2014})}\BibitemShut {NoStop}%
\bibitem [{\citenamefont {Pirandola}(2016)}]{Pirandola2016Capacities}%
  \BibitemOpen
  \bibfield  {author} {\bibinfo {author} {\bibfnamefont {Stefano}\ \bibnamefont
  {Pirandola}},\ }\bibfield  {title} {\enquote {\bibinfo {title} {Capacities of
  repeater-assisted quantum communications},}\ }\href@noop {} {\bibfield
  {journal} {\bibinfo  {journal} {arXiv:1601.00966}\ } (\bibinfo {year}
  {2016})}\BibitemShut {NoStop}%
\bibitem [{\citenamefont {Bennett}\ \emph {et~al.}(1996)\citenamefont
  {Bennett}, \citenamefont {DiVincenzo}, \citenamefont {Smolin},\ and\
  \citenamefont {Wootters}}]{Bennett1996BDSW}%
  \BibitemOpen
  \bibfield  {author} {\bibinfo {author} {\bibfnamefont {Charles~H.}\
  \bibnamefont {Bennett}}, \bibinfo {author} {\bibfnamefont {David~P.}\
  \bibnamefont {DiVincenzo}}, \bibinfo {author} {\bibfnamefont {John~A.}\
  \bibnamefont {Smolin}}, \ and\ \bibinfo {author} {\bibfnamefont {William~K.}\
  \bibnamefont {Wootters}},\ }\bibfield  {title} {\enquote {\bibinfo {title}
  {Mixed-state entanglement and quantum error correction},}\ }\href {\doibase
  10.1103/PhysRevA.54.3824} {\bibfield  {journal} {\bibinfo  {journal} {Phys.
  Rev. A}\ }\textbf {\bibinfo {volume} {54}},\ \bibinfo {pages} {3824--3851}
  (\bibinfo {year} {1996})}\BibitemShut {NoStop}%
\bibitem [{\citenamefont {Ben-Or}\ \emph {et~al.}(2005)\citenamefont {Ben-Or},
  \citenamefont {Leung}, \citenamefont {Mayers},\ and\ \citenamefont
  {Oppenheim}}]{Ben2005composable}%
  \BibitemOpen
  \bibfield  {author} {\bibinfo {author} {\bibfnamefont {Michael}\ \bibnamefont
  {Ben-Or}}, \bibinfo {author} {\bibfnamefont {Debbie~W.}\ \bibnamefont
  {Leung}}, \bibinfo {author} {\bibfnamefont {Dominic}\ \bibnamefont {Mayers}},
  \ and\ \bibinfo {author} {\bibfnamefont {Jonathan}\ \bibnamefont
  {Oppenheim}},\ }\bibfield  {title} {\enquote {\bibinfo {title} {The universal
  composable security of quantum key distribution},}\ }in\ \href
  {https://doi.org/10.1007/978-3-540-30576-7_21} {\emph {\bibinfo {booktitle}
  {International Conference on Theory of Cryptography}}}\ (\bibinfo
  {publisher} {Springer Berlin Heidelberg},\ \bibinfo {year} {2005})\ pp.\
  \bibinfo {pages} {386--406}\BibitemShut {NoStop}%
\bibitem [{\citenamefont {Renner}\ and\ \citenamefont
  {K{\"o}nig}(2005)}]{Renner2005universally}%
  \BibitemOpen
  \bibfield  {author} {\bibinfo {author} {\bibfnamefont {Renato}\ \bibnamefont
  {Renner}}\ and\ \bibinfo {author} {\bibfnamefont {Robert}\ \bibnamefont
  {K{\"o}nig}},\ }\enquote {\bibinfo {title} {Universally composable privacy
  amplification against quantum adversaries},}\ in\ \href {\doibase
  10.1007/978-3-540-30576-7_22} {\emph {\bibinfo {booktitle} {Theory of
  Cryptography: Second Theory of Cryptography Conference, TCC 2005, Cambridge,
  MA, USA, February 10-12, 2005. Proceedings}}},\ \bibinfo {editor} {edited by\
  \bibinfo {editor} {\bibfnamefont {Joe}\ \bibnamefont {Kilian}}}\ (\bibinfo
  {publisher} {Springer Berlin Heidelberg},\ \bibinfo {address} {Berlin,
  Heidelberg},\ \bibinfo {year} {2005})\ pp.\ \bibinfo {pages}
  {407--425}\BibitemShut {NoStop}%
\bibitem [{\citenamefont {Calderbank}\ and\ \citenamefont
  {Shor}(1996)}]{CSS1996qec}%
  \BibitemOpen
  \bibfield  {author} {\bibinfo {author} {\bibfnamefont {A.~R.}\ \bibnamefont
  {Calderbank}}\ and\ \bibinfo {author} {\bibfnamefont {Peter~W.}\ \bibnamefont
  {Shor}},\ }\bibfield  {title} {\enquote {\bibinfo {title} {Good quantum
  error-correcting codes exist},}\ }\href {\doibase 10.1103/PhysRevA.54.1098}
  {\bibfield  {journal} {\bibinfo  {journal} {Phys. Rev. A}\ }\textbf {\bibinfo
  {volume} {54}},\ \bibinfo {pages} {1098--1105} (\bibinfo {year}
  {1996})}\BibitemShut {NoStop}%
\bibitem [{\citenamefont {Min{\'a}{\v{r}}}\ \emph {et~al.}(2008)\citenamefont
  {Min{\'a}{\v{r}}}, \citenamefont {De~Riedmatten}, \citenamefont {Simon},
  \citenamefont {Zbinden},\ and\ \citenamefont {Gisin}}]{minavr2008phase}%
  \BibitemOpen
  \bibfield  {author} {\bibinfo {author} {\bibfnamefont {Ji{\v{r}}{\'\i}}\
  \bibnamefont {Min{\'a}{\v{r}}}}, \bibinfo {author} {\bibfnamefont {Hugues}\
  \bibnamefont {De~Riedmatten}}, \bibinfo {author} {\bibfnamefont {Christoph}\
  \bibnamefont {Simon}}, \bibinfo {author} {\bibfnamefont {Hugo}\ \bibnamefont
  {Zbinden}}, \ and\ \bibinfo {author} {\bibfnamefont {Nicolas}\ \bibnamefont
  {Gisin}},\ }\bibfield  {title} {\enquote {\bibinfo {title} {Phase-noise
  measurements in long-fiber interferometers for quantum-repeater
  applications},}\ }\href {https://doi.org/10.1103/PhysRevA.77.052325}
  {\bibfield  {journal} {\bibinfo  {journal} {Phys. Rev. A}\ }\textbf {\bibinfo
  {volume} {77}},\ \bibinfo {pages} {052325} (\bibinfo {year}
  {2008})}\BibitemShut {NoStop}%
\bibitem [{\citenamefont {Ma}\ \emph {et~al.}(2005)\citenamefont {Ma},
  \citenamefont {Qi}, \citenamefont {Zhao},\ and\ \citenamefont
  {Lo}}]{Ma2005Practical}%
  \BibitemOpen
  \bibfield  {author} {\bibinfo {author} {\bibfnamefont {Xiongfeng}\
  \bibnamefont {Ma}}, \bibinfo {author} {\bibfnamefont {Bing}\ \bibnamefont
  {Qi}}, \bibinfo {author} {\bibfnamefont {Yi}~\bibnamefont {Zhao}}, \ and\
  \bibinfo {author} {\bibfnamefont {Hoi-Kwong}\ \bibnamefont {Lo}},\ }\bibfield
   {title} {\enquote {\bibinfo {title} {Practical decoy state for quantum key
  distribution},}\ }\href {\doibase 10.1103/PhysRevA.72.012326} {\bibfield
  {journal} {\bibinfo  {journal} {Phys. Rev. A}\ }\textbf {\bibinfo {volume}
  {72}},\ \bibinfo {pages} {012326} (\bibinfo {year} {2005})}\BibitemShut
  {NoStop}%
\bibitem [{\citenamefont {Pirandola}\ \emph {et~al.}(2009)\citenamefont
  {Pirandola}, \citenamefont {Garc{\'\i}a-Patr{\'o}n}, \citenamefont
  {Braunstein},\ and\ \citenamefont {Lloyd}}]{pirandola2009direct}%
  \BibitemOpen
  \bibfield  {author} {\bibinfo {author} {\bibfnamefont {Stefano}\ \bibnamefont
  {Pirandola}}, \bibinfo {author} {\bibfnamefont {Raul}\ \bibnamefont
  {Garc{\'\i}a-Patr{\'o}n}}, \bibinfo {author} {\bibfnamefont {Samuel~L}\
  \bibnamefont {Braunstein}}, \ and\ \bibinfo {author} {\bibfnamefont {Seth}\
  \bibnamefont {Lloyd}},\ }\bibfield  {title} {\enquote {\bibinfo {title}
  {Direct and reverse secret-key capacities of a quantum channel},}\ }\href
  {https://doi.org/10.1103/PhysRevLett.102.050503} {\bibfield  {journal}
  {\bibinfo  {journal} {Phys. Rev. Lett.}\ }\textbf {\bibinfo {volume} {102}},\
  \bibinfo {pages} {050503} (\bibinfo {year} {2009})}\BibitemShut {NoStop}%
\bibitem [{\citenamefont {Jaeger}\ and\ \citenamefont
  {Shimony}(1995)}]{jaeger1995optimal}%
  \BibitemOpen
  \bibfield  {author} {\bibinfo {author} {\bibfnamefont {Gregg}\ \bibnamefont
  {Jaeger}}\ and\ \bibinfo {author} {\bibfnamefont {Abner}\ \bibnamefont
  {Shimony}},\ }\bibfield  {title} {\enquote {\bibinfo {title} {Optimal
  distinction between two non-orthogonal quantum states},}\ }\href {\doibase
  https://doi.org/10.1016/0375-9601(94)00919-G} {\bibfield  {journal} {\bibinfo
   {journal} {Physics Letters A}\ }\textbf {\bibinfo {volume} {197}},\ \bibinfo
  {pages} {83 -- 87} (\bibinfo {year} {1995})}\BibitemShut {NoStop}%
\bibitem [{\citenamefont {Wang}\ \emph
  {et~al.}(2018{\natexlab{a}})\citenamefont {Wang}, \citenamefont {Hu},\ and\
  \citenamefont {Yu}}]{Wang2018Effective}%
  \BibitemOpen
  \bibfield  {author} {\bibinfo {author} {\bibfnamefont {Xiang-Bin}\
  \bibnamefont {Wang}}, \bibinfo {author} {\bibfnamefont {Xiao-Long}\
  \bibnamefont {Hu}}, \ and\ \bibinfo {author} {\bibfnamefont {Zong-Wen}\
  \bibnamefont {Yu}},\ }\bibfield  {title} {\enquote {\bibinfo {title}
  {Effective eavesdropping to twin-field quantum key distribution},}\
  }\href@noop {} {\  (\bibinfo {year} {2018}{\natexlab{a}})},\ \bibinfo {note}
  {arXiv:1805.02272v1}\BibitemShut {NoStop}%
\bibitem [{\citenamefont {Tamaki}\ \emph {et~al.}(2018)\citenamefont {Tamaki},
  \citenamefont {Lo}, \citenamefont {Wang},\ and\ \citenamefont
  {Lucamarini}}]{tamaki2018information}%
  \BibitemOpen
  \bibfield  {author} {\bibinfo {author} {\bibfnamefont {Kiyoshi}\ \bibnamefont
  {Tamaki}}, \bibinfo {author} {\bibfnamefont {Hoi-Kwong}\ \bibnamefont {Lo}},
  \bibinfo {author} {\bibfnamefont {Wenyuan}\ \bibnamefont {Wang}}, \ and\
  \bibinfo {author} {\bibfnamefont {Marco}\ \bibnamefont {Lucamarini}},\
  }\bibfield  {title} {\enquote {\bibinfo {title} {Information theoretic
  security of quantum key distribution overcoming the repeaterless secret key
  capacity bound},}\ }\href {https://arxiv.org/abs/1805.05511} {\bibfield
  {journal} {\bibinfo  {journal} {arXiv:1805.05511}\ } (\bibinfo {year}
  {2018})}\BibitemShut {NoStop}%
\bibitem [{\citenamefont {Wang}\ \emph
  {et~al.}(2018{\natexlab{b}})\citenamefont {Wang}, \citenamefont {Yu},\ and\
  \citenamefont {Hu}}]{Wang2018Sending}%
  \BibitemOpen
  \bibfield  {author} {\bibinfo {author} {\bibfnamefont {Xiang-Bin}\
  \bibnamefont {Wang}}, \bibinfo {author} {\bibfnamefont {Zong-Wen}\
  \bibnamefont {Yu}}, \ and\ \bibinfo {author} {\bibfnamefont {Xiao-Long}\
  \bibnamefont {Hu}},\ }\bibfield  {title} {\enquote {\bibinfo {title} {Sending
  or not sending: Twin-field quantum key distrbution with large misalignment
  error},}\ }\href@noop {} {\  (\bibinfo {year} {2018}{\natexlab{b}})},\
  \bibinfo {note} {arXiv: 1805.09222}\BibitemShut {NoStop}%
\end{thebibliography}%

\end{document}